\documentclass[onecolumn,journal,12pt]{IEEEtran}

\ifCLASSINFOpdf
\else
\fi
\usepackage{amssymb, amsmath, latexsym}
\usepackage[dvips]{graphicx}
\usepackage{color}
\usepackage{comment}

\newtheorem{thm} {\textcolor{black}{Theorem}}
\newtheorem{lem}[thm] {\textcolor{black}{Lemma}}
\newtheorem{cor}[thm] {\textcolor{black}{Corollary}}
\newtheorem{prp}[thm] {\textcolor{black}{Proposition}}
\newtheorem{df}[thm]{\textcolor{black}{Definition}}
\def\QED{\mbox{\rule[0pt]{1.5ex}{1.5ex}}}

\def\endproof{\hspace*{\fill}~\QED\par\endtrivlist\unskip}

\def\calD{{\mathcal{D}}}
\def\calX{{\mathcal{X}}}
\def\calY{{\mathcal{Y}}}
\def\calP{{\mathcal{P}}}
\def\R{{\mathbb{R}}}

\def\N{{\mathbb{N}}}

\def\C{{\mathbb{C}}}
\def\ph{{\varphi}}

\def\QED{\mbox{\rule[0pt]{1.5ex}{1.5ex}}}

\def\endproof{\hspace*{\fill}~\QED\par\endtrivlist\unskip}

\def\Label#1{\label{#1}\ [\ #1\ ]\ }
\def\Label{\label}

\begin{document}
\title{Random Number Conversion and LOCC Conversion via Restricted Storage}

\author{Wataru~Kumagai,~
        Masahito~Hayashi~\IEEEmembership{Fellow, IEEE.}
\thanks{W. Kumagai is with Faculty of Engineering, Kanagawa University. e-mail: kumagai@kanagawa-u.ac.jp}
\thanks{M. Hayashi is with Nagoya University and National University of Singapore. e-mail: masahito@math.nagoya-u.ac.jp}
}

\maketitle

\begin{abstract}

We consider random number conversion (RNC) through random number storage with restricted size.
We clarify the relation between the performance of RNC and the size of storage in the framework of first- and second-order asymptotics, and derive their rate regions.
Then, we show that the results for RNC with restricted storage recover those for conventional RNC without storage in the limit of storage size.
To treat RNC via restricted storage,
we introduce a new kind of probability distributions named generalized Rayleigh-normal distributions.
Using the generalized Rayleigh-normal  distributions,
we can describe the second-order asymptotic behaviour of RNC via restricted storage in a unified manner.
As an application to quantum information theory, 
we analyze LOCC conversion via entanglement storage with restricted size.
Moreover, 
we derive the optimal LOCC compression rate under a constraint of conversion accuracy.
\end{abstract}


\IEEEpeerreviewmaketitle


\section{Introduction}
\Label{sec:Introduction}
Random number conversion (RNC) is a fundamental topic in information theory \cite{VV95}, and 
its asymptotic behavior has been well studied in the context of not only the first-order asymptotics but also the second-order asymptotics \cite{Hay08,NH13,KH13-2}.
The second-order analysis for the random number conversion 
has the following remarkable property distinct from that of other information tasks.
The second-order rates cannot be characterized by use of the normal distribution in the case of RNC
although known second-order rates
are mostly given by use of the normal distribution.
To characterize the second-order rates in the random number conversion, 
the previous paper \cite{KH13-2} introduced Rayleigh-normal distributions
as a new family of probability distributions.
This new family of distributions leads us to a new frontier of second order analysis,
which is completely different from existing analysis of the second-order rate.
In this paper, 
we focus on a realistic situation, in which 
one uses this conversion via a storage with a limited size like a hard disk. 
In this case, as the first step, initial random numbers are converted to other random numbers in a storage with a limited size, which is 
called {\it random number storage} or simply storage.
As the second step, the random numbers in the storage are converted to some desired random numbers.
When the memory size of media for the random number conversion is limited,
it is natural to consider the trade-off between the sizes of target random numbers and the storage.

In this paper, we consider this problem 
when the initial and the target random random variables are given as 
multiple copies of respective finite random variables.
That is, the initial random variables are subject to the $n$-fold independent and identical 
distribution (i.i.d.) of a distribution $P$ with finite support and 
the target random variables  are subject to the $m$-fold i.i.d. of another distribution $Q$ with finite support.
In the problem, 
since there is the degree of freedom for the required number of copies of $Q$ in the target distribution,
we have to take care of the trade-off among three factors,
the accuracy of the conversion, 
the size of the storage, and
the required number of copies of $Q$ in the output distribution.
For simplicity, we fix the accuracy of the conversion, 
and investigate the trade-off
between the size of the storage and the required number of copies of $Q$ in the output distribution.
We call this problem RNC via restricted storage.
In particular, 
when $Q=P$, 
this problem can be regarded as random number compression to the given random number storage.

One of our main purposes is to derive the maximum conversion rate when the rate of storage size is properly limited.
If the size of storage is small,
the maximum number of copies of target distribution should also be small
since the conversion has to once pass through the small storage.
Thus, the allowable size of storage closely relates with the conversion rate of RNC via restricted storage. 
In this paper, 
we particularly investigate the region of achievable rate pairs for the size of storage and the number of copies of target distribution in the first- and the second-order settings.
To clarify which rate pairs are truly important in the rate region,
we introduce the relations named ``dominate" and ``simulate" between two rate pairs,
and based on these two relations,
we define the admissibility of rate pairs.
Although admissible rate pairs are only a part of the boundary of the region,
those characterize the whole of the rate region,
and hence, are of special importance in the rate region.

\begin{figure}[h]
 \begin{center}
 \hspace*{0em}\includegraphics[width=70mm, height=60mm]{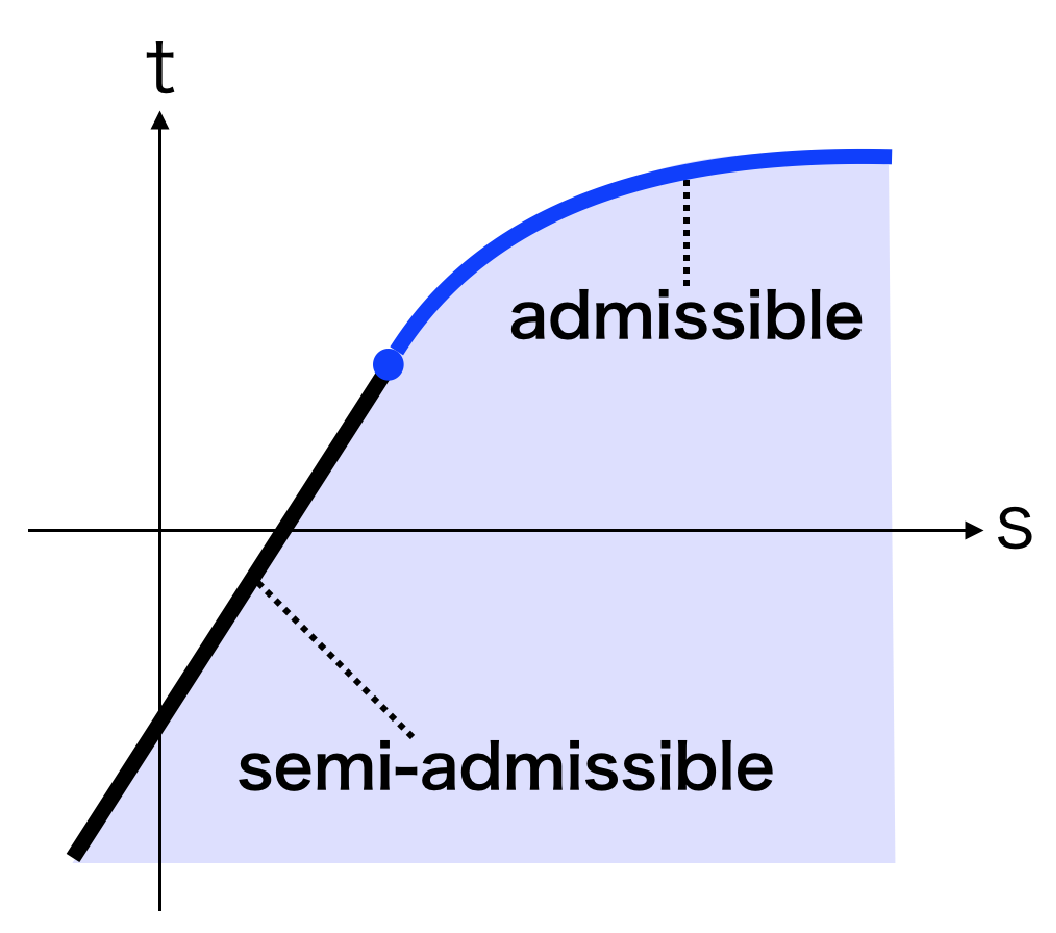}
 \end{center}
 \caption{
A graph of a rate region.
The black straight line represents the set of semi-admissible rate pairs 
and the blue curved line represents the set of admissible rate pairs.
An admissible pair is a pair dominated nor simulated by no other pair.
A semi-admissible pair is a pair dominated by no other pair.}
\Label{fig:admissible}
\end{figure}

In the set of achievable first-order rate pairs,
the admissible rate pair is shown to be unique
and all other rate pairs are not admissible.
In this sense,
the admissible rate pair may seem to be  exceptional.
However, 
the case of the admissible first-order rate pair is most important as stated below.
A first-order rate pair consists of the first-order rates of the size of restricted storage and the number of copies of the target distribution.
At the admissible rate pair,
the first-order rate of the size of restricted storage is shown to be the entropy of the source distribution.
If the first-order rate is strictly less or larger than the entropy of the source distribution,
the size of storage is too small or redundant to store the randomness of the source distribution, respectively.
In this sense,
the entropy of the source distribution is the only suitable first-order rate to store the randomness of the source distribution.
Similarly,
at the admissible rate pair,
the first-order rate of the number of copies of the target distribution is shown to be the entropy ratio of the source distribution and the target distribution.
If the first-order rate is strictly less or larger than the entropy ratio,
random numbers in the storage properly converted from the source distribution are unnecessarily redundant or too few to approximate the target distribution, respectively.
In this sense,
the entropy ratio is the only suitable first-order rate to generate the target distribution.

%


We emphasize that our optimal conversion to the storage is a uniform random number generation independtly of whether the achievable first-order rate pair is admissible or not. 
That is,
the optimal conversion scheme can be constructed as follows:
a source distribution is first approximately converted to the uniform random distribution independent of the target distribution $Q$, and then  converted from the uniform random distribution to the i.i.d. of $Q$.

Here, remember that  the second-order rates of  the random number conversion
are characterized by Rayleigh-normal distributions \cite{KH13-2}.
Since the second-order asymptotic behaviour of other typical  information tasks are often described by the standard normal distribution,
the characterization by such a non-normal distribution is a remarkable feature.
To treat the second-order asymptotics of our problem, 
we introduce a new kind of probability distributions named generalized Rayleigh-normal distributions
as an extension of Rayleigh-normal distributions.
The generalized Rayleigh-normal distributions are a family of probability distributions with two parameters and include the Rayleigh-normal distributions in \cite{KH13-2} as the limit case.
Using the generalized Rayleigh-normal distributions,
we can characterize the second-order rate region of RNC with restricted storage in a unified manner

We also consider LOCC conversion for pure entangled states in quantum information theory.
The asymptotic behavior of LOCC conversion has been intensively studied \cite{BBPS96,BPRST01,HL04,HW03,Hay06,HKMMW03,KH13-2}.
However, 
unlike conventional settings of LOCC conversion,
we assume that LOCC conversion passes through quantum system to store entangled states named {\it entanglement storage}.
In the setting,
an initial i.i.d. pure entangled state is once transformed into the entanglement storage with smaller dimension by LOCC and then transformed again to approximate a target i.i.d. pure state by LOCC. 
In particular, when the target pure entangled state is the same as the original pure entangled state,
this problem can be regarded as LOCC compression of entangled states into the given entanglement storage.
Since the storage to keep the entangled states is implemented with a limited resources,
the analysis for LOCC compression is expected to be useful to store entanglement in small quantum system.
It is known that LOCC convertibility between pure entangled states can be translated to majorization relation between two probability distributions consisting of the squared Schmidt coefficients of the states \cite{Nie99,VJN00}.
Through this translation,
we can reduce the asymptotics of LOCC conversion via entanglement storage into that of RNC via random number storage
as similar to the results of conventional RNC without storage shown in \cite{KH13-2}.
In particular, 
the rate regions for LOCC conversion are immediately derived from those for RNC.



%

The paper is organized as follows.
In Section \ref{sec:family},
we introduce the generalized Rayleigh-normal distribution function as a function defined by an optimization problem.
Then we show its basic properties used in the asymptotics of RNC via restricted storage.
In Section \ref{sec:NARNC}, 
we formulate random number conversion (RNC) via restricted storage by two kinds of approximate conversion methods
and give their relations in non-asymptotic setting.
In Section \ref{sec:ARNC}, 
we proceed to asymptotic analysis for RNC via restricted storage.
Then, we show the relation between the rates of the maximum conversion number and storage size and draw various rate regions in both frameworks of first and second-order asymptotic theory.
In Section \ref{sec:RCR}, 
we see that 
conventional RNC without storage can be regarded as RNC via restricted storage with infinite size.
%
In Section \ref{sec:AQIT}, 
we consider LOCC conversion via entanglement storage for quantum pure states.
Using the results for RNC, 
we derive the asymptotic performance of optimal LOCC conversion.
In particular, optimal LOCC compression rate is derived in the second-order asymptotics.
In Section \ref{sec:Proof}, 
we give technical details of proofs of theorems, propositions and lemmas.
In Section \ref{sec:Conclusion},  
we state the conclusion of the paper.

\section{Generalized Rayleigh-Normal Distribution}\Label{sec:family}


In this section,
we introduce a new family of probability distributions with two parameters  on $\R$.
A function $Z$ on $\R$ is generally called a cumulative distribution function if $Z$ is right continuous, monotonically increasing  and satisfies $\displaystyle\lim_{x\to-\infty} Z(x)=0$ and $\displaystyle\lim_{x\to \infty} Z(x)=1$.
Then, there uniquely exists a probability distribution on $\R$ whose cumulative distribution coincides with $Z$.
That is, 
given a cumulative distribution function in the above sense,
it determines a probability distribution on $\R$.
To define the new probability distribution family,
we give its cumulative distribution function.

We prepare some notations which are needed for the definition of a new distribution function.
For $\mu\in\R$ and $v\in\R_+$,
let $\Phi_{\mu,v}$ and $\phi_{\mu,v}$ be the cumulative distribution function and the probability density function of the normal distribution with the mean $\mu$ and the variance $v$.
We denote $\Phi_{0,1}$ and $\phi_{0,1}$ simply by $\Phi$ and $\phi$.
We employ the continuous fidelity (or the Bhattacharyya coefficient) for
probability density functions $p$ and $q$ on $\R$ defined by
\begin{eqnarray}
{\cal F}(p,q):=\int_{\R}\sqrt{p(x)q(x)}dx.
\end{eqnarray} 
Then, we can define a new probability distribution function as follows, which generalize the Rayleigh-normal distribution function defined in  \cite{KH13-2}.
\begin{df}\Label{rn}
For $v>0$ and $s\in\R$, 
a generalized Rayleigh-normal distribution function $Z_{v,s}$ on $\R$ is defined by
\begin{eqnarray}
Z_{v,s}(\mu)
=1-\sup_{A\in{\cal A}_{s}}{\cal F}\left(\frac{dA}{dx}, \phi_{\mu,v}\right)^2,
\Label{eq:rn}
\end{eqnarray}
where the set ${\cal A}_{s}$ of functions $A:\R\to [0,1]$ is defined by 
\begin{eqnarray}
{\cal A}_{s}=
\left\{A {\Big| }
\begin{array}{l}
{\it continuously~differentiable~monotone}\\
{increasing},
~A(s)=1,
~\Phi\le A\le1
\end{array}
\right\}.
\nonumber
\end{eqnarray}
\end{df}
The generalized Rayleigh-normal distribution function is proven to be a cumulative distribution function later, 
and thus, 
it determines a probability distribution on $\R$.
From the definition,
it can be easily  verified that 
the generalized Rayleigh-normal distribution function has the monotonicity as $Z_{v,s}\ge Z_{v,s'}$ for $s< s'$.
We further remark that Rayleigh-normal distribution function $Z_{v}$ is defined by (\ref{eq:rn}) with $s=\infty$ in \cite{KH13-2},
and thus, the following equation holds
\begin{eqnarray}
\lim_{s\to \infty}Z_{v,s}(\mu)
=\inf_{s\in\R}Z_{v,s}(\mu)
=Z_{v}(\mu).
\Label{limit}
\end{eqnarray}
In this sense,
the family of generalized Rayleigh-normal distribution function $Z_{v,s}$ includes Rayleigh-normal distribution functions as its limit case.

%


The definition of the generalized Rayleigh-normal distribution function is highly abstract and is not in a numerically computable form.
To give a more concrete form of the generalized Rayleigh-normal distribution functions,
we prepare the following three lemmas.

\begin{lem}\Label{sol2}
When $0<v<1$,
the equation with respect to $x$
\begin{eqnarray}\Label{threshold0}
\frac{1-\Phi\left(x\right)}{\Phi_{\mu,v}(s)-\Phi_{\mu,v}(x)}
=\frac{\phi(x)}{\phi_{\mu,v}(x)}
\end{eqnarray}
has the unique solution $\beta_{\mu,v,s}$ and it satisfies 
\begin{eqnarray}
\beta_{\mu,v,s}<\min\{s,\frac{\mu}{1-v}\}.
\Label{conbeta}
\end{eqnarray}
\end{lem}

\begin{lem}\Label{sol0}
When $v=1$ and $\mu>0$, the equation (\ref{threshold0}) with respect to $x$
has the unique solution $\beta_{\mu,v,s}\in\R$. 
\end{lem}

\begin{lem}\Label{sol1}
When $v>1$,
the equation with respect to $x$
\begin{eqnarray}\Label{threshold1'}
\frac{\Phi(x)}{\Phi_{\mu,v}(x)}
=\frac{\phi(x)}{\phi_{\mu,v}(x)}
\end{eqnarray}
has the unique solution $\alpha_{\mu,v}\in\R$.
Moreover, for $s> \Phi_{\mu,v}^{-1}\left(\frac{\Phi_{\mu,v}(\alpha_{\mu,v})}{\Phi(\alpha_{\mu,v})}\right)$, the equation (\ref{threshold0}) with respect to $x$ 
has two solutions and only the larger solution $\beta_{\mu,v,s}$ is larger than $\alpha_{\mu,v}$. 
\end{lem}

Then,
the generalized Rayleigh-normal distribution function is represented as follows.
\begin{thm}\Label{Zform}
The following equations hold:
when $0<v<1$,
\begin{eqnarray}
&&\hspace{-2em}Z_{v,s}(\mu)
=\nonumber\\
&&\hspace{-2em}1-(\sqrt{1-\Phi(\beta_{\mu,v,s})}\sqrt{\Phi_{\mu,v}(s)-\Phi_{\mu,v}(\beta_{\mu,v,s})} + I_{\mu,v}(\beta_{\mu,v,s}))^2;
\Label{Z<1}
\end{eqnarray}
when $v=1$,
\begin{eqnarray}
&&\hspace{-2em}Z_{1,s}(\mu)
=\nonumber\\
&&\hspace{-2em}\left\{
\begin{array}{l}
\Phi(\mu-s) \\
\hspace{1em}{\it if}~~ \mu\le0 \vspace{0.5em} \\
1-(\sqrt{1-\Phi(\beta_{\mu,1,s})}\sqrt{\Phi(s-\mu)-\Phi(\beta_{\mu,1,s}-\mu)} \\
\hspace{2em}+\Phi\left(\beta_{\mu,1,s}-\frac{\mu}{2}\right)
e^{-\frac{\mu^2}{8}})^2\\ 
\hspace{1em}{\it if} ~~\mu>0;
\end{array}
\right.
\Label{Z=1}
\end{eqnarray}
when $v>1$,
\begin{eqnarray}
&&\hspace{-2em} Z_{v,s}(\mu)
=\nonumber\\
&&\hspace{-2em}\left\{
\begin{array}{l}
\hspace{-0.5em}1-\Phi_{\mu,v}(s) \\ 
\hspace{1em}{\it if}~~ s\le\Phi_{\mu,v}^{-1}(\frac{\Phi_{\mu,v}(\alpha_{\mu,v})}{\Phi(\alpha_{\mu,v})}) \vspace{0.5em}\\
\hspace{-0.5em}1-(\sqrt{\Phi(\alpha_{\mu,v}) \Phi_{\mu,v}(\alpha_{\mu,v})} 
 + I_{\mu,v}(\beta_{\mu,v,s}) - I_{\mu,v}(\alpha_{\mu,v})\\
\hspace{2em} +\sqrt{1-\Phi(\beta_{\mu,v,s})}\sqrt{\Phi_{\mu,v}(s)-\Phi_{\mu,v}(\beta_{\mu,v,s})})^2 \\ 
\hspace{1em}{\it if}~~s>\Phi_{\mu,v}^{-1}(\frac{\Phi_{\mu,v}(\alpha_{\mu,v})}{\Phi(\alpha_{\mu,v})}),
\end{array}
\right.
\Label{Z>1}
\end{eqnarray}
where 
\begin{eqnarray}
\hspace{-0em}I_{\mu,v}(x)
&\hspace{-0.7em}:=&\hspace{-0.7em}\sqrt{\frac{2\sqrt{v}}{1+v}}e^{-\frac{\mu^2}{4(1+v)}} \Phi_{\frac{\mu}{1+v}, \frac{2v}{1+v}}\left(x\right),
\Label{Ix}\\
\hspace{-0em}I_{\mu,v}(\infty)
&\hspace{-0.7em}:=&\hspace{-0.7em}\lim_{x\to\infty} I_{\mu,v}(x)
=\sqrt{\frac{2\sqrt{v}}{1+v}}e^{-\frac{\mu^2}{4(1+v)}}.
\Label{Iinfty}
\end{eqnarray}
\end{thm}
Theorem \ref{Zform} is proven in Subsection \ref{Zform.app}
by using lemmas in Subsections  \ref{Zdir.app} and \ref{Zcon.app}.

Using the explicit form in Theorem \ref{Zform}, 
we can prove the following basic property of the Rayleigh-normal distribution function.
\begin{prp}\Label{cum}
The generalized Rayleigh-normal distribution function $Z_{v,s}$ is a cumulative distribution function for $0<v<\infty$.
\end{prp}
Proposition \ref{cum} is proven in Subsection \ref{cum.app}.

Next we show the concrete form of the generalized Rayleigh-normal distribution function in the case of  $v\to 0$.
\begin{prp}\Label{lim1}
\begin{eqnarray}
\lim_{v\to0}Z_{v,s}(\mu)
=
\left\{
\begin{array}{cll}
\Phi\left(\mu\right) & {\it if} & \mu< s
\\
\frac{1}{2}(1+\Phi\left(\mu\right)) & {\it if} & \mu= s
\\
1&{\it if}&\mu> s
\end{array}
\right.\Label{eqlim1}
\end{eqnarray}
\end{prp}
Proposition $\ref{lim1}$ is proven in Subsection \ref{lim1.app}.
The function itself in Proposition \ref{lim1} is not right continuous, and thus, not a cumulative distribution function.
However, 
if we redefine the function value by $1$ only at $\mu=s$ in \eqref{eqlim1},
the function in \eqref{eqlim1} becomes right continuous, and thus is 
a cumulative distribution function.
Nevertheless, 
we define the generalized Rayleigh-normal distribution with $v=0$ as a left-continuous function as follows
to describe the asymptotics of RNC via restricted storage later:
\begin{eqnarray}
Z_{0,s}(\mu)
:=
\left\{
\begin{array}{cll}
\Phi\left(\mu\right) & {\it if} & \mu\le s\\
1&{\it if}&\mu> s.
\end{array}
\right.
\Label{Z0}
\end{eqnarray}
 
We also have the concrete form of the generalized Rayleigh-normal distribution function in the case of $v\to \infty$.
\begin{prp}\Label{lim2}
\begin{eqnarray}
\lim_{v\to\infty}Z_{v,\sqrt{v}s}(\sqrt{v}\mu)
=\Phi\left(\mu-\min\{s,0\})\right)
\end{eqnarray}
\end{prp}
Proposition $\ref{lim2}$ is proven in Subsection \ref{lim2.app}.

The graphs of the generalized Rayleigh-normal distribution functions can be plotted as in Figs. \ref{RNs} and \ref{RNv}.
%
\begin{figure}[t]
 \begin{center}
 \hspace*{0em}\includegraphics[width=80mm, height=50mm]{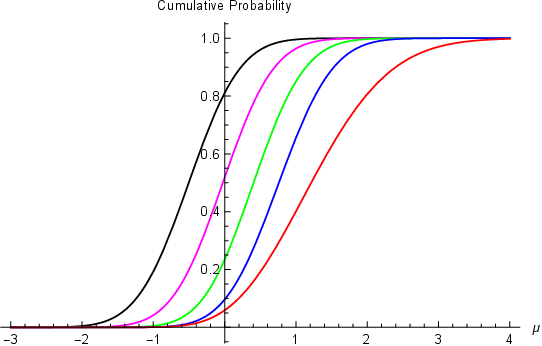}
 \end{center}
 \caption{
The black, purple, green, blue and red lines represent the generalized Rayleigh-normal distribution functions with parameter $s=-0.5$, $0$, $0.5$, $1$ and $\infty$ at $v=1/3$.
}
 \Label{RNs}
\end{figure}
\begin{figure}[t]
 \begin{center}
 \hspace*{0em}\includegraphics[width=80mm, height=50mm]{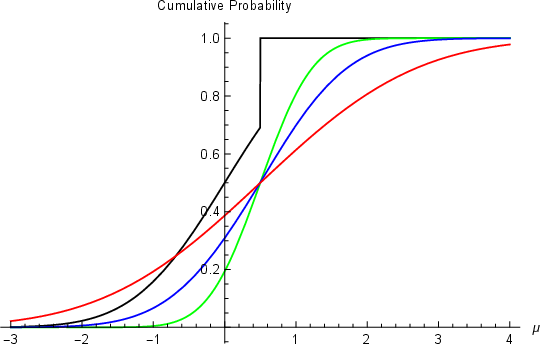}
 \end{center}
 \caption{
The black, green, blue and red lines represent the generalized Rayleigh-normal distribution functions with $v=0$, $1/3$, $1$ and $3$ at $s=0.5$.
}
 \Label{RNv}
\end{figure}

\section{Non-Asymptotics for Random Number Conversion via Restricted Storage}
\Label{sec:NARNC}

We introduce two kinds of conversion methods of probability distributions, i.e., 
deterministic conversions and majorization conversions as follows.
%

\subsection{Deterministic Conversion}
\Label{subsec:DC}

In this subsection, 
we consider approximate conversion problems 
when the conversion is routed through a storage with limited size.

Let $\calP(\calX)$ be the set of all probability distributions on a finite set $\calX$.
For $P\in\calP(\calX)$ and a map $f:\mathcal{X}\to\mathcal{Y}$, 
we define the probability distribution $W_f(P)\in\calP(\calY)$ by 
\begin{eqnarray}
W_f(P)(y):=\sum_{x\in W^{-1}(y)}P(x).
\Label{eq:deterministic}
\end{eqnarray}
We call a map $W_f:\calP(\mathcal{X}) \to \calP(\mathcal{Y})$ defined in (\ref{eq:deterministic}) 
a {\it deterministic conversion}.

In order to treat the quality of conversion,
we introduce the fidelity (or the Bhattacharyya coefficient) $F$ between two probability distributions over the same discrete set $\mathcal{Y}$ as
\begin{eqnarray}
F(Q, Q'):=\sum_{y\in\mathcal{Y}}\sqrt{Q(y)}\sqrt{Q'(y)}.
\end{eqnarray}
Since this value $F(Q, Q')$ relates to the Hellinger distance $d_H$ as $d_H(Q, Q')=\sqrt{1-F(Q, Q')}$ \cite{Vaa98},
it represents how close two probability distributions $Q$ and $Q'$.
Then, we define the maximal fidelity $F^{\cal D}$ from $P\in\calP(\calX)$ to $Q\in\calP(\calY)$ 
among deterministic conversions by
\begin{eqnarray}
F^{\cal D}(P\to Q)
&:=&\sup_{W:\calP(\calX) \to \calP(\calY)}
\{F(W(P), Q)| W~ \text{is a deterministic conversion}\}\\
&=&\sup_{f:\mathcal{X}\to\mathcal{Y}}
F(W_f(P), Q)
\end{eqnarray}
Moreover, when the size of a storage is limited,
the maximal fidelity via restricted storage with size of $N$ bits
is defined by
\begin{eqnarray*}
&&\hspace{-1.5em}F^{\cal D}(P\to Q|N)\\
&&\hspace{-1.9em}:=\sup_{f:\mathcal{X}\to\{0,1\}^N, 
f':\{0,1\}^N\to\mathcal{Y}}
F(W_{f'}\circ W_f(P), Q)
\end{eqnarray*}
where $\{0,1\}^N$ represents the space of $N$-bits.

When a confidence coefficient $0<\nu<1$ is fixed,
we define the maximal conversion number $L$ of copies of $Q$ 
by deterministic conversions with the initial distribution $P$ as
\begin{eqnarray*}
L^{\cal D}(P, Q|\nu)
:=\max
\{L|
\exists f:\mathcal{X}\to\mathcal{Y}^L,~
F(W_f(P), Q^L)\ge\nu
\}.
\end{eqnarray*}
Moreover, when the size of the storage is limited,
the maximum conversion number from $P$ to $Q$ via a restricted storage with size of $N$ bits
is defined by 
\begin{eqnarray*}
&&\hspace{-1.5em}L^{\cal D}(P, Q|\nu, N)\\
&&\hspace{-1.5em}:=\max
\left\{L\Bigg|
\begin{array}{l}
\exists f:\mathcal{X}\to\{0,1\}^N, 
\exists f':\{0,1\}^N\to\mathcal{Y},\\
F(W_{f'}\circ W_f(P), Q^L)\ge\nu
\end{array}
\right\}. \\
\end{eqnarray*}
Then the above values can be rewritten as
\begin{eqnarray*}
L^{\cal D}(P, Q|\nu)&=&
\max\{L|F^{\cal D}(P\to Q^L)\ge\nu\},\\
L^{\cal D}(P, Q|\nu,N)&=&\max\{L|F^{\cal D}(P\to Q^L|N)\ge\nu\}.
\end{eqnarray*}
In particular, when the source distribution is $n$-fold i.i.d. of $P$, 
we define
\begin{eqnarray*}
L^{\cal D}_n(P, Q|\nu)&:=&L^{\cal D}(P^n, Q|\nu),\\
L^{\cal D}_n(P, Q|\nu, N)&:=& L^{\cal D}(P^n, Q|\nu, N). 
\end{eqnarray*}
One of main issues is the asymptotic expansion of $L^{\cal D}_n(P, Q|\nu, N)$ up to the second order $\sqrt{n}$.

\subsection{Majorization Conversion}
\Label{subsec:MC}

In order to relax the condition for deterministic conversions, 
we introduce majorization conversions.
This relaxed condition is useful for the proofs of converse parts. 
Moreover, the concept of majorization conversions is essentially required for entanglement conversion in quantum information.
For a probability distribution $P$ on a finite set, 
let $P^{\downarrow}$ be a probability distribution on $\{1,2,...,|\calX|\}$ 
and $P^{\downarrow}_i$ denote the $i$-th element of $\{P(x)\}_{x\in\mathcal{X}}$ sorted in decreasing order for $1\le i\le|\mathcal{X}|$,
where $|{\cal X}|$ represents the cardinality of the set $\mathcal{X}$.
When two probability distributions $P\in\calP(\calX)$ and $Q\in\calP(\calY)$ 
satisfy $\sum_{i=1}^lP^{\downarrow}_i \le \sum_{i=1}^lQ^{\downarrow}_i$ for any $l$, we say that $P$ is majorized by $Q$ and written as $P\prec Q$.
Here, we note that the sets $\calX$ and $\calY$ 
do not necessarily coincide with each other, 
and the majorization relation is a partial order on a set of probability distributions on finite sets \cite{MO79,Arn86}. 
Then, a map $W'$ from $\calP(\calX)$ to $\calP(\calY)$ is called a {\it majorization conversion} 
when $P\prec W'(P)$ for an arbitrary probability distribution $P \in \calP(\calX)$.

Then, we introduce the maximal fidelity among majorization conversions as
\begin{eqnarray}
F^{\cal M}(P\to Q)
&:=&\sup_{W':\calP(\calX) \to \calP(\calY)}
\{F(W'(P), Q)| W'\text{ is a majorization conversion}\}\\
&=& \sup_{P'\in\calP(\calY)} \{F(P', Q)| P\prec P' \}
\end{eqnarray}
where $P$ and $Q$ are probability distributions on $\mathcal{X}$ and $\mathcal{Y}$, respectively.
%
Moreover, when the size of the storage is limited,
the maximal fidelity via restricted storage with size of $N$ bits
is given by
\begin{eqnarray*}
&&F^{\cal M}(P\to Q|N)\\
&:=&\sup_{P''\in\calP(\calY)}
\{F(P'', Q)|
\exists P'\in\mathcal{P}(\{0,1\}^N),
P\prec P'\prec P''
\}.
\end{eqnarray*}

Similar to the deterministic conversion,
when confidence coefficient $0<\nu<1$ is fixed,
we define the maximum conversion number $L$ of $Q^L$ which can be approximated from $P$ 
by majorization conversions as
\begin{eqnarray*}
L^{\cal M}(P, Q|\nu, N)
&:=&\max \{L|F^{\cal M}(P\to Q^L|N)\ge\nu\}.
\end{eqnarray*}
Moreover, when the size of the storage is limited,
the maximum conversion number from $P$ to $Q$ via restricted storage with size of $N$ bits
is defined by 
\begin{eqnarray*}
&&\hspace{-1.5em}L^{\cal M}(P, Q|\nu, N)\\
&&\hspace{-1.5em}:=
\max
\left\{L\Bigg|
\begin{array}{l}
\exists P''\in\calP(\calY),
\exists P'\in\mathcal{P}(\{0,1\}^N),
\\
P\prec P'\prec P'', 
F(P'', Q^L)\ge\nu
\end{array}
\right\}.
\end{eqnarray*}
Then the above values can be rewritten as
\begin{eqnarray}
\hspace{-0.5em}L^{\cal M}(P, Q|\nu)&=&\max\{L|F^{\cal M}(P\to Q^L)\ge\nu\},\nonumber\\
\hspace{-0.9em}L^{\cal M}(P, Q|\nu,N)&=&\max\{L|F^{\cal M}(P\to Q^L|N)\ge\nu\}.\Label{M1}
\end{eqnarray}
In particular, when the source distribution is $n$-fold i.i.d. of $P$, 
we define
\begin{eqnarray*}
L^{\cal M}_n(P, Q|\nu)&:=&L^{\cal M}(P^n, Q|\nu),\\
L^{\cal M}_n(P, Q|\nu, N)&:=& L^{\cal M}(P^n, Q|\nu, N).
\end{eqnarray*}
One of main issues of this paper is the asymptotic expansion of $L^{\cal M}_n(P, Q|\nu, N)$ up to the order $\sqrt{n}$.
This quantity plays an important role in quantum information theory.

\subsection{Basic Properties of Two Conversions }
\Label{subsec:RPTC}

In this subsection,
we summarize some properties about deterministic and majorization conversions.

First, we summarize some properties about the maximum fidelity of two conversions.
The following lemma holds for the uniform distribution $U_N$ in the non-asymptotic setting.
\begin{lem}\cite{KH13-2}\Label{opt con}
For a probability distribution $P$ and a natural number $N$, 
we define the following distribution ${\cal C}_N(P)$ 
on $\{1, \ldots, N\}$ as a distribution approximating the uniform distribution:
\begin{eqnarray}
{\cal C}_N(P)(j)
:=\left\{
\begin{array}{lll}
P^{\downarrow}(j)&\hbox{ if }&1\le j\le J_{P,N}-1\\
\frac{\sum_{i=J_{P,N}}^{|\mathcal{X}|} P^{\downarrow}(i)}{N+1-J_{P,N}}&\hbox{ if }&~J_{P,N}\le j\le N
\end{array}
\right.\Label{PL}
\end{eqnarray}
where 
\begin{eqnarray}
J_{P,N}
&:=&
\left\{
\begin{array}{ll}
1&\hbox{if } P^{\downarrow}(1)\le \frac{1}{N}\\
\max\left\{j \in \{2, \ldots, N\} 
\left|\frac{\sum_{i=j}^{|\mathcal{X}|} P^{\downarrow}(i)}{N+1-j}<P^{\downarrow}(j-1)\right.\right\} &\hbox{otherwise.}
\end{array}\right.
\Label{J}
\end{eqnarray}
Then, 
$P\prec {\cal C}_N(P)$ and 
the following equation hold:
\begin{eqnarray}
\hspace{-1.5em}F^{\mathcal{M}}(P\to U_N)
&\hspace{-0.1em}=&F({\cal C}_N(P), U_N)\nonumber\\
&\hspace{-0.1em}=&\sqrt{\frac{1}{N}}\left(\sum_{j=1}^{J_{P,N}-1}\sqrt{P^{\downarrow}(j)}+\sqrt{(N+1-J_{P,N})\sum_{i=J_{P,N}}^{|\mathcal{X}|} P^{\downarrow}(i)}\right).
\end{eqnarray}
\end{lem}

%
In addition, the following lemma holds.
\begin{lem}\Label{opt trans}
For probability distributions $P\in\calP(\calX)$,  $Q\in\calP(\calY)$ and a natural number $N$, 
\begin{eqnarray}
F^{\cal M}(P\to Q|N)=F^{\cal M}({\cal C}_{2^N}(P) \to Q)
\end{eqnarray}
where ${\cal C}_{2^N}(P)$ was defined in (\ref{PL}).
\end{lem}
We provide the proof of Lemma \ref{opt trans} in Section \ref{opt trans.app}.
Note that ${\cal C}_\phi(P)$ depends on the source distribution $P$ and 
does not on the target distribution $Q$ in Lemma \ref{opt trans}.
This fact is essential in the asymptotics for $F^{\cal M}(P\to Q|N)$.

We remark that $P\prec W(P)$ holds for a deterministic conversion 
$W:\calP(\calX)\to\calP(\calY)$,
and thus, a deterministic conversion is a majorization conversion.
Therefore, we have the relations
\begin{eqnarray}
F^{\cal D}(P\to Q) &\le& F^{\cal M}(P\to Q),\Label{fidelity-ineq} \\
F^{\cal D}(P\to Q|N) &\le& F^{\cal M}(P\to Q|N). 
\Label{fidelity-ineq2}
\end{eqnarray}
Next, we summarize some properties about the maximum conversion number of two conversion.
From (\ref{fidelity-ineq}) and (\ref{fidelity-ineq2}), we have
\begin{eqnarray}
L^{\cal M}_n(P, Q|\nu) &\ge& L^{\cal D}_n(P, Q|\nu),\Label{number-ineq}\\
L^{\cal M}_n(P, Q|\nu, N)&\ge& L^{\cal D}_n(P, Q|\nu, N).
\Label{number-ineq2}
\end{eqnarray}

One of main issues of this paper is to derive the asymptotic behaviors of 
$L^{\cal M}_n(P, Q|\nu, N)$ and $L^{\cal D}_n(P, Q|\nu, N)$
as stated above.
Fortunately, 
when either the source distribution $P$ or the target distribution $Q$ is a uniform distribution,
their asymptotic behaviors are evaluated by direct conversions without storage in the following way.

\begin{prp}\Label{opt trans cor}
\begin{eqnarray}
L^{\cal D}_n(U_N, Q|\nu, m\log N)
&\ge& L^{\cal D}_{\min\{ n ,m\}}(U_N, Q|\nu), \Label{D1} \\
L^{\cal M}_n(U_N, Q|\nu, m\log N)
&=& L^{\cal M}_{\min\{ n ,m\}}(U_N, Q|\nu), \Label{M1}
\end{eqnarray}
where $\log$ indicate the logarithm to the base $2$. 
\end{prp}

\begin{prp}\Label{opt trans cor2}
Let $i={\cal D}$ or ${\cal M}$.
When 
$m \ge L^i_n(P, U_N|\nu) $,
\begin{eqnarray}
L^i_n(P, U_N|\nu, m\log  N)
=L^i_n(P, U_N|\nu). \Label{D2-1}
\end{eqnarray}
Otherwise,
\begin{eqnarray}
m \le L^i_n(P, U_N|\nu, m\log  N) \le m -2\log_N \nu. \Label{D2-2}
\end{eqnarray}
\end{prp}

We provide the proofs of Lemmas \ref{opt trans cor} and \ref{opt trans cor2}  in Appendices \ref{opt trans cor.app} and \ref{opt trans cor2.app}, respectively.

\section{Asymptotics for Random Number Conversion via Restricted Storage}
\Label{sec:ARNC}

When the number of copies of an initial distribution is $n$,
we  consider the relation of the size $S_n$ of storage and the number $T_n$ of copies of a target distribution in this section.
\begin{df}
A sequence $\{(S_n,T_n)\}_{n=1}^{\infty}$ is called $\nu$-{\it achievable}  with respect to the deterministic conversion or the majorization conversion
if it satisfies
\begin{eqnarray}
\underset{n\to\infty}{\rm liminf} F^{i}(P^{ n}\to Q^{T_n}|{S_n})
\ge \nu
\end{eqnarray}
for $i={\cal D}$ or ${\cal M}$, respectively.
\end{df}
For a sequence $\{(S_n,T_n)\}$,
smaller $S_n$ and larger $T_n$ give a better performance.
Hence, we say that a sequence $\{(S_n,T_n)\}$ {\it dominates} another one $\{(S'_n,T'_n)\}$
when there exists $N\in\N$ such that $S_n \le S'_n$ and $T_n \ge T'_n$ for $n\ge N$.
Similarly, we say that a sequence $\{(S_n,T_n)\}$ {\it simulates} another sequence $\{(S'_n,T'_n)\}$ when there exists a sequence $\{a_n\}\subset(0,1]$ such that  $(S'_n,T'_n)=(S_{a_nn},T_{a_nn})$.


When a $\nu$-achievable sequence $\{(S_n,T_n)\}$ dominates a sequence $\{(S'_n,T'_n)\}$,
the sequence $\{(S'_n,T'_n)\}$ is also $\nu$-achievable obviously.
Moreover, the following lemma holds.
\begin{lem}\Label{simulate}
When a $\nu$-achievable sequence $\{(S_n,T_n)\}$ simulates a sequence $\{(S'_n,T'_n)\}$,
the sequence $\{(S'_n,T'_n)\}$ is also $\nu$-achievable.
\end{lem}
We provide the  proof of Lemma \ref{simulate} in Section \ref{simulate.app}.




\subsection{First-Order Rate Region }~
\Label{subsec:FORR}

In this subsection,
we assume that a sequence $\{(S_n,T_n)\}$ is represented by $S_n=s_1n+o(n)$ and $T_n=t_1n+o(n)$ with the first-order rates $s_1>0$ and $t_1>0$
and focus on the first-order asymptotics of RNC via restricted storage.
In the following, Then, we omit the $o(n)$ term unless otherwise noted.
%
\begin{df}
A first-order rate pair $(s_1,t_1)$ is called  $\nu$-{\it achievable} 
when a sequence $\{(s_1n,t_1n)\}$ is $\nu$-achievable.
The set of $\nu$-achievable rate pairs for $i={\cal D}$ and ${\cal M}$ is denoted by
\begin{eqnarray}
\hspace{-0em}{\cal R}^{1,i}_{P,Q}(\nu)
:=\left\{\left(s_1,t_1\right)\bigg|
\underset{n\to\infty}{\rm liminf}F^{i}(P^{ n}\to Q^{t_1 n}|{s_1 n})
\ge \nu \right\}.
\label{reg1D}
\end{eqnarray}
\end{df}

Then, we have the following characterization.
\begin{thm}\Label{region1}
For $\nu\in (0,1)$,
\begin{eqnarray}
&&\hspace{-1em}{\cal R}^{1,{\cal D}}_{P,Q}(\nu)={\cal R}^{1,{\cal M}}_{P,Q}(\nu)\nonumber
\\&\hspace{-2em}=&
\hspace{-1em}\left\{\left(s_1,t_1\right)\bigg|0 < s_1, 0< t_1\le \frac{\min\{H(P), s_1 \}}{H(Q)} \right\},
\label{reg1}
\end{eqnarray}
where $H(P)$ and $H(Q)$ are the Shannon entropy of $P$ and $Q$, respectively.
\end{thm}
We give the proof of Theorem \ref{region1} in Section \ref{region1.app}.
From Theorem \ref{region1}, ${\cal R}^{1,{\cal D}}_{P,Q}(\nu)$ and ${\cal R}^{1,{\cal M}}_{P,Q}(\nu)$ coincide with each other and do not depend on $\nu\in (0,1)$. 
In the following, we denote the rate regions by ${\cal R}^{1}_{P,Q}$ simply.

We say that $(s_1,t_1)$ dominates or simulates $(s'_1,t'_1)$ 
when the sequence $\{(s_1n, t_1n)\}$ dominates or simulates the sequence $\{(s'_1n,t'_1n)\}$.
Then, $(s_1,t_1)$  dominates $(s'_1,t'_1)$ if and only if $s_1\le s'_1$ and $t_1\ge t'_1$.
Similarly, $(s_1,t_1)$ simulates $(s'_1,t'_1)$ if and only if ${s'_1}/{s_1}={t'_1}/{t_1}\le 1$.
\begin{df}
When no other achievable rate pair dominates
$(s_1,t_1)\in{\cal R}^{1}_{P,Q}$,
the rate pair $(s_1,t_1)$ is called semi-admissible.
Moreover, when no other rate pair dominates or simulates $(s_1,t_1)\in{\cal R}^{1}_{P,Q}$,
the rate pair $(s_1,t_1)$ is called admissible.
\end{df}
We obtain the following corollary by Theorem \ref{region1}.
\begin{cor}\Label{adm}
The set of semi-admissible rate pairs is given by  
\begin{eqnarray}
\left\{\left(s_1,\frac{s_1 }{H(Q)}\right)\bigg|0 < s_1 \le H(P) \right\}
\Label{semiad}
\end{eqnarray}
and $(H(P),H(P)/H(Q))$ is the unique admissible rate pair.
\end{cor}

The rate region is illustrated as Fig. \ref{1st}.
Then, the set of semi-admissible rate pairs are illustrated as the line with the slope $H(Q)^{-1}$ and the admissible rate pair is dotted at the tip of the line.
%
\begin{figure}[t]
 \begin{center}
 \hspace*{0em}\includegraphics[width=80mm, height=55mm]{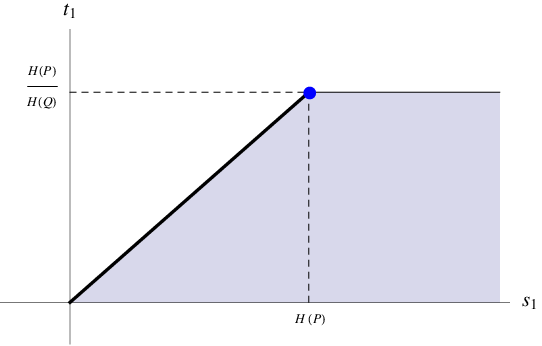}
 \end{center}
 \caption{
The first-order rate region ${\cal R}^{1,{\cal D}}_{P,Q}(\nu)$ and ${\cal R}^{1,{\cal M}}_{P,Q}(\nu)$.
The thick line corresponds to the semi-admissible rate pairs.
}
 \Label{1st}
\end{figure}
%
We note that the admissible first-order rate pair can determine whether a rate pair is in the rate region.
That is, a rate pair is in the rate region if and only if 
the admissible rate pair simulates or dominates the rate pair.
Thus, the admissible rate pair uniquely determines the whole of rate region although it is a single point in the boundary of the rate region.

In later discussion, 
we separately treat the problem according to whether a semi-admissible rate pair is the admissible rate pair or not.

\subsection{Second-Order Rate Region }~
\Label{subsec:SORR}

In this subsection,
we fix a first-order rate pair $(s_1,t_1)$ of each sequence $\{(S_n,T_n)\}$
and assume it to be $\nu$-achievable.
Let the sequence $(S_n,T_n)$ be represented by $S_n=s_1n+s_2\sqrt{n}+o(\sqrt{n})$ and $T_n=t_1n+t_2\sqrt{n}+o(\sqrt{n})$ with second-order rates $s_2\in\R$ and $t_2\in\R$.
Then we focus on the second-order asymptotics of RNC via restricted storage in terms of $s_2$ and $t_2$.
We omit the $o(\sqrt{n})$ term unless otherwise noted.
%
\begin{df}\Label{rate-region2}
A second-order rate pair $(s_2,t_2)$ is called  $\nu$-{\it achievable} 
when the sequence $\{(s_1n+s_2\sqrt{n},t_1n+t_2\sqrt{n})\}$ is $\nu$-achievable.
The set of $\nu$-achievable rate pairs for $i={\cal D}$ and ${\cal M}$ is denoted  by
\begin{eqnarray}
&&{\cal R}^{2,i}_{P,Q}(s_1,t_1,\nu)\nonumber\\
&:=&
\left\{\left(s_2,t_2\right)\bigg|
\liminf_{n\to\infty} F^i\left(P^{ n}\to Q^{t_1n+t_2\sqrt{n}}|{s_1n+s_2\sqrt{n}}\right)\ge \nu \right\}.
\label{reg2D}\nonumber
\end{eqnarray}
\end{df}

%
If the first-order rate pair is $\nu$-achievable and not semi-admissible,
the second-order rate region is trivially the whole of $\R^2$.
In the following, we treat the case that the first-order rate pair $(s_1,t_1)$ is semi-admissible, i.e., $0< s_1 \le H(P)$ and $t_1=s_1/H(Q)$. 
Then, we set as
 \begin{eqnarray*}
&&F^i_{P,Q,s_1,s_2}(t_2)\\
&:=&
\liminf_{n\to\infty} F^i\left(P^{ n}\to Q^{\frac{s_1}{H(Q)}n+t_2\sqrt{n}}|{s_1n+s_2\sqrt{n}}\right).
\end{eqnarray*}

%
%
\begin{lem}\Label{fid}
Let $P$ and $Q$ be arbitrary probability distributions on finite sets
and $0<s_1\le H(P)$.
Then, there is a continuous function $F_{P,Q,s_1,s_2}:\R\to[0,1]$
satisfying the following conditions.
(1) The function $F_{P,Q,s_1,s_2}$ is strictly monotonically decreasing on $F_{P,Q,s_1,s_2}^{-1}((0,1))$.
(2) The relation 
\begin{eqnarray}
F_{P,Q,s_1,s_2}(t_2)
=F^{\cal D}_{P, Q,s_1,s_2}(t_2)
=F^{\cal M}_{P, Q,s_1,s_2}(t_2)
\Label{feq}
\end{eqnarray}
following holds for an arbitrary $t_2\in\R$.
\end{lem}
Lemma \ref{fid} is derived from Theorems \ref{region2-non-admissible}, \ref{main}, \ref{region2-uni1} and \ref{region2-uni2} in the later subsections.
From the above lemma, 
we obtain the asymptotic expansions of the maximal conversion numbers.
%
%
\begin{thm} \Label{expansion}
Let $P$ and $Q$ be arbitrary probability distributions on finite sets.
For arbitrary $s_1>0$, $s_2\in\R$ and $\nu\in(0,1)$,
\begin{eqnarray}
&&\hspace{-0em}L^{\cal D}_n(P,Q|\nu, {s_1n+s_2\sqrt{n}})
\cong L^{\cal M}_n(P,Q|\nu, {s_1n+s_2\sqrt{n}})\nonumber\\
&\hspace{-0.5em}\cong&\hspace{-0.5em}\frac{\min\{H(P),s_1\}}{H(Q)}n+
 F_{P,Q,s_1,s_2}^{-1}(\nu)\sqrt{n},\Label{exp.=2}
\end{eqnarray}
where $\cong$ means that the difference between the right-hand side and the left-hand side of $\cong$ is $o(\sqrt{n})$.
\end{thm}

Theorem \ref{expansion} is derived as follows.
When we expand as $L^{i}_n(P,Q|\nu, {s_1n+s_2\sqrt{n}})=t_1n+t_2\sqrt{n}$ for $i={\cal D}$ or ${\cal M}$,
the first order rate $t_1$ is determined by Lemma \ref{region1} as $t_1=\frac{\min\{H(P),s_1\}}{H(Q)}$.
Moreover, since the second order rate $t_2$ satisfies $F_{P, Q,s_1,s_2}(t_2)=\nu$ from the definition of $t_2$,
we have Theorem \ref{expansion}.
%

Moreover, Theorem \ref{expansion} implies the following theorem about the second-order rate regions.
\begin{thm} \Label{general-2reg}
Let $P$ and $Q$ be arbitrary probability distributions on finite sets.
For $0<s_1\le H(P)$, $s_2\in\R$ and $\nu\in(0,1)$,
\begin{eqnarray}
&&{\cal R}^{2,{\cal D}}_{P,Q}\left(s_1,\frac{s_1}{H(Q)},\nu\right)
={\cal R}^{2,{\cal M}}_{P,Q}\left(s_1,\frac{s_1}{H(Q)},\nu\right)\nonumber\\
&=&\left\{\left(s_2,t_2\right)\bigg|
t_2\le F^{-1}_{P, Q,s_1,s_2}(\nu) \right\}.
\nonumber
\end{eqnarray}
\end{thm}

We say that $(s_2,t_2)$ dominates or simulates $(s'_2,t'_2)$ 
when the sequence $\{(s_1n+s_2\sqrt{n},t_1n+t_2\sqrt{n})\}$ dominates or simulates
the sequence $\{(s_1n+s'_2\sqrt{n},t_1n+t'_2\sqrt{n})\}$.
Then, $(s_2,t_2)$  dominates $(s'_2,t'_2)$ if and only if $s_2\le s'_2$ and $t_2\ge t'_2$.
In addition, the following lemma holds. 
\begin{lem}\Label{simulate-2nd}
A $\nu$-achievable rate pair $(s_2,t_2)$ simulates another one $(s'_2,t'_2)$
if and only if $s_2\ge s'_2$ and 
\begin{eqnarray}
t'_2
=t_2+\frac{t_1}{s_1}(s'_2-s_2).
\Label{eq-st}
\end{eqnarray}
\end{lem}
We provide the  proof of Lemma \ref{simulate-2nd} in Section \ref{simulate-2nd.app}.

\begin{df}
Let $(s_2,t_2)$ be a $\nu$-achievable second-order rate pair.
The rate pair $(s_2,t_2)$ is called semi-admissible
when no other $\nu$-achievable rate pair dominates $(s_2,t_2)$.
Moreover, the rate pair $(s_2,t_2)$ is called admissible
when no other $\nu$-achievable rate pair dominates or simulates $(s_2,t_2)$.
\end{df}
In the following subsections,
we separately derive the concrete forms of  second-order rate regions
and determine the set of second-order semi-admissible and admissible rate pairs
for the non-admissible and the admissible first-order rate pair.

Unlike the first-order case,
the set of admissible second-order rate pairs does not necessarily consist of a single point 
and there are also the cases that 
multiple admissible rate pairs exist and
no admissible rate pair exists
as shown in later subsections. 
On the other hand, 
similar to the first-order asymptotics,
the admissible second-order rate pairs can determine whether a rate pair is in the rate region.
That is, 
a rate pair is in the rate region
if and only if there is an admissible rate pair such that the admissible rate pair simulates or 
dominates the rate pair.
Thus, the admissible rate pairs uniquely determine the whole of rate region although those are a subset of the boundary of the rate region.
Moreover, 
since any admissible rate pair does not simulate or dominate another admissible one,
a proper subset of the admissible rate pairs can not determine the rate region as above.
In the sense, the admissible rate pairs can be regarded as the ``minimal generator" of the rate region, and hence, are of special importance in the rate pairs.

\subsection{Second-Order Asymptotics: Non-Admissible Case}~
\Label{subsec:NEC}

We derive the second-order rate region in the following.
We say that a second-order rate pair $(s_2,t_2)$ is $(s_1,t_1,\nu)$-{\it achievable} 
by deterministic conversions or majorization conversions 
when $(s_2,t_2)\in{\cal R}^{2,{\cal D}}_{P,Q}(s_1,t_1,\nu)$ or ${\cal R}^{2,{\cal M}}_{P,Q}(s_1,t_1,\nu)$.

\begin{thm}\Label{region2-non-admissible}
When $(s_1,t_1)$ is semi-admissible but not admissible,
the function 
\begin{eqnarray}
F_{P,Q,s_1,s_2}(t_2)
=\sqrt{\Phi\left(\sqrt{\frac{H(Q)}{V(Q)s_1}}(s_2-H(Q)t_2)\right)}
\Label{non-ex-fid}
\end{eqnarray}
is continuous and strictly monotonically decreasing on $F_{P,Q,s_1,t_1,s_2}^{-1}((0,1))$ and satisfies (\ref{feq}),
where 
\begin{eqnarray}
V(Q):=\displaystyle\sum_{x\in\mathcal{X}}Q(x)(-\mathrm{log}Q(x)-H(Q))^2. 
\end{eqnarray} 
\end{thm}

We give the proof of Theorem \ref{region2-non-admissible} in Section \ref{region2-non-admissible.app}.
When $(s_1,t_1)$ is semi-admissible but not admissible,
from Theorems \ref{general-2reg} and \ref{region2-non-admissible},
the second-order rate region is given by
\begin{eqnarray}
&&\hspace{-1em}
{\cal R}^{2,{\cal D}}_{P,Q}(s_1,t_1,\nu)
={\cal R}^{2,{\cal M}}_{P,Q}(s_1,t_1,\nu)\nonumber
\\&\hspace{-2em}=&
\hspace{-1em}\left\{\left(s_2,t_2\right)\bigg|t_2 \le \frac{s_2}{H(Q)} -\sqrt{\frac{V(Q) s_1}{H(Q)^3}} \Phi^{-1}(\nu^2) \right\}.
\label{reg1}
\end{eqnarray}
In particular,
the set of admissible rate pairs is represented by 
\begin{eqnarray}
\left\{\left(s_2, \frac{s_2}{H(Q)} -\sqrt{\frac{V(Q) s_1}{H(Q)^3}} \Phi^{-1}(\nu^2)\right)\bigg|s_2\in\R \right\}.
\end{eqnarray}
In this case, 
there is no admissible rate pair.
The second-order rate region is illustrated as Fig. \ref{2nd}
and the boundary of the region is the set of semi-admissible rate pairs from Lemma \ref{simulate-2nd}.

\begin{figure}[t]
 \begin{center}
 \hspace*{0em}\includegraphics[width=80mm, height=55mm]{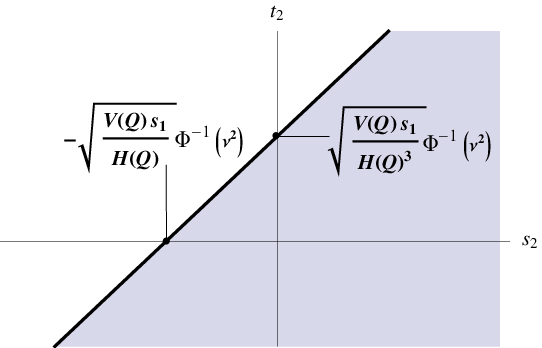}
 \end{center}
 \caption{
The second-order rate region ${\cal R}^{2,{\cal D}}_{P,Q}(s_1,t_1,\nu)$ and ${\cal R}^{2,{\cal M}}_{P,Q}(s_1,t_1,\nu)$ when a first-order rate pair $(s_1,t_1)$ is semi-admissible but not admissible.
}
 \Label{2nd}
\end{figure}
%
\begin{figure}[t]
 \begin{center}
 \hspace*{0em}\includegraphics[width=80mm, height=55mm]{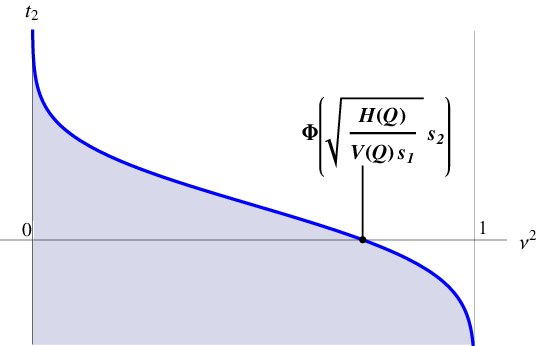}
 \end{center}
 \caption{
The relation between permissible accuracy and second-order rate of the number of copies of a target distribution.
}
 \Label{2nd-fid}
\end{figure}

\subsection{Second-Order Asymptotics: Admissible Case}~
\Label{subsec:Uniform}

The remaining problem is to identify the second-order
rate region at the admissible first-order rate pair. 
Hence, we fix as $s_1 = H(P)$ and $t_1 =  \frac{H(P)}{H(Q)}$ 
and simply denote as
\begin{eqnarray}
F^i_{P,Q,s_2}(t_2)
&:=&F_{P,Q,H(P),s_2}^i(t_2),\label{2ndF}\\
{\cal R}^{2,i}_{P,Q}(\nu)
&:=&{\cal R}^{2,{i}}_{P,Q}\left(H(P),\frac{H(P)}{H(Q)},\nu\right)
\label{2ndR}
\end{eqnarray}
for $i={\cal D}$ or ${\cal M}$ in the following subsections.


First, we treat the case when both $P$ and $Q$ are non-uniform distributions.
Here, we introduce two values as 
\begin{eqnarray}\Label{C}
C_{P,Q}&:=&\frac{H(P)}{V(P)}\left(\frac{H(Q)}{V(Q)}\right)^{-1},\\
D_{P,Q}&:=&\frac{H(Q)}{\sqrt{V(P)}}.
\end{eqnarray}
Then, the optimal accuracy $F_{P, Q,s_2}(t_2)$ is charcterized by the generalized Rayleigh-normal distribution function as follows.
\begin{thm}\Label{main}
When $P$ and $Q$ are non-uniform distributions,
the following equation holds:
\begin{eqnarray}
F_{P, Q,s_2}(t_2)
=\sqrt{1-Z_{C_{P,Q},\frac{s_2}{\sqrt{V(P)}}}(t_2D_{P,Q})}\Label{equ}
\end{eqnarray}
\end{thm}
To obtain Theorem \ref{main},
it is enough to show the direct part 
\begin{eqnarray}
F^{\cal D}_{P, Q,s_2}(t_2)
&\ge&\sqrt{1-Z_{C_{P,Q},\frac{s_2}{\sqrt{V(P)}}}(t_2D_{P,Q})},\Label{maxim1}
\end{eqnarray}
and the converse part 
\begin{eqnarray}
F^{\cal M}_{P, Q,s_2}(t_2)
&\le&\sqrt{1-Z_{C_{P,Q},\frac{s_2}{\sqrt{V(P)}}}(t_2D_{P,Q})}\Label{maxim2}
\end{eqnarray}
by (\ref{fidelity-ineq2}).
In particular, to prove the direct part (\ref{maxim1}),
it is enough to show the following lemma.
\begin{lem}\Label{separation} 
Let $\epsilon>0$.
For a non-uniform probability distribution $P$ on a finite set, 
there exists a sequence of maps $f_n:\calX^n \to \{0,1\}^{{H(P)n+s_2\sqrt{n}}} $ 
such that 
\begin{eqnarray}
\liminf_{n\to\infty} F(W_{f_n}(P^{n}), U_2^{H(P)n + s_2\sqrt{n}}) 
&\ge& \limsup_{n\to\infty}F^{\calD}(P^{n}\to U_2^{H(P)n + s_2\sqrt{n}}) - \epsilon.
\Label{sepopt'}
\end{eqnarray}
Moreover, 
for  two non-uniform probability distributions $P$ and $Q$ on finite sets, 
there exists a sequence of maps $f'_n: \{0,1\}^{{H(P)n+s_2\sqrt{n}}} \to \calY^{\frac{H(P)}{H(Q)}n + t_2\sqrt{n}}$
such that 
\begin{eqnarray}
\liminf_{n\to\infty} F(W_{f'_n}\circ W_{f_n}(P^{n}), Q^{\frac{H(P)}{H(Q)}n + t_2\sqrt{n}})
&\ge&\sqrt{1-Z_{C_{P,Q},\frac{s_2}{\sqrt{V(P)}}}(t_2D_{P,Q})} -\epsilon.
\Label{sepopt}
\end{eqnarray}
\end{lem}

The inequality (\ref{sepopt'}) shows that the conversion  $W_{f_n}$ is  almost optimal as a uniform random number generation.
Combining Lemma \ref{separation} with Theorem \ref{main},
such a conversion $W_{f_n}$ is almost optimal also as a random number compression to the storage.
Moreover, since $f_n$ does not depend on the target distribution $Q$,
the compression $W_{f_n}$ to the storage is universal with respect to the choice of the target distribution $Q$.
%
We prove Thoerem \ref{main} by showing Lemma \ref{separation} and (\ref{maxim2}) in Subsections\ref{main.app} and \ref{main2.app}.

Then we obtain the second-order rate region by Theorems \ref{general-2reg} and \ref{main}.
Moreover, since the explicit value of the generalized Rayleigh-normal distribution function in (\ref{equ}) is given in Theorem \ref{Zform}, 
we can determine the concrete form of the second-order rate region.
The second-order rate region is illustrated as Figs. \ref{fig-case3} and  \ref{fig-case2} for $C_{P,Q}<1$ and $C_{P,Q}\ge1$, respectively.

When $C_{P,Q}<1$,
there is no semi-admissible rate pair and the boundary of the rate region represents the set of admissible rate pairs.
When $C_{P,Q}\ge1$,
the straight line in the boundary represents semi-admissible rate pairs from Lemma \ref{simulate-2nd}
and the curved line does admissible rate pairs.

\begin{figure}[t]
\begin{center}
 \hspace*{0em}\includegraphics[width=80mm, height=55mm]{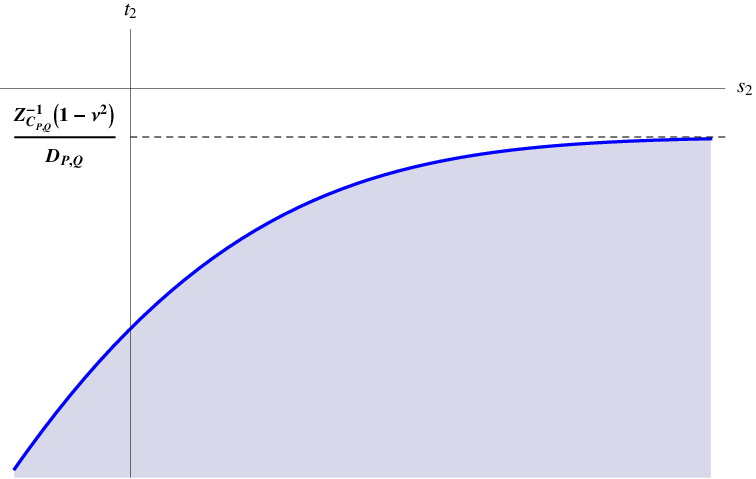}
 \end{center}
 \caption{
The second-order rate region ${\cal R}^{2,{\cal D}}_{P,Q}(s_1,t_1,\nu)$ and ${\cal R}^{2,{\cal M}}_{P,Q}(s_1,t_1,\nu)$ when $(s_1,t_1)$ is an admissible first-order rate pair and both $P$ and $Q$ are uniform with $C_{P,Q}<1$. 
}
 \Label{fig-case3}
\end{figure}

\begin{figure}[t]
\begin{center}
 \hspace*{0em}\includegraphics[width=80mm, height=55mm]{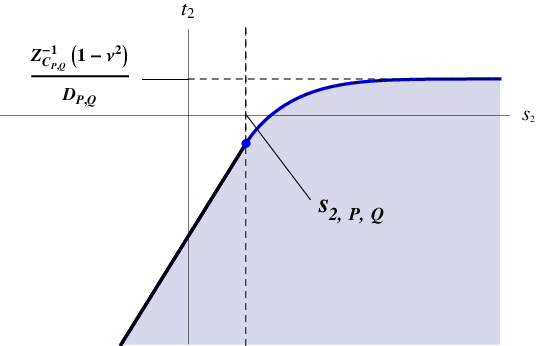}
 \end{center}
 \caption{
The second-order rate region ${\cal R}^{2,{\cal D}}_{P,Q}(s_1,t_1,\nu)$ and ${\cal R}^{2,{\cal M}}_{P,Q}(s_1,t_1,\nu)$ when $(s_1,t_1)$ is an admissible first-order rate pair and both $P$ and $Q$ are uniform with $C_{P,Q}\ge1$. 
The boundary of the region is straight line on the left side of a threshold value $s_{2,P,Q}$.
In particular, $\frac{Z_{1}^{-1}(1-\nu^2)}{D_{P,Q}}=\frac{\sqrt{-8V(P)\ln\nu}}{H(Q)}$ and $s_{2,P,Q}=\sqrt{V(P)}\Phi^{-1}(\nu^2)$ when $C_{P,Q}=1$.
}
 \Label{fig-case2}
\end{figure}

When either $P$ or $Q$ is the uniform distribution $U_l$ with size $l$, 
the asymptotics is reduced to the problem of resolvability or intrinsic randomness, and the second-order rate regions are obtained as follows.

\begin{thm}\Label{region2-uni1}
When $P=U_l$ and $Q$ is a non-uniform distribution,
the following equation holds:
\begin{eqnarray}
&&\hspace{-2em}F_{U_l,Q,s_2}(t_2)=\sqrt{\Phi\left(\sqrt{\frac{H(Q)}{V(Q)\log l}}(\min\{s_2,0\}-H(Q)t_2)\right)}.
\Label{2nd-dil}
\end{eqnarray}
In particular, the above value is described by the limit of the generalized Rayleigh-normal distribution function as follows: 
\begin{eqnarray}
F_{U_l,Q,s_2}(t_2)
=\lim_{P\to U_l}\sqrt{1-Z_{C_{P,Q},\frac{s_2}{\sqrt{V(P)}}}(t_2D_{P,Q})}.\nonumber
\end{eqnarray}
\end{thm}

We give the proof of Lemma \ref{region2-uni1} in Section \ref{region2-uni1.app}.
When $P=U_l$ and $(s_1,t_1)$ is the admissible rate pair $(\log l,\frac{\log l}{H(Q)})$, 
from Theorem \ref{general-2reg} and Lemma \ref{region2-uni1},
the second-order rate region is given by
\begin{eqnarray}
&&\hspace{-1em}
{\cal R}^{2}_{U_l,Q}\left(\nu\right)
\nonumber
\\&\hspace{-2em}=&
\hspace{-1em}\left\{\left(s_2,t_2\right)\bigg|
t_2 \le \frac{\min\{s_2,0\}}{H(Q)} -\sqrt{\frac{V(Q)\log l}{H(Q)^3}} \Phi^{-1}(\nu^2)
\right\}.
\label{reg2-uni1}
\end{eqnarray}
%
%
The second-order rate region is illustrated as Fig. \ref{2nd-reg2-uni1}.
Then the line with the slope $H(Q)^{-1}$ is the set of semi-admissible rate pairs 
from Lemma \ref{simulate-2nd}
and the extreme point is the unique admissible pair.
%
\begin{figure}[t]
 \begin{center}
 \hspace*{0em}\includegraphics[width=80mm, height=55mm]{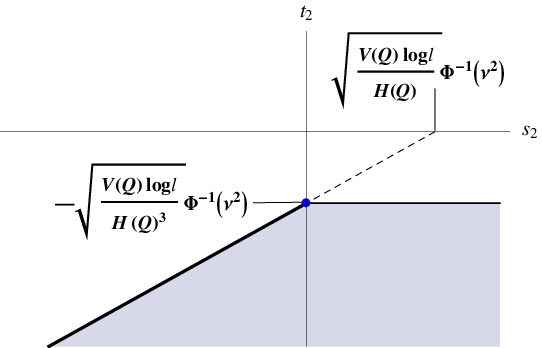}
 \end{center}
 \caption{
The second-order rate region ${\cal R}^{2,{\cal D}}_{U_l,Q}(s_1,t_1,\nu)$ and ${\cal R}^{2,{\cal M}}_{U_l,Q}(s_1,t_1,\nu)$ when $(s_1,t_1)$ is an admissible first-order rate pair.
}
 \Label{2nd-reg2-uni1}
\end{figure}

\begin{thm}\Label{region2-uni2}
When $P$ is a non-uniform distribution and $Q=U_l$,
the following equation holds:
\begin{eqnarray}
F_{P,U_l,s_2}(t_2)
=
\left\{
\begin{array}{cll}
\sqrt{\Phi\left(\frac{-\log l}{\sqrt{V(P)}}t_2\right)}&{\it if}&(\log l) t_2\le s_2
\\
0&{\it if}&otherwise.
\end{array}
\right.
\label{2nd-con}
\end{eqnarray}
In particular, the above value is described by the limit of the generalized Rayleigh-normal distribution function as follows: 
\begin{eqnarray}
F_{P,U_l,s_2}(t_2)
=\sqrt{1-Z_{0,\frac{s_2}{\sqrt{V(P)}}}(t_2D_{P,U_l})},\nonumber
\end{eqnarray}
\end{thm}
where $Z_{0,s}$ was defined in (\ref{Z0}).

We give the proof of Lemma \ref{region2-uni2} in Section \ref{region2-uni2.app}.
When $Q=U_l$ and $(s_1,t_1)$ is the admissible rate pair $(H(P),\frac{H(P)}{\log l})$, 
from Theorem \ref{general-2reg} and Lemma \ref{region2-uni2},
the second-order rate region is given by
\begin{eqnarray}
&&\hspace{-1em}
{\cal R}^{2,{\cal D}}_{P,U_l}\left(\nu\right)
=
{\cal R}^{2,{\cal M}}_{P,U_l}\left(\nu\right)
\nonumber
\\&\hspace{-2em}=&
\hspace{-1em}\left\{\left(s_2,t_2\right)\bigg|
t_2 \le\frac{\min\{s_2,-\sqrt{V(P)}\Phi^{-1}(\nu^2)\}}{\log l}
\right\}.
\label{reg2-uni2}
\end{eqnarray}
%
The second-order rate region is illustrated as Fig. \ref{2nd-reg2-uni2}.
Then the line with the slope $H(Q)^{-1}=(\log l)^{-1}$ is the set of semi-admissible rate pairs 
from Lemma \ref{simulate-2nd}
and the extreme point is the unique admissible pair.

\begin{figure}[t]
 \begin{center}
 \hspace*{0em}\includegraphics[width=80mm, height=55mm]{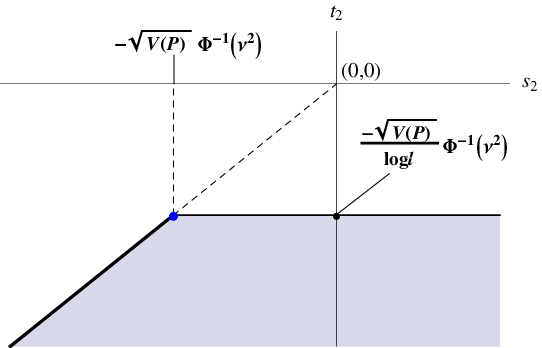}
 \end{center}
 \caption{
The second-order rate region ${\cal R}^{2,{\cal D}}_{P,U_l}(s_1,t_1,\nu)$ and ${\cal R}^{2,{\cal M}}_{P,U_l}(s_1,t_1,\nu)$ when $(s_1,t_1)$ is an admissible first-order rate pair.
}
 \Label{2nd-reg2-uni2}
\end{figure}

\section{Related topics}\Label{sec:RCR}

\subsection{Random Number Compression}~

As a special case of RNC via restricted storage, 
we consider random number compression.
Here, our random number compression is given as a two-stage random number conversion, namely, the combination of compression conversion and decompression conversion.
Compression conversion maps an initial random number 
subject to a probability distribution $P^n$ 
to another random number on a storage with size of ${H(P)n+s_2\sqrt{n}}$ bits.
After that, decompression conversion maps the random number on the storage to a random number approximately subject to the initial probability distribution $P^n$.
The process corresponds to RNC via restricted storage when $Q=P$ and $t_2=0$.
Then, the optimal accuracy of random number compression is given by Theorems \ref{Zform} and \ref{main} as follows:
\begin{eqnarray}
F_{P,P,s_2}(0)
=\sqrt{\Phi\left(\frac{s_2}{\sqrt{V(P)}}\right)}.
\end{eqnarray}
Thus, we obtain the following corollary.
\begin{cor}\Label{=}
Let $P$ be an arbitrary non-uniform probability distribution on a finite set.
For random number compression,
the minimum size of storage to guarantee an accuracy $\nu$ is represented by ${H(P)n+\sqrt{V(P)\Phi^{-1}(\nu^2)}\sqrt{n}}$.
\end{cor}
%

Note that the purpose of the random number compression is not to recover the initial random number itself but to regenerate a random number subject to the same distribution $P^n$
and the process itself differs from the data compression.
However, Corollary \ref{=} shows that the minimum size of storage in data compression has the same form with that of random number compression (see the equation (1) in \cite{Hay08}).

\subsection{Relation with Conventional RNC}

We have treated RNC via restricted storage.
On the other hand, in the previous paper \cite{KH13-2}, we treated random number conversion without restriction of storage.
Here, it is expected that the rate of the generated copies of the target distribution
approaches to the conversion rate in the previous paper as the size of storage gets larger.
In the following, we discuss this relation in terms of the asymptotic maximum fidelity of RNC.

When the first-order rate of the size of storage is the entropy of the source distribution,
the asymptotic maximal fidelity in RNC with restricted storage is given as
\begin{eqnarray}\Label{prelimit}
F_{P,Q,s_2}(t_2)
:=
F^{\cal D}_{P,Q,s_2}(t_2)
=
F^{\cal M}_{P,Q,s_2}(t_2).
\end{eqnarray}
On the other hand,
the asymptotic maximal fidelity in RNC without restricted storage is given as follows shown in \cite{KH13-2}
\begin{eqnarray}\Label{postlimit}
F_{P,Q}(t_2)
&:=&\lim_{n\to\infty}F^{\cal D}(P^n\to Q^{\frac{H(P)}{H(Q)}n+t_2\sqrt{n}})\nonumber\\
&=&
\lim_{n\to\infty}F^{\cal M}(P^n\to Q^{\frac{H(P)}{H(Q)}n+t_2\sqrt{n}}).
\end{eqnarray}
Fig. \ref{loss ratio} represents the graph of the ratio $F_{P,Q,s_2}(t_2)/F_{P,Q}(t_2)$  with respect to $s_2\in\R$ when $C_{P,Q}=1$.
We can read off that the value of $F_{P,Q,s_2}(t_2)$ converges to that of $F_{P,Q}(t_2)$ for each $t_2\in\R$ when $s_2$ goes to infinity
and
the existence of storage does not affect the accuracy (i.e. the asymptotic maximum fidelity) of RNC via restricted storage so much as long as the second-order rate is large enough 
even when the first-order rate strictly achieves the optimal value. 
In particular,  when $s_2$ tends to infinity, the second order asymptotic expansion in Theorem \ref{expansion} recovers Theorem $3$ of \cite{KH13-2} for RNC without restricted storage
by Theorems \ref{region2-non-admissible}, \ref{main}, \ref{region2-uni1}, \ref{region2-uni2} and (\ref{limit}).
\begin{figure}[t]
 \begin{center}
 \hspace*{0em}\includegraphics[width=80mm, height=55mm]{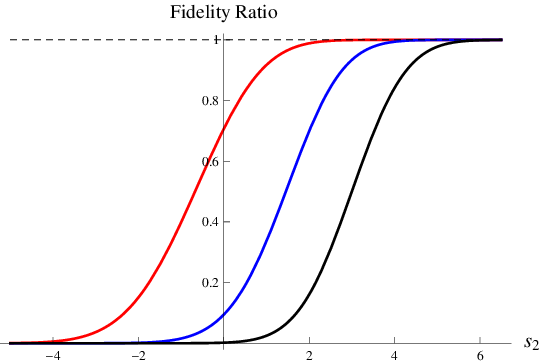}
 \end{center}
 \caption{
The graph of the ratio $\frac{F_{P,Q,s_2}(t_2)}{F_{P,Q}(t_2)}$ with respect to the second-order rate $s_2$ of storage when $C_{P,Q}=V(P)=H(Q)=1$. 
The left red line shows the case when $t_2\le0$.
The middle blue and the right black lines show the cases when $t_2=-3$ and $t_2=-6$.
In particular, the ratio of fidelities does not depend on $t_2$ if $t_2\le0$.
}
 \Label{loss ratio}
\end{figure}
%

\section{Application to Quantum Information Theory}
\Label{sec:AQIT}~

In this section, we apply the results of RNC via restricted storage for quantum information theory.
%
\begin{figure*}[t]
 \begin{center}
 \hspace*{-8em}\includegraphics[width=120mm, height=35mm]{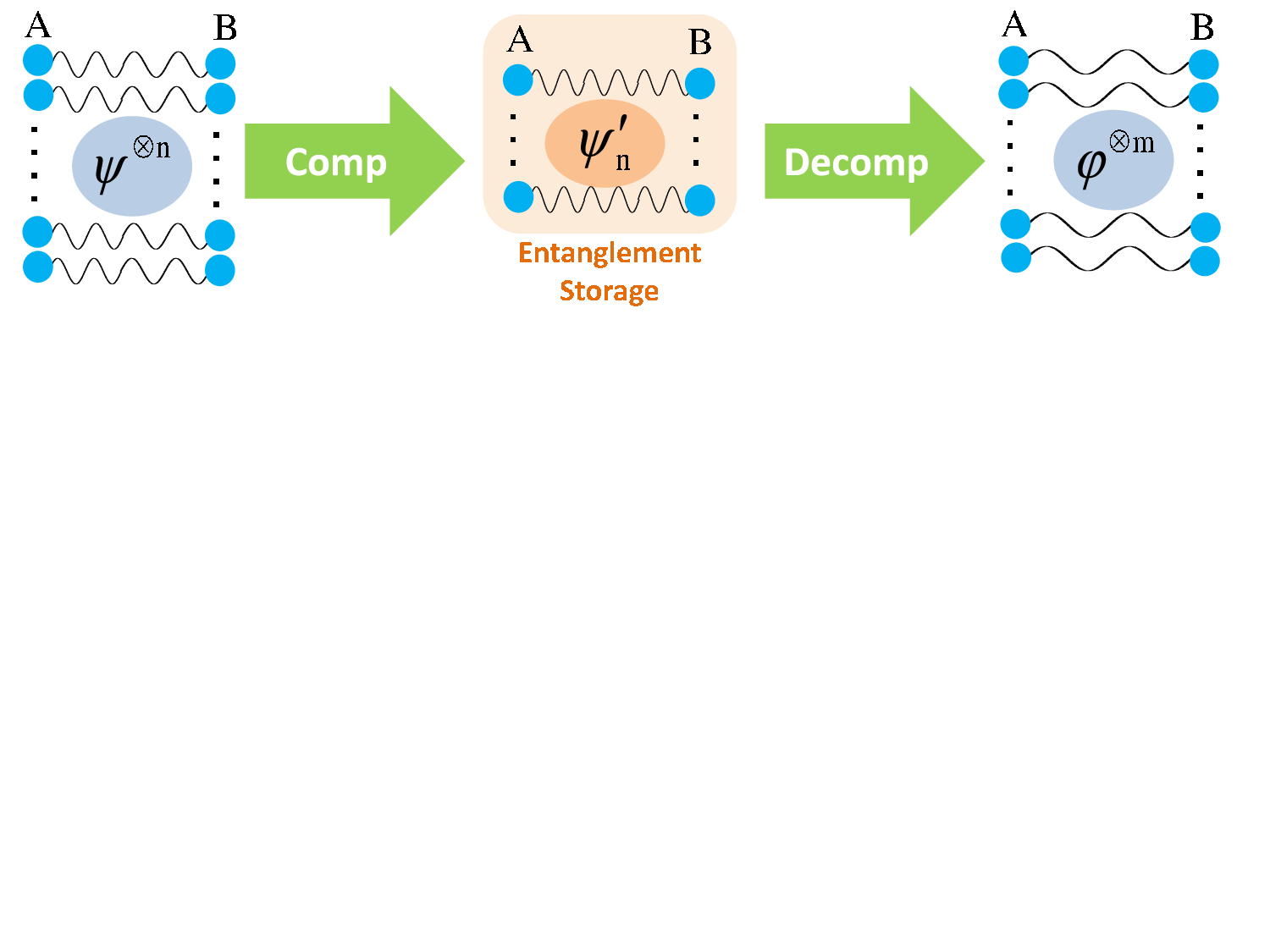}
 \end{center}
 \caption{
Process of entanglement compression by LOCC.
}
 \Label{compression}
\end{figure*}

\subsection{LOCC Conversion via Restricted Storage}\Label{LCRS.app}

When two distant parties perform some quantum protocol using a specific suitable entangled state (e.g. quantum teleportation, superdense coding, channel estimation),
those parties need to prepare the desired entangled state.
To do so,
the parties share some initial entangled states which are not necessarily the desired entangled states by a quantum communication channel,
and then,
they generate the desired entangled states by performing LOCC for given entangled states.
However, 
a quantum protocol which is performed may not be determined at the time of sharing of initial entangled states.
Then, 
it is desirable to store entangled states in some storage
and,
after the determination of a quantum protocol which is performed,
to be able to convert the stored states to desired states depending on the quantum protocol. 
To model the situation, we consider the following two-step process.
In the first part, an initial state is converted into the storage by LOCC.
In the second part, the converted state is converted again to a target state by LOCC.
We call such a process LOCC conversion via {\it entanglement storage}.
In the following, 
let us represent the quantum system of entanglement storage by ${\cal H}_{qubit}^{\otimes N}$ where ${\cal H}_{qubit}:=\C^2\otimes\C^2$,
and
we analyze the asymptotic behavior of LOCC conversion via entanglement storage
when an initial state and a target state are i.i.d. and pure.

We consider the maximam recovery number by LOCC:
\begin{eqnarray*}
&&L_n^{\cal Q}(\psi, \ph|\nu, N)\nonumber\\
&&:=\hspace{0em}\max
\left\{L\in\N\Bigg|
\begin{array}{l}
\exists\Gamma:\mathcal{S}(\mathcal{H}^{\otimes n})\to\mathcal{S}({\cal H}_{qubit}^{\otimes N}) : \text{LOCC},\\ 
\exists\Gamma':\mathcal{S}({\cal H}_{qubit}^{\otimes N}) \to\mathcal{S}(\mathcal{H}'^{\otimes L}) : \text{LOCC},\\
F(\Gamma' \circ\Gamma(\psi^{\otimes n}), \ph^{\otimes L})\ge\nu.
\end{array}
\right\}.
\end{eqnarray*}
Here, note that the converted state in the entanglement storage is not necessarily pure, 
and thus, two-step process of LOCCs may not be simply represented by majorization conversion for the Schmidt coefficients of an initial state in general.
Therefore, the results for majorization conversion of probability distributions can not be directly applied for the maximam recovery number by LOCC from its definition yet.
To analyse the maximam recovery number,
we introduce the maximum accuracy of LOCC conversion via entanglement storage as follows:
\begin{eqnarray*}
&&\hspace{-1.5em}F^{\cal Q}(\psi\to \ph|N)\\
&&\hspace{-1.7em}:=\sup_{\Gamma,\Gamma'}
\left\{F(\Gamma'\circ \Gamma(\psi), \ph)\Bigg|
\begin{array}{l}
\Gamma:\mathcal{S}(\mathcal{H}) \to \mathcal{S}({\cal H}_{qubit}^{\otimes N}): \text{LOCC},\\ 
\Gamma':\mathcal{S}({\cal H}_{qubit}^{\otimes N}) \to \mathcal{S}(\mathcal{H}'): \text{LOCC}
\end{array}
\right\}
\end{eqnarray*}
where 
$\psi$ and $\ph$ are quantum states on bipartite systems ${\cal H}$ and ${\cal H}'$ respectively, 
${\cal S}(\cal H)$ is the set of all quantum states on ${\cal H}$.
%
Then, we obtain 
\begin{eqnarray}
L_n^{\cal Q}(\psi, \ph|\nu, N)
=\max\{L\in\N| F^{\cal Q}(\psi\to \ph^{\otimes L}|N)\ge\nu\}
\Label{Q1}
\end{eqnarray}
by the definition.
Moreover,
the following lemma holds for the squared Schmidt coefficients $P_{\psi}$ and $P_{\ph}$ of $\psi$ and $\ph$.
\begin{lem}\Label{q-to-c}
\begin{eqnarray}
F^{\cal Q}(\psi\to \ph|N)
=F^{\cal M}(P_{\psi} \to P_{\ph}|N)
\end{eqnarray}
\end{lem}
We give the proof of Lemma \ref{q-to-c} in Section \ref{q-to-c.app}.
Here, as stated above, a converted state by LOCC in storage is not necessarily a pure state.
However, in the optimal process, 
we can assume that the converted state by LOCC in storage is pure from the proof of Lemma \ref{q-to-c}.
From (\ref{M1}), (\ref{Q1}) and Lemma \ref{q-to-c}, 
the following proposition holds.
\begin{prp}\Label{expansion2}
\begin{eqnarray*}
L_n^{\cal Q}(\psi, \ph|\nu, N)
=L_n^{\cal M}(P_{\psi}, P_{\ph}|\nu, N)
\nonumber\label{q-exp}
\end{eqnarray*}
\end{prp}
In particular, the asymptotic expansion of $L_n^{\cal Q}$ is obtained by Theorem \ref{expansion}.

Next, let us consider the rate regions of LOCC conversion via entanglement storage.
For simplicity, 
we employ the following abbreviate notation:
\begin{eqnarray*}
F^{\cal Q}_{\psi,\ph,s_1}(t_1)
:=
\liminf_{n\to\infty} F^{\cal Q}\left(\psi^{\otimes n}\to \ph^{\otimes t_1n\sqrt{n}}|s_1n\right).
\end{eqnarray*}
In order to treat the asymptotic relation between the second-order rates of storage and target entangled state, 
Then we define the second-order rate region as
%
\begin{eqnarray}
&&{\cal R}^{1,{\cal Q}}_{\psi,\ph}(\nu)
:=
\left\{\left(s_1,t_1\right)\bigg|
F^{\cal Q}_{\psi,\ph,s_1}(t_1)\ge \nu \right\}.
\label{reg2q}\nonumber
\end{eqnarray}
When $S_{\psi}$ is the von Neumann entropy of the partial density matrix of $\psi$,
Lemma \ref{q-to-c} and Theorem \ref{region1} imply the following theorem about first-order rate region.
\begin{prp} \Label{Qreg1}
Let $\psi$ and $\ph$ be pure entangled states on finite dimensional bipartite quantum systems.
For $0<s_1\le S_{\psi}$, $s_2\in\R$ and $\nu\in(0,1)$,
\begin{eqnarray}
&&{\cal R}^{1,{\cal Q}}_{\psi,\ph}\left(\nu\right)
={\cal R}^{1,{\cal M}}_{P_{\psi},P_{\ph}}\left(\nu\right).
\nonumber
\end{eqnarray}
\end{prp}
Similarly,
we employ the following abbreviate notation:
\begin{eqnarray*}
&&F^{\cal Q}_{\psi,\ph,s_1,t_1,s_2}(t_2)\\
&&:=
\liminf_{n\to\infty} F^{\cal Q}\left(\psi^{\otimes n}\to \ph^{\otimes t_1n+t_2\sqrt{n}}|s_1n+s_2\sqrt{n}\right).
\end{eqnarray*}
Then we define the second-order rate region as 
%
\begin{eqnarray}
&&{\cal R}^{2,{\cal Q}}_{\psi,\ph}(s_1,t_1,\nu)
:=
\left\{\left(s_2,t_2\right)\bigg|
F^{\cal Q}_{\psi,\ph,s_1,t_1,s_2}(t_2)\ge \nu \right\}.
\label{reg2q}\nonumber
\end{eqnarray}
%
Then, Lemma \ref{q-to-c} and Theorem \ref{general-2reg} imply the following theorem about the second-order rate region.
\begin{prp} \Label{Qreg2}
Let $\psi$ and $\ph$ be pure entangled states on finite dimensional bipartite quantum systems.
For $0<s_1\le S_{\psi}$, $s_2\in\R$ and $\nu\in(0,1)$,
\begin{eqnarray}
&&{\cal R}^{2,{\cal Q}}_{\psi,\ph}\left(s_1,\frac{s_1}{S_{\ph}},\nu\right)
={\cal R}^{2,{\cal M}}_{P_{\psi},P_{\ph}}\left(s_1,\frac{s_1}{H(P_{\ph})},\nu\right).
\nonumber
\end{eqnarray}
\end{prp}
Therefore, the second-order rate region is obtained by Theorem \ref{general-2reg}.
and is especially described by the generalized Rayleigh-normal distribution function  at the semi-admissible rate pairs by Theorem \ref{main}.

\subsection{Entangled State Compression by LOCC}\Label{ESCL.app}

When an initial state $\ph$ equals a target state $\psi$, 
the LOCC conversion via restricted entanglement storage is regarded as a compression process for entangled states.
There already exist some studies about LOCC compression for entangled states.
In particular, Schumacher \cite{Sch95} derived the optimal first-order rate of LOCC compression for entangled states in the framework of the first-order asymptotics.
Here, we consider the LOCC compression in the framework of the second-order asymptotics
and derive some observations which essentially can not be obtained from the first-order asymptotics.
When the size of storage has the optimal first-order compression rate $S_{\psi}$ and the second-order rate $s_2$,
the difference between the numbers of the initial and recovered copies is given as 
\begin{eqnarray}
n-L_n(\psi, \psi|\nu, s_2)
\cong -F_{P_{\psi},P_{\psi},s_2}^{-1}(\nu)\sqrt{n}, \Label{loss2}
\end{eqnarray}
where the concrete form of $F_{P_{\psi},P_{\psi},s_2}$ was given in Themrem \ref{main}.
The formula (\ref{loss2}) relates with the irreversibility of entanglement concentration \cite{KH13}. 
That is, when $s_2$ is smaller than $\sqrt{V(P_{\psi})}\Phi^{-1}(\nu^2)$ for a required accuracy $\nu$,
the right-hand side in (\ref{loss2}) is positive from Corollary \ref{=} and represents the loss which inevitably occurs even in the optimal compression process.
Moreover, from Lemma \ref{opt con} and the proof of Lemma \ref{q-to-c}, 
the LOCC conversion in the optimal compression coincides with LOCC conversion used in the optimal entanglement concentration. 
In addition, (\ref{loss2}) also relates with LOCC cloning \cite{KH13-2}.
That is, when $s_2$ is larger than $\sqrt{V(P_{\psi})}\Phi^{-1}(\nu^2)$,
the right-hand side in (\ref{loss2}) is negative from Corollary \ref{=} and it represents that the number of copies of the recovered state after the compression process exceeds that of the initial state under the accuracy constraint.
While we argued about approximate LOCC cloning without entanglement storage (or with infinite storage) in \cite{KH13-2},
the above fact says that approximate LOCC cloning can be realized even when there is entanglement storage with the tight first-order rate $S_{\psi}$ as long as the second-order rate of the size of storage is large enough.

\section{Proofs of Theorems, Propositions and Lemmas}
\Label{sec:Proof}

\subsection{Proof of Lemma\ref{sol2}}\Label{sol2.app}
The existence of the unique solution of the equation (\ref{threshold0})  is equivalent to the existence of the unique zero point of the function
\begin{eqnarray}
f(x):=(\Phi_{\mu,v}(s)-\Phi_{\mu,v}(x))-(1-\Phi(x))\frac{\phi_{\mu,v}(x)}{\phi(x)}.
\Label{f}
\end{eqnarray}
Since
\begin{eqnarray}
\frac{df}{dx}(x)
= -\frac{1-v}{v}\left(\frac{\mu}{1-v}-x\right)\frac{\phi_{\mu,v}(x)}{\phi(x)}(1-\Phi(x)) 
\Label{differential}
\end{eqnarray}
 and $0<v<1$,
the function $f$ is strictly monotonically decreasing when $x<\frac{\mu}{1-v}$ and is strictly monotonically increasing when $x>\frac{\mu}{1-v}$.
Since
\begin{eqnarray}
\lim_{x\to-\infty}f(x)&=&\Phi_{\mu,v}(s)>0, \Label{f>0}\\
\lim_{x\to\infty}f(x)&=&\Phi_{\mu,v}(s)-1<0,  \Label{f<0}
\end{eqnarray}
the function $f$ has the unique zero point $\beta_{\mu,v,s}<\frac{\mu}{1-v}$ due to the intermediate value theorem.
In addition, $\beta_{\mu,v,s}<s$  holds 
because the left-hand side of (\ref{threshold0}) is negative for any $x>s$ although the right-hand side is always positive.
\endproof

\subsection{Proof of Lemma \ref{sol0}}\Label{sol0.app}
The existence of the unique solution of the equation (\ref{threshold0})  is equivalent to the existence of the unique zero point of the function (\ref{f}).
Since
\begin{eqnarray}
\frac{df}{dx}=-\mu\frac{\phi_{\mu,1}}{\phi}(1-\Phi), 
\end{eqnarray}
the function $f$ is strictly monotonically decreasing over $\R$ because of $\mu>0$.
Since $f$ satisfies (\ref{f>0}) and (\ref{f<0}),
the function $f$ has the unique zero point $\beta_{\mu,v,s}$ due to the intermediate value theorem.
In addition, $\beta_{\mu,v,s}<s$  holds 
because the left-hand side of (\ref{threshold0}) is negative for any $x>s$ although the right-hand side is always positive.
\endproof

\subsection{Proof of Lemma \ref{sol1}}\Label{sol1.app}
There exists the unique solution $\alpha_{\mu,v}$ of (\ref{threshold1'}) with respect to $x$ in Lemma $3$ of \cite{KH13-2}.
Next, we show that there are two solutions $\beta'_{\mu,v}<\beta_{\mu,v}$ for the equation (\ref{threshold0}) and $\beta_{\mu,v}$ satisfies $\beta_{\mu,v}>\alpha_{\mu,v}$ under the condition $s>\Phi_{\mu,v}^{-1}\left(\frac{\Phi_{\mu,v}(\alpha_{\mu,v})}{\Phi(\alpha_{\mu,v})}\right)$.
Here, the existence of the solutions is equivalent to the existence of the zero points of the function (\ref{f}).
Since $f$ satisfies (\ref{differential}) and $v>1$,
the function $f$ is strictly monotonically increasing when $x<\frac{\mu}{1-v}$ and is strictly monotonically decreasing $x>\frac{\mu}{1-v}$.
Here, by the definition of $\alpha_{\mu,v}$ and the condition $s>\Phi_{\mu,v}^{-1}\left(\frac{\Phi_{\mu,v}(\alpha_{\mu,v})}{\Phi(\alpha_{\mu,v})}\right)$, we obtain the following inequality:
\begin{eqnarray}
f(\alpha_{\mu,v})
&=&\Phi_{\mu,v}(s)-\Phi_{\mu,v}(\alpha_{\mu,v})\nonumber
\\&&-(1-\Phi(\alpha_{\mu,v}))\frac{\Phi_{\mu,v}(\alpha_{\mu,v})}{\Phi(\alpha_{\mu,v})}>0.
\end{eqnarray}
Moreover, since
\begin{eqnarray}
\displaystyle\lim_{x\to-\infty}f(x)&=&-\infty, 
\\\displaystyle\lim_{x\to\infty}f(x)
&\le&\lim _{x\to\infty}(\Phi_{\mu,v}(s)-\Phi_{\mu,v}(x))\nonumber
\\&=&\Phi_{\mu,v}(s)-1<0, 
\end{eqnarray}
the function $f$ has two zero points $\beta'_{\mu,v}<\beta_{\mu,v}$ and $\beta_{\mu,v}>\alpha_{\mu,v}$ due to the intermediate value theorem.
\endproof

\subsection{Lemmas for Direct Part of Theorem \ref{Zform}}\Label{Zdir.app}

The following lemma is given as Lemma $22$ in \cite{KH13-2}.
\begin{lem}\Label{monotone}
The ratio $\frac{\phi(x)}{\phi_{\mu,v}(x)}$ is strictly monotonically decreasing only on the interval $\mathcal{I}_{\mu,v}$ defined by 
\begin{eqnarray}
\mathcal{I}_{\mu,v}
=\left\{
\begin{array}{ccl}
\R&\textit{if}& v=1~ {\rm and}~ \mu>0\\
 \emptyset &\textit{if}& v=1~ {\rm and}~ \mu\le 0\\
(\frac{\mu}{1-v},\infty)&\textit{if}& v>1\\
(-\infty,\frac{\mu}{1-v})&\textit{if}& v<1,
\end{array}
\right.
\Label{area}
\end{eqnarray}
where $ \emptyset$ is the empty set.
\end{lem}

Using $\beta_{\mu,v,s}$ and $\alpha_{\mu,v}$ in Lemmas \ref{sol2}, \ref{sol0} and \ref{sol1},
we define a function $A_{\mu,v,s}:\R\to[0,1]$ which has different forms depending on $v>0$ as follows.
When $v<1$,
\begin{eqnarray*}
&&A_{\mu,v,s}(x)\\
&=&\left\{
\begin{array}{ll}
\Phi(x) & \hspace{-10em} \textit{if}~x\le \beta_{\mu,v,s} \\
\Phi(\beta_{\mu,v,s})
+\frac{1-\Phi\left(\beta_{\mu,v,s}\right)}{\Phi_{\mu,v}(s)-\Phi_{\mu,v}(\beta_{\mu,v,s})}
(\Phi_{\mu,v}(x)-\Phi_{\mu,v}(\beta_{\mu,v,s})) & \\  
& \hspace{-10em} \textit{if}~\beta_{\mu,v,s}\le x\le s\\
1 & \hspace{-10em} \textit{if}~s\le x, \\
\end{array}
\right.
\end{eqnarray*}
When $v=1$,
\begin{eqnarray*}
&&A_{\mu,1,s}(x)\\
&=&\left\{
\begin{array}{ll}
\frac{\Phi_{\mu,1}(x)}{\Phi_{\mu,1}(s)} & \hspace{-13em}\textit{if}~\mu\le0, x\le s \\
\Phi(x) & \hspace{-13em} \textit{if}~\mu>0, x\le \beta_{\mu,v,s} \\
\Phi(\beta_{\mu,v,s}) 
+ \frac{1-\Phi\left(\beta_{\mu,1,s}\right)}{\Phi_{\mu,1}(s)-\Phi_{\mu,1}(\beta_{\mu,1,s})}
(\Phi_{\mu,1}(x)-\Phi_{\mu,1}(\beta_{\mu,v,s})) & \\
& \hspace{-13em} \textit{if}~\mu>0, \beta_{\mu,v,s}\le x\le s\\
1& \hspace{-13em} \textit{if}~s\le x.
\end{array}
\right.
\end{eqnarray*}
When $v>1$ and $s\le \Phi_{\mu,v}^{-1}\left(\frac{\Phi_{\mu,v}(\alpha_{\mu,v})}{\Phi(\alpha_{\mu,v})}\right)$, 
\begin{eqnarray}
&&A_{\mu,v,s}(x)
=\left\{
\begin{array}{ll}
\frac{\Phi_{\mu,v}(x)}{\Phi_{\mu,v}(s)} & \textit{if}~x\le s \\
1 & \textit{if}~s\le x .
\end{array}
\right.
\end{eqnarray}
Wthen $v>1$ and $s\ge \Phi_{\mu,v}^{-1}\left(\frac{\Phi_{\mu,v}(\alpha_{\mu,v})}{\Phi(\alpha_{\mu,v})}\right)$, 
\begin{eqnarray*}
&&A_{\mu,v,s}(x)\\
&=&\left\{
\begin{array}{ll}
\frac{\Phi(\alpha_{\mu,v})}{\Phi_{\mu,v}(\alpha_{\mu,v})}\Phi_{\mu,v}(x) & \hspace{-11em}\textit{if}~x\le \alpha_{\mu,v} \nonumber\\
\Phi(x) & \hspace{-11em}\textit{if}~\alpha_{\mu,v}\le x\le \beta_{\mu,v,s}\\
\Phi(\beta_{\mu,v,s})
+\frac{1-\Phi\left(\beta_{\mu,v,s}\right)}{\Phi_{\mu,v}(s)-\Phi_{\mu,v}(\beta_{\mu,v,s})}
(\Phi_{\mu,v}(x)-\Phi_{\mu,v}(\beta_{\mu,v,s}))&\\
 & \hspace{-11em}\textit{if}~\beta_{\mu,v,s}\le x\le s\\
1 & \hspace{-11em}\textit{if}~s\le x. \\
\end{array}
\right.
\end{eqnarray*}%

\begin{lem}\Label{lem.direct}
Suppose that  $\mu\in\R$ and $v>0$ satisfy
(i) $v<1$,
(ii) $v=1$ and $\mu>0$, 
or (iii) $v>1$ and $s> \Phi_{\mu,v}^{-1}\left(\frac{\Phi_{\mu,v}(\alpha_{\mu,v})}{\Phi(\alpha_{\mu,v})}\right)$.
For an arbitrary $\epsilon>0$,
there exist real numbers $b\le b'\le s$ which satisfy the following condition {\rm ($\star$)}:\\
\hspace{-1em}{\rm ($\star$)} There exist $a$ and $a'$ which satisfy the following three conditions:
\begin{eqnarray}
&&\hspace{-1.5em}~{\rm (I)}a\le b\le b'\le a' ,\label{I}\nonumber\\
&&\hspace{-1.5em}~{\rm (II)}\frac{\Phi(b)}{\Phi_{\mu,v}(b)}=\frac{\phi(a)}{\phi_{\mu,v}(a)} ~and\nonumber\\
&&\hspace{-1.5em}\hspace{2em}\frac{1-\Phi(b')}{\Phi_{\mu,v}(s)-\Phi_{\mu,v}(b')}=\frac{\phi(a')}{\phi_{\mu,v}(a')},\label{II}\\
&&\hspace{-1.5em}~{\rm (III)}\frac{\phi(x)}{\phi_{\mu,v}(x)}~is~monotonically~decreasing~on~(a,a').\label{III}\nonumber
\end{eqnarray}
Then such $b$ and $ b'$ satisfy the following inequality
\begin{eqnarray}
&&\sqrt{{\Phi}(b)} \sqrt{{\Phi}_{\mu,v}(b)}
+
\int_{b}^{b'}\sqrt{\phi(x)} \sqrt{\phi_{\mu,v}(x)}dx\nonumber\\
&& +\sqrt{1-{\Phi}(b')} \sqrt{{\Phi}_{\mu,v}(s)-{\Phi}_{\mu,v}(b')}\nonumber\\
&\le&{\cal F}\left(\frac{dA_{\mu,v}}{dx},\phi_{\mu,v}\right) +\epsilon.
\Label{lem.direct.ineq}
\end{eqnarray}
\end{lem}

\begin{proof}
First, we simultaneously treat the cases (i) $v<1$ and (ii) $v=1$ and $\mu>0$.
We take a constant $\lambda\in\R$ which satisfies $\lambda<\beta_{\mu,v}$ and $\sqrt{\Phi(\lambda)}\sqrt{\Phi_{\mu,v}(\lambda)}<\epsilon$.
We verify that $b=\lambda$ and $b'=\beta_{\mu,v}$ satisfy the condition ($\star$) in the following.
First, there exists a real number $a$ such that 
\begin{eqnarray}
\frac{\Phi(\lambda)}{\Phi_{\mu,v}(\lambda)}
=\frac{\Phi(\lambda) - \Phi(-\infty)}{\Phi_{\mu,v}(\lambda) - \Phi_{\mu,v}(-\infty)}
=\frac{\phi(a)}{\phi_{\mu,v}(a)}
\end{eqnarray}
and $a\le\lambda$ by the mean value theorem.
Moreover, since $\beta_{\mu,v}$ satisfies (\ref{threshold0}), $\beta_{\mu,v}$ can be taken as $a'=b'$.
Thus, the conditions (I) and (II) in ($\star$) hold.
Next, since $\frac{\phi(x)}{\phi_{\mu,v}(x)}$ is monotonically decreasing on $(\lambda,\beta_{\mu,v})$ from Lemma \ref{sol2} and Lemma \ref{monotone}, 
 the condition (III) in ($\star$) holds.
Therefore, $\lambda$ and $\beta_{\mu,v}$ satisfy the condition ($\star$).
Then the following holds: 
\begin{eqnarray}
&&\sqrt{\Phi(\lambda)}\sqrt{\Phi_{\mu,v}(\lambda)}
+\int_{\lambda}^{\beta_{\mu,v}}\sqrt{\phi(x)}\sqrt{\phi_{\mu,v}(x)}dx\nonumber\\
&&+\sqrt{1-\Phi(\beta_{\mu,v})}\sqrt{\Phi_{\mu,v}(s)-\Phi_{\mu,v}(\beta_{\mu,v})}\nonumber\\
&\le&\int_{-\infty}^{\beta_{\mu,v}}\sqrt{\phi(x)}\sqrt{\phi_{\mu,v}(x)}dx\nonumber\\
&&+\sqrt{1-\Phi(\beta_{\mu,v})}\sqrt{\Phi_{\mu,v}(s)-\Phi_{\mu,v}(\beta_{\mu,v})}
+\epsilon, \nonumber\\
&=&{\cal F}\left(\frac{dA_{\mu,v}}{dx}, \phi_{\mu,v}\right)+\epsilon.
\end{eqnarray}
Thus, the proof is completed for the case when (i) $v<1$ and (ii) $v=1$ and $\mu>0$.

Next, we treat the case when (iii) $v>1$ and $s> \Phi_{\mu,v}^{-1}\left(\frac{\Phi_{\mu,v}(\alpha_{\mu,v})}{\Phi(\alpha_{\mu,v})}\right)$.
Then we can  take as $a=b=\alpha_{\mu,v}$ and $a'=b'=\beta_{\mu,v,s}$ in ($\star$)
from Lemma \ref{sol1} and Lemma \ref{monotone}.
Then the following holds:
\begin{eqnarray*}
&&\hspace{-2em}\sqrt{\Phi(\alpha_{\mu,v})}\sqrt{\Phi_{\mu,v}(\alpha_{\mu,v})}
+\int_{\alpha_{\mu,v}}^{\beta_{\mu,v,s}}\sqrt{\phi(x)}\sqrt{\phi_{\mu,v}(x)}dx \nonumber\\
&&\hspace{-2em}+\sqrt{1-\Phi(\beta_{\mu,v,s})}\sqrt{\Phi_{\mu,v}(s)-\Phi_{\mu,v}(\beta_{\mu,v,s})}\nonumber\\
&\hspace{-2em}=&\hspace{-1.5em}{\cal F}\left(\frac{dA_{\mu,v}}{dx}, \phi_{\mu,v}\right).
\end{eqnarray*}
Thus, the proof is completed for the case when (iii) $v>1$ and $s> \Phi_{\mu,v}^{-1}\left(\frac{\Phi_{\mu,v}(\alpha_{\mu,v})}{\Phi(\alpha_{\mu,v})}\right)$.
\end{proof}

The following lemma is obvious by the definition of $A_{\mu,v}$.
\begin{lem}\Label{lem.direct2}
Suppose that  $\mu\in\R$ and $v>0$ satisfy
$v=1$ and $\mu\le0$, 
or $v>1$ and $s\le \Phi_{\mu,v}^{-1}\left(\frac{\Phi_{\mu,v}(\alpha_{\mu,v})}{\Phi(\alpha_{\mu,v})}\right)$.
Then,
the following equality holds
\begin{eqnarray}
{\cal F}\left(\frac{dA_{\mu,v}}{dx},\phi_{\mu,v}\right)
=\sqrt{{\Phi}_{\mu,v}(s)}.
\Label{lem.direct.ineq2}
\end{eqnarray}
\end{lem}

\subsection{Lemmas for Converse Part of Theorem \ref{Zform}}\Label{Zcon.app}

The following lemma is given as Lemma $15$ of \cite{KH13-2}.
\begin{lem}
\Label{naiseki2}
Let  $a=\{a_i\}_{i=0}^{I}$ and $b=\{b_i\}_{i=0}^{I}$ be probability distributions and satisfy $\frac{a_{i-1}}{b_{i-1}}>\frac{a_i}{b_i}$.
When $c=\{c_i\}_{i=0}^{I}$ is a probability distribution and satisfies 
\begin{eqnarray}
\sum_{i=0}^k a_k\le \sum_{i=0}^k c_k
\Label{gijimajo}
\end{eqnarray}
for any $k=0,1,...,I$,
the following holds:
\begin{eqnarray}
\sum_{i=0}^I \sqrt{a_i}\sqrt{b_i}\ge\sum_{i=0}^I \sqrt{c_i}\sqrt{b_i}.
\end{eqnarray}
Moreover, the equation holds for $c$ if and only if $c=a$.
\end{lem}

\begin{lem}\Label{Zcon}
Suppose that  $\mu\in\R$ and $v>0$ satisfy
(i) $v<1$,
(ii) $v=1$ and $\mu>0$, 
or (iii) $v>1$ and $s> \Phi_{\mu,v}^{-1}\left(\frac{\Phi_{\mu,v}(\alpha_{\mu,v})}{\Phi(\alpha_{\mu,v})}\right)$.
When real numbers $b\le b'$ satisfy the condition {\rm ($\star$)} in Lemma \ref{lem.direct},
the following inequality holds:
\begin{eqnarray}
&&\sup_{A\in{\cal A}_{s}}{\cal F}\left(\frac{dA}{dx}, \phi_{\mu,v}\right)\nonumber\\
&&\le
\sqrt{\Phi(b)} \sqrt{\Phi_{\mu,v}(b)}
+
\int_{b}^{b'}\sqrt{\phi(x)} \sqrt{\phi_{\mu,v}(x)}dx \nonumber\\
&&~~~+\sqrt{1-\Phi(b')} \sqrt{\Phi_{\mu,b}(s)-\Phi_{\mu,v}(b')}.\Label{lem.con.ineq}
\end{eqnarray}
\end{lem}

\begin{proof}
We set a sequence $\{x_i^I\}_{i=0}^{I}$ for $I\in\N$ as $x_i^I:=b+\frac{b'-b}{I}i$.
Then, we have the following  for an arbitrary $A$ in ${\cal A}_s$ defined in Definition \ref{rn}:
\begin{eqnarray}
&&\hspace{0em}{\cal F}\left(\frac{dA}{dx},\phi_{\mu,v}\right) \nonumber\\
&\hspace{0em}=&
\hspace{0em}\int_{-\infty}^{b}\sqrt{\frac{dA}{dx}(x)} \sqrt{\phi_{\mu,v}(x)}dx
+\int_{b'}^{s}\sqrt{\frac{dA}{dx}(x)} \sqrt{\phi_{\mu,v}(x)}dx\nonumber\\
&&\hspace{0em}+\sum_{i=1}^{I} \int_{x_{i-1}^I}^{x_{i}^I}\sqrt{\frac{dA}{dx}(x)} \sqrt{\phi_{\mu,v}(x)}dx\nonumber\\
&\hspace{0em}\le&
\hspace{0em}\sqrt{A(b)} \sqrt{{\Phi}_{\mu,v}(b)} + \sqrt{1-A(b')} \sqrt{{\Phi}_{\mu,v}(s)-{\Phi}_{\mu,v}(b')} \Label{ine1}\\
&&\hspace{0em}+\sum_{i=1}^{I} \sqrt{A(x_i^I)-A(x_{i-1}^I)} \sqrt{\Phi_{\mu,v}(x_i^I)-\Phi_{\mu,v}(x_{i-1}^I)} \nonumber\\
&\hspace{0em}\le&
\hspace{0em}\sqrt{{\Phi}(b)} \sqrt{{\Phi}_{\mu,v}(b)} +\sqrt{1-{\Phi}(b')} \sqrt{{\Phi}_{\mu,v}(s)-{\Phi}_{\mu,v}(b')} \Label{ine2}\\
&&\hspace{0em}+\sum_{i=1}^{I} \sqrt{\Phi(x_i^I)-\Phi(x_{i-1}^I)} \sqrt{\Phi_{\mu,v}(x_i^I)-\Phi_{\mu,v}(x_{i-1}^I)} \nonumber\\
&\hspace{0em}=&
\hspace{0em}\sqrt{{\Phi}(b)} \sqrt{{\Phi}_{\mu,v}(b)} +\sqrt{1-{\Phi}(b')} \sqrt{{\Phi}_{\mu,v}(s)-{\Phi}_{\mu,v}(b')} \nonumber\\
&&\hspace{0em}+\sum_{i=1}^{I} \sqrt{\frac{\Phi(x_i^I)-\Phi(x_{i-1}^I)}{x_{i}^I-x_{i-1}^I}
} \sqrt{
\frac{\Phi_{\mu,v}(x_i^I)-\Phi_{\mu,v}(x_{i-1}^I)}{x_{i}^I-x_{i-1}^I}}
(x_{i}^I-x_{i-1}^I) \Label{riemannsum}
\end{eqnarray}
where 
the inequality (\ref{ine1}) is obtained from the Schwartz inequality and
the inequality (\ref{ine2}) is obtained from Lemmas \ref{monotone} and  \ref{naiseki2}.
Here, the mean value theorem guarantees the existence of $\bar{x}_i^I\in[x_{i-1},x_i]$ and ${\tilde{x}_i^I}\in[x_{i-1},x_i]$ for $2\le i\le I$ which satisfy
\begin{eqnarray}
\sqrt{\frac{\Phi(x_i^I)-\Phi(x_{i-1}^I)}{x_{i}^I-x_{i-1}^I} }
&=&\sqrt{\phi(\bar{x}_i^I)}, \Label{meanvalue1}\\
\sqrt{\frac{\Phi_{\mu,v}(x_i^I)-\Phi_{\mu,v}(x_{i-1}^I)}{x_{i}^I-x_{i-1}^I} }
&=&\sqrt{\phi_{\mu,v}({\tilde{x}_i^I})}\nonumber\\
&=& \sqrt{\phi_{\mu,v}(\bar{x}_i^I) + (\phi_{\mu,v}({\tilde{x}_i^I}) - \phi_{\mu,v}(\bar{x}_i^I))}\nonumber\\ 
&\le& \sqrt{\phi_{\mu,v}(\bar{x}_i^I)} + \sqrt{\phi_{\mu,v}({\tilde{x}_i^I}) - \phi_{\mu,v}(\bar{x}_i^I)}.
\Label{meanvalue2}
\end{eqnarray}
Thus,
\begin{eqnarray}
&&\sum_{i=1}^{I}\sqrt{\frac{\Phi(x_i^I)-\Phi(x_{i-1}^I)}{x_{i}^I-x_{i-1}^I}
} \sqrt{
\frac{\Phi_{\mu,v}(x_i^I)-\Phi_{\mu,v}(x_{i-1}^I)}{x_{i}^I-x_{i-1}^I}}
(x_{i}^I-x_{i-1}^I)\nonumber\\ 
&\le&
\sum_{i=1}^{I} \sqrt{\phi(\bar{x}_i^I)}  \sqrt{\phi_{\mu,v}(\bar{x}_i^I)} (x_{i}^I-x_{i-1}^I)\nonumber\\
&& + \sum_{i=1}^{I} \sqrt{\phi(\bar{x}_i^I)}  \sqrt{\phi_{\mu,v}({\tilde{x}_i^I}) - \phi_{\mu,v}(\bar{x}_i^I)}  (x_{i}^I-x_{i-1}^I)\Label{meanvalue}\\
&\overset{I\to\infty}{\longrightarrow}&
\int_{b}^{b'}\sqrt{\phi(x)} \sqrt{\phi_{\mu,v}(x)}dx + 0. \Label{riemann}
\end{eqnarray}
where (\ref{riemann}) follows from the Riemann integrability of the continuous function $\sqrt{\phi}\sqrt{\phi_{\mu,v}}$.
Therefore, (\ref{lem.con.ineq}) is obtained from (\ref{riemannsum}) and (\ref{riemann}).
\end{proof}

\begin{lem}\Label{lem.converse2}
%
The following inequality holds:
\begin{eqnarray}
\sup_{A\in{\cal A}_{s}}{\cal F}\left(\frac{dA}{dx}, \phi_{\mu,v}\right)
 \le \sqrt{{\Phi}_{\mu,v}(s)}.
\Label{lem.converse.ineq2}
\end{eqnarray}
\end{lem}
\begin{proof}
For $A\in{\cal A}_s$ and $x\ge s$, $\frac{dA}{dx}(x)=0$ holds.
Thus,
\begin{eqnarray}
\sup_{A\in{\cal A}_{s}}{\cal F}\left(\frac{dA}{dx}, \phi_{\mu,v}\right)
&=& \int_{-\infty}^{s}\sqrt{\frac{dA}{dx}(x)}\sqrt{\phi_{\mu,v}(x)}dx\\
&\le& \sqrt{\int_{-\infty}^{s}\frac{dA}{dx}(x)dx}\sqrt{\int_{-\infty}^{s}\phi_{\mu,v}(x)dx}\\
&=& \sqrt{A(s)-A(-\infty)}\sqrt{{\Phi}_{\mu,v}(s)-{\Phi}_{\mu,v}(-\infty)}\\
&\le& \sqrt{{\Phi}_{\mu,v}(s)},
\end{eqnarray}
where we used the Schwartz inequality in the first inequality and $A(s)=1$ for $A\in {\cal A}_s$.
\end{proof}

\subsection{Proof of of Theorem \ref{Zform}}\Label{Zform.app}

Let $A_{\mu,v,s}$ be the function defined in Subsection \ref{Zdir.app}.
When $\mu\in\R$ and $v>0$ satisfy
$v<1$,
or $v=1$ and $\mu>0$, 
or $v>1$ and $s> \Phi_{\mu,v}^{-1}\left(\frac{\Phi_{\mu,v}(\alpha_{\mu,v})}{\Phi(\alpha_{\mu,v})}\right)$,
Lemmas \ref{lem.direct} and \ref{Zcon} derives 
\begin{eqnarray}
\sup_{A\in{\cal A}_{s}}{\cal F}\left(\frac{dA}{dx}, \phi_{\mu,v}\right)
={\cal F}\left(\frac{dA_{\mu,v,s}}{dx}, \phi_{\mu,v}\right).
\Label{Zequal}
\end{eqnarray}
Similarly, when  $\mu\in\R$ and $v>0$ satisfy
$v=1$ and $\mu\le0$, 
or $v>1$ and $s\le \Phi_{\mu,v}^{-1}\left(\frac{\Phi_{\mu,v}(\alpha_{\mu,v})}{\Phi(\alpha_{\mu,v})}\right)$, 
Lemmas \ref{lem.direct2} and \ref{lem.converse2} derives (\ref{Zequal}).
From the direct calculation of the right hand side of (\ref{Zequal}),
we obtain the concrete form of the generalized Rayleigh-normal distribution as in Theorem \ref{Zform}.
\endproof

\subsection{Proof of Proposition \ref{cum}}\Label{cum.app}
%
%
First, we show that $Z_{v,s}(\mu)$ is monotonically increasing.
We define a shift operator $S_{\mu}$ for a map $A:\R\to\R$ by $(S_{\mu}A)(x):=A(x-\mu)$.
Then we have ${\cal F}(S_{\mu}p,S_{\mu}q)={\cal F}(p,q)$.
Thus when we define the set of functions $A:\R\to [0,1]$ as 
\begin{eqnarray*}
{\cal A}_{s}(\mu):=\left\{A\Big|
\begin{array}{l}
{\it continuously~differentiable~monotone}\\
{increasing},
~A(s)=1,
~\Phi_{\mu,1}\le A\le1
\end{array}
\right\},
\end{eqnarray*}
we obtain the following form of the Rayleigh-normal distribution function
\begin{eqnarray*}
Z_{v,s}(\mu)
&:=&1-\sup_{A\in{\cal A}_{s}(0)}{\cal F}\left(\frac{dA}{dx}, \phi_{\mu,v}\right)^2\\
&=&1-\sup_{A\in{\cal A}_{s}(0)}{\cal F}\left(S_{-\mu}\frac{dA}{dx}, S_{-\mu}\phi_{\mu,v}\right)^2\\
&=&1-\sup_{A\in{\cal A}_{s}(0)}{\cal F}\left(\frac{d(S_{-\mu}A)}{dx}, \phi_{0,v}\right)^2\\
&=&1-\sup_{\tilde{A}\in{\cal A}_{s-\mu}(-\mu)}{\cal F}\left(\frac{d\tilde{A}}{dx}, \phi_{0,v}\right)^2.
\end{eqnarray*}
For $\mu<\tau$,  ${\cal A}_{s-\mu}(-\mu)\supset{\cal A}_{s-\tau}(-\tau)$ holds, and thus we obtain $Z_{v}(\mu)\le Z_{v}(\tau)$.

Next we show $\displaystyle\lim_{\mu\to\infty}Z_{v,s}(\mu)=1$.
Since the Rayleigh-normal distribution function $Z_{v}$ is a cumulative distribution function as was shown in \cite{KH13-2}, we have
$\displaystyle\lim_{\mu\to\infty}Z_{v,s}(\mu)\ge \displaystyle\lim_{\mu\to\infty}Z_{v}(\mu)=1$
from (\ref{limit}).

Next we show $\displaystyle\lim_{\mu\to-\infty}Z_{v}(\mu)=0$.
Since the generalized Rayleigh-normal distribution function is monotonically increasing,
it is enough to show that for an arbitrary $\epsilon$ there exists $\mu_{\epsilon}$ such that 
\begin{eqnarray}
\sup_{A\in{\cal A}_{s}}{\cal F}\left(\frac{dA}{dx}, \phi_{\mu_{\epsilon},v}\right)\ge1-\epsilon.
\Label{mu-limit}
\end{eqnarray}
Let $M_{\epsilon}>0$ be a real number such that $\Phi_{0,v}(M_{\epsilon})-\Phi_{0,v}(-M_{\epsilon})\ge1-\epsilon$.
Then, it is easily verified that we can take $\mu_{\epsilon}\ll0$ which satisfies $\mu_{\epsilon}+M_{\epsilon}<s$ and
$\Phi_{\mu_{\epsilon},v}(x)>\Phi(x)$ on $(\mu_{\epsilon}-M_{\epsilon},\mu_{\epsilon}+M_{\epsilon})$.
Then it implies that there exists a function $A_{\epsilon}\in{\cal A}_s$ such that $A_{\epsilon}=\Phi_{\mu_{\epsilon},v}$ on $(\mu_{\epsilon}-M_{\epsilon},\mu_{\epsilon}+M_{\epsilon})$.
Thus, we obtain (\ref{mu-limit}) as follows:
\begin{eqnarray*}
\sup_{A\in{\cal A}_{s}}{\cal F}\left(\frac{dA}{dx}, \phi_{\mu_{\epsilon},v}\right)
&\ge& {\cal F}\left(\frac{dA_{\epsilon}}{dx}, \phi_{\mu_{\epsilon},v}\right)\\
&\ge& \int_{\mu_{\epsilon}-M_{\epsilon}}^{\mu_{\epsilon}+M_{\epsilon}} \phi_{\mu_{\epsilon},v}dx\\
&=&\Phi_{0,v}(M_{\epsilon})-\Phi_{0,v}(-M_{\epsilon})\\
&\ge&1-\epsilon.
\end{eqnarray*}

Finally, we show that $Z_{v,s}(\mu)$ is continuous.
From Lemmas \ref{sol2}, \ref{sol0}, \ref{sol1} and the implicit function theorem, $\alpha_{\mu,v}$ and $\beta_{\mu,v,s}$ are differentiable, especially continuous, with respect to $\mu$.
Thus, we can verify that 
$Z_v(\mu)$ is continuous from Theorem \ref{Zform}.
\endproof

\subsection{Proof of Proposition \ref{lim1}}\Label{lim1.app}
From the definition of $I_{\mu,v}$,
\begin{eqnarray*}
0
\le\lim_{v\to0}I_{\mu,v}(\beta_{\mu,v,s})
\le\lim_{v\to0}I_{\mu,v}(\infty)
=0.
\end{eqnarray*}
Thus, to derive $\displaystyle\lim_{v\to0}Z_{v,s}(\mu)$, 
it is enough to evaluate $\displaystyle\lim_{v\to0}\Phi(\beta_{\mu,v,s})$, $\displaystyle\lim_{v\to0}\Phi_{\mu,v}(s)$ and $\displaystyle\lim_{v\to0}\Phi_{\mu,v}(\beta_{\mu,v,s})$ in (\ref{Z<1}).

First, we treat the case when $\mu>s$.
Since
\begin{eqnarray}
0
\le\lim_{v\to\infty}\Phi_{\mu,v}(\beta_{\mu,v,s})
\le\lim_{v\to\infty}\Phi_{\mu,v}(s)
=0
\end{eqnarray}
from (\ref{conbeta}) and the condition $\mu>s$,
we obatian
\begin{eqnarray*}
\lim_{v\to0}Z_{v,s}(\mu)
=1.
\end{eqnarray*}

Next, we treat the case when $\mu\le s$.
Then we obtain the following equations as shown below:
\begin{eqnarray}
\lim_{v\to0}\Phi(\beta_{\mu,v,s})&=&\Phi(\mu),\Label{limbeta1}\\
\lim_{v\to0}\Phi_{\mu,v}(\beta_{\mu,v,s})&=&0.\Label{limbeta2}
\end{eqnarray}
Since it holds that
\begin{eqnarray}
\lim_{v\to0}\Phi_{\mu,v}(s)
=
\left\{
\begin{array}{cll}
1&{\it if}&\mu< s\\
\frac{1}{2} & {\it if} & \mu= s,
\end{array}
\right.
\Label{slim}
\end{eqnarray}
we obtain the following equation from (\ref{limbeta1}), (\ref{limbeta2}) and (\ref{slim}):
\begin{eqnarray*}
\lim_{v\to0}Z_{v,s}(\mu)
=
\left\{
\begin{array}{cll}
\Phi\left(\mu\right) & {\it if} & \mu< s\\
\frac{1}{2}(1+\Phi\left(\mu\right)) & {\it if} & \mu= s.
\end{array}
\right.
\end{eqnarray*}

In the following,
we derive (\ref{limbeta1}) and (\ref{limbeta2}).

To show (\ref{limbeta1}),
it is enough to show that $\lim_{v\to0}\beta_{\mu,v,s}=\mu$. 
Since $\beta_{\mu,v,s}<\frac{\mu}{1-v}$ holds from Lemma \ref{sol2}, 
we obtain ${\rm limsup}_{v\to0}\beta_{\mu,v,s}\le\mu$.
Next, we show ${\rm liminf}_{v\to0}\beta_{\mu,v,s}\ge\mu$.
Note that $\beta_{\mu,v,s}$ is the unique zero point of 
\begin{eqnarray}
f_{\mu,v,s}(x)=(\Phi_{\mu,v}(s)-\Phi_{\mu,v}(x))-(1-\Phi(x))\frac{\phi_{\mu,v}(x)}{\phi(x)}
\end{eqnarray}
as was stated in Proof of Lemma \ref{sol2}.
To derive ${\rm liminf}_{v\to0}\beta_{\mu,v,s}\ge\mu$,
it is enough to show that an arbitrary $x\in\R$ less than $\mu$ is not the zero point of $f_{\mu,v,s}$ when $v$ is close to $0$.
From $\lim_{v\to0}\phi_{\mu,v}(x)=0$, $\lim_{v\to0}\Phi_{\mu,v}(x)=0$ and (\ref{slim}),
the inequality $\lim_{v\to0}f_{\mu,v,s}(x)=\lim_{v\to0}\Phi_{\mu,v}(s)\ge 1/2$ holds. 
Therefore, $x$ is not a zero point of $f_{\mu,v,s}$ when $v$ is close to $0$.
Thus, we obtain $\lim_{v\to0}\beta_{\mu,v,s}=\mu$.

Then, we show (\ref{limbeta2}).
In order to show it, it is enough to prove that $\lim_{v\to0}\frac{\beta_{\mu,v,s}-\mu}{\sqrt{v}}=-\infty$ by the definition of $\Phi_{\mu,v}$.
Since $\beta_{\mu,v,s}<\frac{\mu}{1-v}$ and $\lim_{v\to0}\frac{\mu}{1-v}=\mu$, $\beta_{\mu,v,s}$ is bounded above by some constant $\gamma$ as $\beta_{\mu,v,s}<\gamma$ when $v$ is close to $0$, and then, we have the following inequality:
\begin{eqnarray}
\frac{\phi(\beta_{\mu,v,s})}{\phi_{\mu,v}(\beta_{\mu,v,s})}
= \frac{1-\Phi(\beta_{\mu,v,s})}{\Phi_{\mu,v}(s)-\Phi_{\mu,v}(\beta_{\mu,v,s})} 
\ge 1-\Phi(\beta_{\mu,v,s})
= \Phi(-\beta_{\mu,v,s})
\ge \Phi(-\gamma).
\end{eqnarray}
Thus, the following holds:
\begin{eqnarray}
&&\hspace{-1em}2\log \Phi(-\gamma)\nonumber\\
&\hspace{-1.5em}\le&\hspace{-1em}2\mathrm{log}\frac{\phi(\beta_{\mu,v,s})}{\phi_{\mu,v}(\beta_{\mu,v,s})}\nonumber\\
&\hspace{-1.5em}=&\hspace{-1em}(1-v)
\left(
\frac{\beta_{\mu,v,s} - \frac{\mu}{1-v}}{\sqrt{v}}
\right)^2
+\mathrm{log}v
-\frac{\mu^2}{1-v}.
\Label{inequality12}
\end{eqnarray}
Since $-\infty < 2\log \Phi(-\gamma)$ and $\lim_{v\to0}\log v =-\infty$, 
we have 
\begin{eqnarray}
\lim_{v\to0}\left(\frac{\beta_{\mu,v,s} - \frac{\mu}{1-v}}{\sqrt{v}}\right)^2
=\infty.
\end{eqnarray}
Since Lemma \ref{sol2} guarantees that
$\beta_{\mu,v,s}<\frac{\mu}{1-v}$, 
we obtain
\begin{eqnarray}
\lim_{v\to0}\frac{\beta_{\mu,v,s}-\mu}{\sqrt{v}}
=\lim_{v\to0}\frac{\beta_{\mu,v,s} - \frac{\mu}{1-v}}{\sqrt{v}}
=-\infty.
\end{eqnarray}
\endproof

\subsection{Proof of Proposition \ref{lim2}}\Label{lim2.app}
From (\ref{Z>1}) of Theorem \ref{Zform},
the generalized Rayleigh-normal distribution function $Z_{v,\sqrt{v}s}(\sqrt{v}\mu)$ has two different forms
depending on the sign of $s-\sqrt{v}^{-1}\Phi_{\sqrt{v}\mu,v}^{-1}(\frac{\Phi_{\sqrt{v}\mu,v}(\alpha_{\sqrt{v}\mu,v})}{\Phi(\alpha_{\sqrt{v}\mu,v})})$.
To analyze the sign in the limit $v\to\infty$, 
we first see the behaviour of $\alpha_{\sqrt{v}\mu,v}$.
When $1<v$, the equation with respect to $x$
\begin{eqnarray}\Label{threshold-1}
\frac{1-\Phi\left(x\right)}{1-\Phi_{\mu,v}(x)}
=\frac{\phi(x)}{\phi_{\mu,v}(x)}
\end{eqnarray}
has the unique solution $\beta_{\mu,v}$  and 
the following equation holds from (19) of \cite{KH13-2}:
\begin{eqnarray}
\alpha_{\sqrt{v}\mu,v}
=\sqrt{v}(\mu-\beta_{\mu,1/v}).
\end{eqnarray}
Then, from (31) of \cite{KH13-2}, we obtain
\begin{eqnarray}
\lim_{v\to\infty}\alpha_{\sqrt{v}\mu,v}
&=&\lim_{v\to\infty}\sqrt{v}(\mu-\beta_{\mu,1/v})
=\lim_{v\to0}\frac{\mu-\beta_{\mu,v}}{\sqrt{v}}
=\infty,
\Label{al1}
\end{eqnarray}
Similarly, from (26) of \cite{KH13-2}, we obtain
\begin{eqnarray}
\lim_{v\to\infty}\frac{\alpha_{\sqrt{v}\mu,v}}{\sqrt{v}}
&=&\lim_{v\to\infty}\mu-\beta_{\mu,1/v}
=\lim_{v\to0}\mu-\beta_{\mu,v}
=0.
\Label{al2}
\end{eqnarray}
Then we have 
\begin{eqnarray}
\lim_{v\to\infty}\sqrt{v}^{-1}\Phi_{\sqrt{v}\mu,v}^{-1}\left(\frac{\Phi_{\sqrt{v}\mu,v}(\alpha_{\sqrt{v}\mu,v})}{\Phi(\alpha_{\sqrt{v}\mu,v})}\right)=0,
\end{eqnarray}
and thus,
the form of the generalized Rayleigh-normal distribution function $Z_{v,\sqrt{v}s}(\sqrt{v}\mu)$ is determined according to the sign of $s$ when $v\to\infty$.
When $s\le0$,
\begin{eqnarray}
\lim_{v\to\infty}Z_{v,\sqrt{v}s}(\sqrt{v}\mu)
&=&\displaystyle\lim_{v\to\infty} 1-\Phi_{\sqrt{v}\mu,v}(\sqrt{v}s) \nonumber\\
&=&\Phi(\mu-s). 
\Label{vtoinfty1}
\end{eqnarray}
Next we treat the case when $s>0$.
From the inequality $\alpha_{\sqrt{v}\mu,v}\le\beta_{\sqrt{v}\mu,v,\sqrt{v}s}$ of Lemma \ref{sol1} and (\ref{al1}),
\begin{eqnarray*}
\lim_{v\to\infty}\Phi(\alpha_{\sqrt{v}\mu,v})
=\lim_{v\to\infty}\Phi(\beta_{\sqrt{v}\mu,v,\sqrt{v}s})
=1.
\end{eqnarray*}
From  (\ref{al2}),
\begin{eqnarray*}
\lim_{v\to\infty} \Phi_{\sqrt{v}\mu,v}(\alpha_{\sqrt{v}\mu,v})
=\Phi(-\mu).
\end{eqnarray*}
From the definition of $I_{\mu,v}$,
\begin{eqnarray*}
\lim_{v\to\infty}I_{\sqrt{v}\mu,v}(\beta_{\sqrt{v}\mu,v,\sqrt{v}s}) 
=\lim_{v\to\infty} I_{\sqrt{v}\mu,v}(\alpha_{\sqrt{v}\mu,v})
=0.
\end{eqnarray*}
Thus, when $s>0$,
\begin{eqnarray}
\lim_{v\to\infty}Z_{v,\sqrt{v}s}(\sqrt{v}\mu)
=\Phi(\mu).
\Label{vtoinfty2}
\end{eqnarray}
From (\ref{vtoinfty1}) and (\ref{vtoinfty2}), the proof is completed.
\endproof

\subsection{Proof of Lemma \ref{opt trans}}\Label{opt trans.app}

We set as  
\begin{eqnarray*}
&&\tilde{F}^{\cal M}(P\to Q|M)\\
&:=&\sup_{P'' \in \mathcal{P}(\calY)}
\left\{F(P'', Q)\Bigg|
\begin{array}{l}
\exists P'\in\mathcal{P}(\N_M), \\
P\prec P'\prec P''
\end{array}
\right\}
\end{eqnarray*}
where $\N_M:=\{1,...,M\}$.
Then, it satisfies 
\begin{eqnarray}
F^{\cal M}(P\to Q|N)=\tilde{F}^{\cal M}(P\to Q|2^N).
\end{eqnarray}
Thus, to prove Lemma \ref{opt trans}, it is enough to show the equality
\begin{eqnarray}
\tilde{F}^{\cal M}(P\to Q|M)=F^{\cal M}({\cal C}_{M}(P) \to Q)
\Label{op}
\end{eqnarray}
for an arbitrary $M\in\N$.

Because of Lemma \ref{opt con},
$P\prec {\cal C}_{M}(P)$ and ${\cal C}_{M}(P)\in\calP(\N_M)$ hold.
Thus, from the definition of $\tilde{F}^{\cal M}(P\to Q|M)$, 
we have 
\begin{eqnarray}
\tilde{F}^{\cal M}(P\to Q|M)
\ge F^{\cal M}({\cal C}_{M}(P) \to Q).
\Label{op2}
\end{eqnarray}
Then, we show 
\begin{eqnarray}
\tilde{F}^{\cal M}(P\to Q|M)
\le F^{\cal M}({\cal C}_{M}(P) \to Q).
\Label{op3}
\end{eqnarray}
To prove (\ref{op3}), 
it is enough to prove that ${\cal C}_M(P)\prec P'$ for an arbitrary $P'\in \calP(\N_M)$ such that $P\prec P'$.
Without loss of generality, 
we assume that $P' = P'^{\downarrow}$.
Here, we use the inductive method.
When $M=1$, then (\ref{op3})  holds for any probability distribution $P$. 
Let us assume that (\ref{op3}) holds for any $P$ when $M=k-1$.
In the following, we show that (\ref{op}) holds for any $P$ when $M=k$.
When $J_{P, k}=1$, 
${\cal C}_k(P)$ equals the uniform distribution $U_k$ on $\N_k$ and satisfies ${\cal C}_k(P)=U_k\prec P'$.

Let $J_{P, k}\ge2$ in the following.
There exists $Q'$ which satisfies 
\begin{eqnarray}
P\prec Q'\prec P'~{\rm and} ~Q'(1)=P^{\downarrow}(1)
\Label{Q'}
\end{eqnarray}
as shown below.
Then,
$P^{\downarrow}|_{\{2, ..., M\}}\prec Q'|_{\{2, ..., M\}}$ holds since $P\prec Q'$.
By the assumption of the inductive method, 
$\frac{1}{C}{\cal C}_k(P)|_{\{2, ..., M\}}\prec \frac{1}{C'}Q'_{\{2, ..., M\}}$ where $C=\sum_{i=2}^{M}{\cal C}_k(P)(i)$ and $C'=\sum_{i=1}^M Q'(i)$ are normalizing constants.
Thus, it follows that ${\cal C}_k(P)\prec Q' \prec P'$.

All we have to do is to show the existence of $Q'$ which satisfies (\ref{Q'}).
When $P'(1)=P^{\downarrow}(1)$, 
we can take as $Q'=P'$.
When $P'(1)>P^{\downarrow}(1)$, 
let $l_0:=\max\{ l \in\{1,...,M\}|P'(1)=P'(l)\}$ and $\omega:=\sum_{l=1}^{l_0}(P'(l)-P^{\downarrow}(1))$.
Moreover, we define the set $K$ by $\{l\in\{1,...,M\}|P'(l)<P^{\downarrow}(l)\}=\{l_1, ..., l_m\}$ where $l_i\le l_{i+1}$ and determine $r_0\in K$ by the condition 
\begin{eqnarray}
\sum_{i=1}^{r_0-1}(P^{\downarrow}(l_i)-P'(l_i))<\omega\le \sum_{i=1}^{r_0}(P^{\downarrow}(l_i)-P'(l_i)).
\Label{r0}
\end{eqnarray}
By using those notations, 
we set a probability distribution $Q'$ by 
\begin{eqnarray*}
&Q'(l)=&\hspace{-0.5em}\left\{
\begin{array}{l}
P^{\downarrow}(1) \hspace{1.1em} \textit{if}~1\le l \le l_0 \\
P^{\downarrow}(l)  \hspace{1.2em} \textit{if}~l=l_1, ..., l_{r_0-1} \\
P'(l_{r_0})+\omega-\displaystyle\sum_{i=1}^{r_0-1}(P^{\downarrow}(l_i)-P'(l_i)) ~~ \textit{if}~l=l_{r_0} \\
P'(l) \hspace{1.2em} {otherwise}. 
\end{array}
\right.
\end{eqnarray*}
Then, $Q'(1)=P^{\downarrow}(1)$ by the definition and
we can verify $Q'=Q'^{\downarrow}$.

We show $P\prec Q'$.
For $1\le l\le l_0$, 
we have $Q'(l) = P^{\downarrow}(1)\ge P^{\downarrow}(l)$.
For $l_0< l\le l_{r_0}-1$ and $l\in K$, 
we have $Q'(l) = P^{\downarrow}(l)$.
For $l_0< l\le l_{r_0}-1$ and $l\notin K$, 
we have $Q'(l) = P'(l) \ge P^{\downarrow}(l)$
by the definition of $K$.
For $l=l_{r_0}$, 
we have $Q'(l) \le P'(l_{r_0})+\omega-\displaystyle\sum_{i=1}^{r_0-1}(P^{\downarrow}(l_i)-P'(l_i)) \le P^{\downarrow}(l_{r_0})$
by (\ref{r0}). 
Thus, when $1\le k\le l_{r_0}$,
we obtain
\begin{eqnarray*}
\sum_{l=1}^k Q'(l) 
\ge \sum_{l=1}^k P^{\downarrow}(l).
\end{eqnarray*}
Moreover, when $l_{r_0}\le k$,
we obtain
\begin{eqnarray*}
\sum_{l=1}^k Q'(l) 
= \sum_{l=1}^{k} P'(l) 
\ge \sum_{l=1}^k P^{\downarrow}(l),
\end{eqnarray*}
where we used $\sum_{l=1}^{l_{r_0}} Q'(l) = \sum_{l=1}^{l_{r_0}} P'(l)$, $Q'(l)= P'(l)$ for $l>l_{r_0}$ and $P\prec P'$.
From the above discussion,
we obtain $P\prec Q'$.

Next, we show $Q'\prec P'$.
When $1\le l\le l_{0}$,
$Q'(l)=P^{\downarrow}(1)<P'(1)=P'(l)$ holds.
Thus,
when $1\le k\le l_{0}$,
we obtain
\begin{eqnarray}
\sum_{l=1}^k Q'(l) 
\le \sum_{l=1}^k P'(l).
\Label{majo1}
\end{eqnarray}
When $l_{0}< k\le l_{r_0}-1$,
\begin{eqnarray}
\sum_{l=1}^k P'(l) - \sum_{l=1}^k Q'(l) 
&=& \sum_{l=1}^{l_0} (P'(l) - P^{\downarrow}(1)) - \sum_{i=1}^{r-1}(P^{\downarrow}(l_i)-P'(l_i))\\
&=&\omega - \sum_{i=1}^{r-1}(P^{\downarrow}(l_i)-P'(l_i)),
\Label{diff1}
\end{eqnarray}
where $r$ is defined by $l_{r-1}\le k \le l_{r}-1$.
Since $l_{0}< k\le l_{r_0}-1$,
$r\le r_0$ holds.
Thus,
the right hand side of (\ref{diff1}) is non-negative by (\ref{r0})
and $(\ref{majo1})$ holds for $l_{0}< k\le l_{r_0}-1$.
Moreover, 
$(\ref{majo1})$ holds for $l_{r_0}\le k$
since $\sum_{l=1}^k Q'(l) 
= \sum_{l=1}^{k} P'(l)$.
From the above discussion,
we obtain $Q'\prec P'$.
\endproof

\subsection{Proof of Proposition \ref{opt trans cor}}\Label{opt trans cor.app}

Let $m\ge n$.
Then, the size of storage is greater than or equal to the size of support of the source distribution $U_N^n$, and thus the performances of deterministic (or majorization) conversions via storage and that without storage coincide with each other.
Thus, we have 
\begin{eqnarray}
L^{\cal D}_n(U_N, Q|\nu, m\log N)
&=& L^{\cal D}_{n}(U_N, Q|\nu),\\
L^{\cal M}_n(U_N, Q|\nu, m\log N)
&=& L^{\cal M}_{n}(U_N, Q|\nu). 
\end{eqnarray}

Next, let $m\le n$.
Then, $U_N^m$ on the storage with size $N^m$ can be converted from $U_N^n$ by deterministic and majorization conversion.
Thus, we have 
\begin{eqnarray}
L^{\cal D}_n(U_N, Q|\nu, m\log N)
&\ge& L^{\cal D}_{m}(U_N, Q|\nu),\\
L^{\cal M}_n(U_N, Q|\nu, m\log N)
&\ge& L^{\cal M}_{m}(U_N, Q|\nu). 
\end{eqnarray}
Moreover, since any probability distribution on a set with size $N^m$ can be converted from $U_N^n$ by majorization conversion.
Therefore we have 
\begin{eqnarray}
L^{\cal M}_n(U_N, Q|\nu, m\log N)
&\le& L^{\cal M}_{m}(U_N, Q|\nu). 
\end{eqnarray}
\endproof

\subsection{Proof of Proposition \ref{opt trans cor2}}\Label{opt trans cor2.app}

When 
$m \ge L^i_n(P, U_N|\nu) $,
the equation
\begin{eqnarray}
L^i_n(P, U_N|\nu, m\log N)
=L^i_n(P, U_N|\nu)
\end{eqnarray}
 holds by the definition.

Let $m \le L^{\cal M}_n(P, U_N|\nu) $.
Then, by the definition,
\begin{eqnarray}
m 
\le  L^{\cal D}_n(P, U_N|\nu, m\log N). 
\end{eqnarray}
Moreover,
since any probability distribution on a set with size $N^m$ can be converted from $U_N^n$ by majorization conversion,
we obtain
\begin{eqnarray}
 L^{\cal D}_n(P, U_N|\nu, m\log N) 
&\le& L^{\cal M}_n(P, U_N|\nu, m\log N) \nonumber\\ 
&\le& L^{\cal M}_n(U_N^m, U_N|\nu) \nonumber\\
&\le& m - 2\log_N \nu.
\end{eqnarray}
where the first inequality follows from (\ref{fidelity-ineq2}).
\endproof

\subsection{Proof of Lemma \ref{simulate}}\Label{simulate.app}
Since $\{(S'_n,T'_n)\}$ is simulated by $\{(S_n,T_n)\}$,
there exists a sequence of $0<a_n\le1$ such that $S'_n=S_{a_nn}$ and $T'_n=T_{a_nn}$.
From the $\nu$-achievability of $\{(S_n,T_n)\}$,
we have the following inequality:
\begin{eqnarray*}
&&\liminf_{m\to\infty} F^i(P^m\to Q^{T'_{m}}|S'_{m})\nonumber\\
&=&\liminf_{m\to\infty} F^i(P^m\to Q^{T_{a_mm}}|S_{a_mm})\nonumber\\
&=&\liminf_{m\to\infty} F^i(P^{\frac{n_m}{a_m}}\to Q^{T_{n_m}}|S_{n_m})\\
&\ge&\liminf_{m\to\infty} F^i(P^{n_m}\to Q^{T_{n_m}}|S_{n_m})\\
&\ge&\liminf_{n\to\infty} F^i(P^{n}\to Q^{T_{n}}|S_{n})\\
&\ge&\nu,
\end{eqnarray*}
where $n_m:=a_mm$ and $i={\cal D}$ or ${\cal M}$.
\endproof

\subsection{Proof of Theorem \ref{region1}}\Label{region1.app}

We prepare the following lemma.
\begin{lem}\Label{transition'}
Let $\{M_n\}_{n\in\N}$ be a sequence of natural numbers.
Let $\{P_n\}_{n\in\N}$ and $\{P'_n\}_{n\in\N}$ be sequences of probability distributions on $\{\calX_n\}_{n\in\N}$ and  $\{\calX'_n\}_{n\in\N}$, respectively.
Suppose that 
there exists a sequence of deterministic conversion $W_n:\calX_n\to\calX'_n$ such that 
\begin{eqnarray}
\liminf_{n\to \infty} F(W_n(P_n), P'_n) 
= 1.
\Label{trans1'}
\end{eqnarray}
Then, 
for an arbitrary sequence $\{Q_n\}_{n\in\N}$ of probability distributions on $\{\calY_n\}_{n\in\N}$
and  arbitrary deterministic conversions  $W'_n:\calX'_n\to\calY_n$,
 the following holds:
\begin{eqnarray}
\liminf_{n\to \infty} F(W'_n\circ W_n(P_n), Q_n)
\ge\liminf_{n\to \infty} F(W'_n(P'_n), Q_n).
\Label{trans2'}
\end{eqnarray}
\end{lem}
\begin{proof}
Using the Hellinger distance $d_H(\cdot, \cdot)=\sqrt{1-F(\cdot, \cdot)}$,
we have the following inequalities:
\begin{eqnarray}
&&\limsup_{n\to \infty} d_H(W'_n\circ W_n(P_n), Q_n)\nonumber\\
&\le& \limsup_{n\to \infty} d_H(W'_n\circ W_n(P_n), W'_n(P'_n)) + \limsup_{n\to \infty} d_H(W'_n(P'_n), Q_n)
\Label{triangle}\\
&\le& \limsup_{n\to \infty} d_H(W_n(P_n), P'_n) + \limsup_{n\to \infty} d_H(W'_n(P'_n), Q_n)
\Label{monotone}\\
&=&\limsup_{n\to \infty} d_H(W'_n(P'_n), Q_n),
\Label{lasteq}
\end{eqnarray}
where 
(\ref{triangle}) and (\ref{monotone}) follow from the triangle inequality and the monotonicity of the Hellinger distance, respectively,
and (\ref{lasteq}) follows from (\ref{trans1'})
From the definition of the Hellinger distance,
we obtain (\ref{trans2'}).
\end{proof}

From Lemma \ref{transition'},
we have the following lemma.
\begin{lem}\Label{transition}
Let $\{M_n\}_{n\in\N}$ be a sequence of natural numbers
and $\{P'_n\}_{n\in\N}$ be a sequence of probability distributions on $\{0,1\}^{{M_n}}$.
Suppose that a sequence $\{P_n\}_{n\in\N}$ of probability distributions on $\{\calX_n\}_{n\in\N}$ satisfies
\begin{eqnarray}
\liminf_{n\to \infty} F^{\cal D}(P_n\to P'_n) 
= 1.
\Label{trans1}
\end{eqnarray}
Then, the following holds for an arbitrary sequence $\{Q_n\}_{n\in\N}$ of probability distributions on $\{\calY_n\}_{n\in\N}$:
\begin{eqnarray}
\liminf_{n\to \infty} F^{\cal D}(P_n\to Q_n| M_n)
\ge \liminf_{n\to \infty} F^{\cal D}(P'_n\to Q_n).
\Label{trans2}
\end{eqnarray}
\end{lem}

First, we prove the direct part of Theorem \ref{region1}.
Let $s_1\ge H(P)$.
From the results about the asymptotic maximal fidelity in \cite{KH13-2}, 
when $\epsilon$ is in $(0,1/2)$,
\begin{eqnarray}
\liminf_{n\to \infty} F^{\cal D}(P^n\to U_2^{H(P)n-n^{1/2+\epsilon/2}}) 
= 1.
\end{eqnarray}
Thus, using Lemma \ref{transition},
\begin{eqnarray*}
&&\lim_{n\to \infty} F^{\cal D}(P^n\to Q^{\frac{H(P)}{H(Q)}n-n^{1/2+\epsilon}}|{s_1n})
\nonumber\\
&\ge&\lim_{n\to \infty} F^{\cal D}(U_2^{H(P)n-n^{1/2+\epsilon/2}}\to Q^{\frac{H(P)}{H(Q)}n-n^{1/2+\epsilon}})
=1
\end{eqnarray*}
holds.
Thus, a first-order achievable rate $t_1$ satisfies $t_1\ge \frac{H(P)}{H(Q)}$.
Next, let $s_1< H(P)$.
From the results about the asymptotic maximal fidelity in \cite{KH13-2}, 
\begin{eqnarray}
\liminf_{n\to \infty} F^{\cal D}(P^n\to U_2^{s_1n}) 
= 1.
\end{eqnarray}
Thus, using Lemma \ref{transition},
\begin{eqnarray*}
&&\lim_{n\to\infty} F^{\cal D}(P^n\to Q^{\frac{s_1}{H(Q)}n-n^{1/2+\epsilon}}|{s_1n})\nonumber\\
&\ge& \lim_{n\to\infty} F^{\cal D}(U_2^{s_1n}\to Q^{\frac{s_1}{H(Q)}n-n^{1/2+\epsilon}})=1
\end{eqnarray*}
holds. 
Thus, a first-order achievable rate $t_1$ satisfies $t_1\ge \frac{s_1}{H(Q)}$.

Then, we prove the converse part.
Let $s_1\ge H(P)$.
From the results about the asymptotic maximal fidelity in \cite{KH13-2}, 
when $\epsilon$ is in $(0,1/2)$,
\begin{eqnarray*}
&&\lim_{n\to \infty} F^{\cal M}(P^n\to Q^{\frac{H(P)}{H(Q)}n+n^{1/2+\epsilon}}|{s_1n})
\nonumber\\
&\le&\lim_{n\to \infty} F^{\cal M}(P^n\to Q^{\frac{H(P)}{H(Q)}n+n^{1/2+\epsilon}})
=0
\end{eqnarray*}
holds.
Thus, a first-order achievable rate $t_1$ satisfies $t_1\le \frac{H(P)}{H(Q)}$.
Next, let $s_1< H(P)$.
Then, 
\begin{eqnarray*}
&&\lim_{n\to\infty} F^{\cal M}(P^n\to Q^{\frac{s_1}{H(Q)}n+n^{1/2+\epsilon}}|{s_1n})\nonumber\\
&\le& \lim_{n\to\infty} F^{\cal M}(U_2^{s_1n}\to Q^{\frac{s_1}{H(Q)}n+n^{1/2+\epsilon}})=0
\end{eqnarray*}
holds,
where we used the fact that an arbitrary distribution on the storage $\{0,1\}^{s_1n}$ can be converted from $U_2^{s_1n}$. 
Thus, a first-order achievable rate $t_1$ satisfies $t_1\le \frac{s_1}{H(Q)}$.
\endproof

\subsection{Proof of Lemma \ref{simulate-2nd}}\Label{simulate-2nd.app}

We set as $S_n=s_1n +s_2\sqrt{n}$, $T_n=t_1n +t_2\sqrt{n}$, $S'_n=s_1n +s'_2\sqrt{n}$, $T'_n=t_1n +t'_2\sqrt{n}$.

First, we show the ``only if" part.
Since $(s_2,t_2)$ simulates $(s'_2,t'_2)$,
there exists $0 < a_n \le 1$ such that
\begin{eqnarray}
&&(S'_n, T'_n) = (S_{a_nn}, T_{a_nn})\\
&\Leftrightarrow &
s_1n +s'_2\sqrt{n} = s_1(a_nn) +s_2\sqrt{a_nn},\\
&&t_1n +t'_2\sqrt{n} = t_1(a_nn) +t_2\sqrt{a_nn}\\
&\Leftrightarrow &
s_1(1-a_n) \sqrt{n}= s_2\sqrt{a_n}-s_2', \Label{eq-s}\\
&&t_1(1-a_n)\sqrt{n} = t_2\sqrt{a_n}-t_2'. \Label{eq-t}
\end{eqnarray}
Then we obtain $\lim _{n\to\infty}a_n=1$ by taking the limit $n\to\infty$ since the right hand sides of (\ref{eq-s}) and (\ref{eq-t}) are finite.
In addition, we also obtain $s_2\ge s'_2$ since $\lim_{n\to\infty}a_n=1$ and the left-hand side of (\ref{eq-s}) is non-negative because of $a_n\le1$.
Since (\ref{eq-s}) is equivalent with 
\begin{eqnarray}
t_1(1-a_n)\sqrt{n}= \frac{t_1}{s_1}(s_2\sqrt{a_n}-s_2'), \Label{eq-s2}
\end{eqnarray}
we obtain the equation (\ref{eq-st}) by (\ref{eq-t}) and (\ref{eq-s2}).

Next, we show the ``if" part. 
We can give the concrete value of $a_n$ from the quadratic equation with respect to $\sqrt{a_n}$:
\begin{eqnarray}
t_1(1-a_n)\sqrt{n}= t_2\sqrt{a_n}-t_2'. \Label{eq-t2}
\end{eqnarray}
From the assumption (\ref{eq-st}), the same $\sqrt{a_n}$ satisfies 
\begin{eqnarray}
s_1(1-a_n)\sqrt{n}= s_2\sqrt{a_n}-s_2'. 
\Label{eq-s3}
\end{eqnarray}
Thus, we obtain (\ref{eq-s}) and (\ref{eq-t}) and  the proof is completed.
\endproof

\subsection{Proof of Theorem \ref{region2-non-admissible}}\Label{region2-non-admissible.app}

To prove Theorem \ref{region2-non-admissible},
we prepare the following lemma which was given in the subsection $4.2$ of \cite{KH13-2}.
\begin{lem}\Label{lem.uni}
When $P$ and  $Q$ are non-uniform distributions,
the following equations hold for $i={\cal D}$ and ${\cal M}$:
\begin{eqnarray*}
\lim_{n\to\infty}F^{i}(U_2^{n} \to Q^{\frac{1}{H(Q)}n+t_2 \sqrt{n}})
&=&\sqrt{\Phi\left(-\sqrt{\frac{H(Q)^3}{V(Q)}}t_2\right)},\\
\lim_{n\to\infty}F^{i}(P^{n} \to U_2^{H(P)n+t_2 \sqrt{n}})
&=&\sqrt{\Phi\left(-\frac{t_2}{\sqrt{V(P)}}\right)}.
\end{eqnarray*}
\end{lem}

The function $F_{P,Q,s_1,s_2}$ in (\ref{non-ex-fid}) is obviously continuous and strictly monotonically decreasing on $F_{P,Q,s_1,s_2}^{-1}((0,1))$.
In the following,
we show that (\ref{feq}) holds.

We first prove the direct part.
Since $s_1<H(P)$, 
\begin{eqnarray}
\liminf_{n\to \infty} F^{\cal D}(P^n\to U_2^{s_1n + s_2\sqrt{n}}) 
= 1
\end{eqnarray}
from the results about the asymptotic maximal fidelity in \cite{KH13-2}.
Thus, using Lemma \ref{transition},
\begin{eqnarray}
F_{P,Q,s_1,s_2}^{\cal D}(t_2)
&\ge& \lim_{n\to\infty}F^{\cal D}(U_2^{s_1 n + s_2\sqrt{n}} \to Q^{\frac{s_1}{H(Q)}n+t_2 \sqrt{n}})\nonumber\\
&=& F_{P,Q,s_1,s_2}(t_2),\Label{13-1}
\end{eqnarray}
where the equality follows from Lemma \ref{lem.uni}.

Next, we  prove the converse part.
Since an arbitrary probability distribution on $\{0,1\}^{{s_1 n}}$ can be converted from the uniform distribution with size of ${s_1n}$ bits by majorization conversion.
Thus, we have
\begin{eqnarray}
F_{P,Q,s_1,s_2}^{\cal M}(t_2)
&\le& \lim_{n\to\infty}F^{\cal M}(U_2^{s_1 n + s_2\sqrt{n}} \to Q^{\frac{s_1}{H(Q)}n+t_2 \sqrt{n}})\nonumber\\
&=& F_{P,Q,s_1,s_2}(t_2),\Label{13-2}
\end{eqnarray}
where the equality follows from Lemma \ref{lem.uni}.
From (\ref{fidelity-ineq}), (\ref{13-1}) and (\ref{13-2}),
we obtain (\ref{feq}).
\endproof

\subsection{Proof of Direct Part of Theorem \ref{main} (Proof of Lemma \ref{separation})}\Label{main.app}

We first give a sketch of a proof of Lemma \ref{separation} in the following.
Then, we give a detailed proof of Lemma \ref{separation}.



\vspace{0.5em}
\noindent[{\it Sketch of Proof of Lemma \ref{separation}}]~
We first show (\ref{sepopt'}) of Lemma \ref{separation}.
We will construct probability distributions $P'_{n}$
such that 
\begin{eqnarray}
\liminf_{n\to\infty} F(P'_{n}, U_2^{ H(P)n+s_2\sqrt{n}\downarrow})
&\ge&
 \limsup_{n\to\infty}F^{\calD}(P^{n}\to U_2^{H(P)n + s_2\sqrt{n}}) - \epsilon.
\Label{sep1}
\end{eqnarray}
Then,
we will show that there exist maps $f_n$ such that 
\begin{eqnarray}
\liminf_{n\to\infty} F(W_{f_n}(P^{n\downarrow}), P'_{n})
&=& 1.
\Label{sep2}
\end{eqnarray}
Then, (\ref{sepopt'}) of Lemma \ref{separation} is derived from  (\ref{sep1}) and (\ref{sep2}).

Next, we show (\ref{sepopt}) of Lemma \ref{separation}.
We will show the existence of probability distributions $Q'_{n}$
such that 
\begin{eqnarray}
\liminf_{n\to\infty} 
F(Q'_{n}, Q^{ \frac{H(P)}{H(Q)}n+s_2\sqrt{n}\downarrow})
&\ge&\sqrt{1-Z_{C_{P,Q},\frac{s_2}{\sqrt{V(P)}}}(t_2D_{P,Q})} - \frac{\epsilon}{2}.
\Label{sep3}
\end{eqnarray}
Then,
we will show that there exist maps $f'_{n}$ such that 
\begin{eqnarray}
 \liminf_{n\to\infty} 
F(W_{f'_{n}}(P'_n), Q'_{n})
&\ge&  1  - \epsilon',
\Label{sep4}
\end{eqnarray}
where 
\begin{eqnarray}
\epsilon'
:=\left( \sqrt{1- \left( \sqrt{1-Z_{C_{P,Q},\frac{s_2}{\sqrt{V(P)}}}(t_2D_{P,Q})}-\epsilon\right)} 
- \sqrt{1 - \left(\sqrt{1-Z_{C_{P,Q},\frac{s_2}{\sqrt{V(P)}}}(t_2D_{P,Q})} - \frac{\epsilon}{2} \right)}~ \right)^2.
\end{eqnarray}
From (\ref{sep3}) and (\ref{sep4}), we have the following inequality with respect to the Hellinger distance $d_H(\cdot, \cdot) = \sqrt{1-F(\cdot,\cdot)}$:
\begin{eqnarray}
&&\limsup_{n\to\infty} 
d_H(W_{f'_{n}}(P'_n), Q^{ \frac{H(P)}{H(Q)}n+s_2\sqrt{n}\downarrow})\\
&\le&
\limsup_{n\to\infty} d_H(W_{f'_{n}}(P'_n), Q'_{n}) 
+ \limsup_{n\to\infty} d_H(Q'_{n}, Q^{ \frac{H(P)}{H(Q)}n+s_2\sqrt{n}\downarrow}) \nonumber\\
&\le& \sqrt{1- \left( \sqrt{1-Z_{C_{P,Q},\frac{s_2}{\sqrt{V(P)}}}(t_2D_{P,Q})} - \frac{\epsilon}{2}\right)} + \sqrt{\epsilon'}\\
&=& \sqrt{1- \left( \sqrt{1-Z_{C_{P,Q},\frac{s_2}{\sqrt{V(P)}}}(t_2D_{P,Q})} - \epsilon\right)}.
\Label{helineq}
\end{eqnarray}
Thus,
we obtain
\begin{eqnarray}
\limsup_{n\to\infty} 
F(W_{f'_{n}}(P'_n), Q^{ \frac{H(P)}{H(Q)}n+s_2\sqrt{n}\downarrow})
&\ge&
\sqrt{1-Z_{C_{P,Q},\frac{s_2}{\sqrt{V(P)}}}(t_2D_{P,Q})} - \epsilon.
\Label{sep4.5}
\end{eqnarray}
From Lemma \ref{transition'}, 
we obtain 
\begin{eqnarray}
\liminf_{n\to\infty} F(W_{f'_n}\circ W_{f_n}(P^{n\downarrow}), Q^{\frac{H(P)}{H(Q)}n + t_2\sqrt{n}})
\ge \liminf_{n\to\infty} F(W_{f'_{n}}(P'_n), Q^{ \frac{H(P)}{H(Q)}n+s_2\sqrt{n}\downarrow}).
\Label{sep5}
\end{eqnarray}
Then, (\ref{sepopt}) of Lemma \ref{separation} is derived from (\ref{sep4.5}) and (\ref{sep5}).


\vspace{0.5em}
\noindent[{\it Detailed Proof of Lemma \ref{separation}}]~
From the sketch of proof of Lemma \ref{separation},
it is enough to show (\ref{sep1})-(\ref{sep4}).
In this proof,
considering appropriate one-to-one maps,
we identify $\{0,1\}^N$, $\calX^N$, $\calY^N$, $P^N$ and $Q^N$ with $\{1,2,3,\ldots, 2^N\}$, $\{1,2,3,\ldots, |\calX|^N\}$, $\{1,2,3,\ldots, |\calY|^N\}$, $P^{N\downarrow}$ and $Q^{N\downarrow}$, respectively.

\vspace{0.5em}
\subsubsection{Proof of (\ref{sep1})}

First, we show (\ref{sep1}).
Let $\gamma>0$ satisfy
\begin{eqnarray}
\sqrt{1-\Phi\left(\frac{s_2}{\sqrt{V(P)}}\right)}
- \sqrt{1-\Phi\left(\frac{s_2+\gamma}{\sqrt{V(P)}}\right)} 
\le \epsilon.
\end{eqnarray}
Then,
we define a sequence of  probability distributions $P'_{n}$ 
satisfying that
\begin{eqnarray}
P'_{n}(j)
&:=&P^{n\downarrow}(j) + P^{n\downarrow}(S_n^P(s_2+\gamma, \infty)) U_2^{H(P) n + s_2\sqrt{n}}(j)\\
&=&P^{n\downarrow}(j) + P^{n\downarrow}(S_n^P(s_2+\gamma, \infty)) 2^{-(H(P)n+s_2\sqrt{n})}
\Label{newprob-1}
\end{eqnarray}
for any $j\in S_n^P(s_2)$.
Here, there is no constraint for $P'_{n}(j)$ with $j\in\N\setminus S_n^P(s_2)$ 
as long as $P'_{n}$ is a probability distribution.
Then, we obtain the following inequality:
\begin{eqnarray}
&&\liminf_{n\to\infty} F(P'_{n}, U_2^{ H(P)n+s_2\sqrt{n}\downarrow})\nonumber\\
&\ge&
\sqrt{1-\Phi\left(\frac{s_2+\gamma}{\sqrt{V(P)}}\right)} 
\\
&=&\sqrt{1-\Phi\left(\frac{s_2}{\sqrt{V(P)}}\right)}  - \left(\sqrt{1-\Phi\left(\frac{s_2}{\sqrt{V(P)}}\right)}
- \sqrt{1-\Phi\left(\frac{s_2+\gamma}{\sqrt{V(P)}}\right)} 
\right)\\
&\ge& \sqrt{1-\Phi\left(\frac{s_2}{\sqrt{V(P)}}\right)}  - \epsilon\\
&=& \limsup_{n\to\infty}F^{\calD}(P^{n}\to U_2^{H(P)n + s_2\sqrt{n}}) - \epsilon,
\end{eqnarray}
where the first inequality and the last equality were derived in \cite{KH13-2}.
Thus we obtain (\ref{sep1}).
\endproof

\vspace{0.5em}
\subsubsection{Proof of (\ref{sep2})}
Next, we show (\ref{sep2}).
To do so,
we prepare the following lemma.
\begin{lem}\cite{KH13-2}\Label{Wlem}
Let $S_1$ and $S_2$ be subsets of  the set $\N$ of natural numbers.
Suppose that $B=\{B(i)\}_{i\in S_1}$ and $C=\{C(j)\}_{j\in S_2}$ are non-negative real numbers in decreasing order  and  
\begin{eqnarray*}
\sum_{i\in S_1}B(i)\le \sum_{i\in S_2}C(j).
\end{eqnarray*}
Then, there exists a map $f:S_1\to S_2$ such that 
\begin{eqnarray}
B(i) \le W_f(C)(i) + \max_{j\in S_2}C(j)
\Label{Wmap2}
\end{eqnarray}
for any $i\in S_1$
where $W_f(C)(i):=\sum_{j\in f^{-1}(i)}C(j)$.
\end{lem}
We note that
\begin{eqnarray}
&&P^{n\downarrow}(S_n^P(s_2) \cup S_n^P(s_2+\gamma, \infty))\\
&=& P^{n\downarrow}(S_n^P(s_2))
+ P^{n\downarrow}(S_n^P(s_2+\gamma, \infty)) \\
&=&P^{n\downarrow}(S_n^P(s_2)) +  P^{n\downarrow}(S_n^P(s_2+\gamma, \infty)) U_2^{H(P) n + s_2\sqrt{n}}(S_n^P(s_2) ) \\
&=& P'_{n}(S_n^P(s_2)).
\end{eqnarray}
Thus, from Lemma \ref{Wlem},
there exists a map $f_{n}$ 
such that $f_n$ is the identity map on $S_n^P(s_2)$ and satisfies
\begin{eqnarray}
f_n(S_n^P(s_2+\gamma, \infty))
&\subset& S_n^P(s_2),
\Label{W2-1}\\
P'_{n}(j)
&\le& 
W_{f_n}(P^{n\downarrow})(j) 
+ \displaystyle\max_{k\in S_n^P(s_2+\gamma, \infty)} P^{n\downarrow}(k)
\Label{W1-1}
\end{eqnarray}
for any $j\in S_n^P(s_2)$.
Since 
\begin{eqnarray}
\displaystyle\max_{k\in S_n^P(s_2+\gamma, \infty)} P^{n\downarrow}(k)
=
P^{n\downarrow}(\lceil 2^{H(P)n + (s_2+\gamma)\sqrt{n}} \rceil)
\le 2^{-(H(P)n + (s_2+\gamma)\sqrt{n})},
\end{eqnarray}
we have
\begin{eqnarray}
&& F(W_{f_{n}}(P^{n\downarrow}), P'_{n})\nonumber\\
&\ge&\sum_{j\in S_n^P(s_2)}
\sqrt{W_{f_{n}}(P^{n\downarrow})(j)} \sqrt{P'_{n}(j)}
\nonumber\\
&\ge&\sum_{j\in S_n^P(s_2)}
\sqrt{\max\{P'_{n}(j) - 2^{-(H(P)n + (s_2+\gamma)\sqrt{n})} ,0\}} \sqrt{P'_{n}(j)}
\nonumber\\
&\ge&
\sum_{j\in S_n^P(s_2)}
\left(P'_{n}(j) - \sqrt{2^{-(H(P)n + (s_2+\gamma)\sqrt{n})}} 
\sqrt{P'_{n}(j)}\right)\\
&=&
1-
\sum_{j\in S_n^P(s_2)}
\sqrt{2^{-(H(P)n + (s_2+\gamma)\sqrt{n})}} 
\sqrt{P'_{n}(j)}.
\label{H6-16-4'}
\end{eqnarray}

Using the Schwarz inequality,
the second term of (\ref{H6-16-4'}) can be evaluated as follows:
\begin{eqnarray}
&&\sum_{j\in S_n^P(s_2)}
\sqrt{2^{-(H(P)n + (s_2+\gamma)\sqrt{n})}} 
\sqrt{P'_{n}(j)}
\nonumber\\
&\le&
\sqrt{2^{-(H(P)n + (s_2+\gamma)\sqrt{n})}} 
\sqrt{|S_n^P(s_2)|}
\sqrt{
\sum_{j\in S_n^P(s_2)}
P'_{n}(j)}
\nonumber\\
&\le&
\sqrt{2^{-(H(P)n + (s_2+\gamma)\sqrt{n})}} 
\sqrt{2^{H(P)n + s_2\sqrt{n}}}
\nonumber\\
&\le & 
\sqrt{ 2^{-\epsilon\sqrt{n}} }
\nonumber\\
&\overset{n\to\infty}{\to}&0, 
\Label{H6-16-9'}
\end{eqnarray}
Thus, we obtain (\ref{sep2}) from (\ref{H6-16-4'}) and (\ref{H6-16-9'}).
\endproof

\vspace{0.5em}
\subsubsection{Proof of (\ref{sep3})}
Next, we show (\ref{sep3}).
By the definition,
it holds that 
\begin{eqnarray}
\sqrt{1-Z_{C_{P,Q},\frac{s_2}{\sqrt{V(P)}}}(t_2D_{P,Q})}
&=& 
\sup_{A\in{\cal A}_{\frac{s_2}{\sqrt{V(P)}}}}
{\cal F}\left(\frac{dA}{dx}, \phi_{P,Q,t_2} \right),
\end{eqnarray}
where $\phi_{P,Q,t_2} := \phi_{t_2D_{P,Q},C_{P,Q}}$.
Thus, to obtain (\ref{sep3}),
it is enough to show the following inequality for an arbitrary $A\in{\cal A}_{\frac{s_2}{\sqrt{V(P)}}}$:
\begin{eqnarray}
\liminf_{n\to\infty} F(Q'_{n}, Q^{ \frac{H(P)}{H(Q)}n+s_2\sqrt{n}\downarrow})
&\ge&
{\cal F}\left(\frac{dA}{dx}, \phi_{t_2D_{P,Q},C_{P,Q}}\right) - \frac{\epsilon}{2}.
\Label{FFineq}
\end{eqnarray}
First, we prepare some notations.
We arbitrarily fix $A\in{\cal A}_{\frac{s_2}{\sqrt{V(P)}}}$
and define a function $y_{P,A}:\R\to\R$ as 
\begin{eqnarray}
y_{P,A}(x)
&:=&\sqrt{V(P)}\Phi^{-1}\left(A\left(\frac{x}{\sqrt{V(P)}}\right)\right).
\Label{y}
\end{eqnarray}
Let $0<\gamma<s_2$ satisfy
\begin{eqnarray}
\int^{\frac{s_2}{\sqrt{V(P)}}}_{\frac{s_2-\gamma}{\sqrt{V(P)}}}\sqrt{\frac{dA}{dx}(x)}\sqrt{\phi_{P,Q,b}(x)}dx
&\le & \frac{\epsilon}{6}.
\Label{bishou0'}
\end{eqnarray} 
In addition, let $\lambda>0$ satisfy
\begin{eqnarray}
\int_{-\infty}^{-\lambda}\sqrt{\frac{dA}{dx}(x)}\sqrt{\phi_{P,Q,b}(x)}dx
&\le & \frac{\epsilon}{6}.
\Label{bishou0}
\end{eqnarray} 
Then, for arbitrary $I\in\N$, we set sequences of real numbers as 
\begin{eqnarray}
x_{i}^I
:= \sqrt{V(P)}\left(-\lambda + \frac{\frac{s_2-\gamma}{\sqrt{V(P)}}+\lambda}{I} i\right)
~{\rm  and}~y_{i}^I:=y_{P,A}(x_{i}^I),
\Label{new}
\end{eqnarray}
where $0\le i\le I$.
Here we introduce a probability distribution $Q'_{n,I}$.
For any $j\in S_n^P(x_{0}^I, x_{I}^I) = \cup_{i=1}^{I} S_n^P(x_{i-1}^I, x_{i}^I)$, 
we note that there uniquely exists $i$ such that $j\in S_n^P(x_{i-1}^I, x_{i}^I)$.
Then we define $Q'_{n,I}$ as 
\begin{eqnarray}
Q'_{n,I}(j)
=\frac{P^{n\downarrow}(S_n^P(y_{i+1}^I, y_{i+2}^I))}{Q^{\frac{H(P)}{H(Q)}n+s_2\sqrt{n}\downarrow}( S_n^P(x_{i-1}^I, x_{i}^I))}Q^{\frac{H(P)}{H(Q)}n+s_2\sqrt{n}\downarrow}(j)
\Label{newprob}
\end{eqnarray}
for $1\le i\le I-2$ and $j\in S_n^P(x_{i-1}^I, x_{i}^I)$.
Here, there is no constraint for $Q'_{n,I}(j)$ with $j\in\N\setminus S_n^P(x_{0}^I, x_{I-2}^I)$ as long as $Q'_{n,I}$ is a probability distribution.
Using the definition (\ref{newprob}) of $ Q'_{n,I}(j)$, 
we have 
\begin{eqnarray}
&&\liminf_{n\to\infty} F(Q'_{n,I}, Q^{ \frac{H(P)}{H(Q)}n+s_2\sqrt{n}\downarrow})\nonumber\\
&\ge&\liminf_{n\to\infty}
\sum_{i=1}^{I-2}\sum_{j\in S_n^P(x_{i-1}^I, x_{i}^I)}
\sqrt{Q'_{n,I}(j)}\sqrt{Q^{ \frac{H(P)}{H(Q)}n+s_2\sqrt{n}\downarrow}(j)}
\nonumber \\
&=&\liminf_{n\to\infty}
\sum_{i=1}^{I-2}\sum_{j\in S_n^P(x_{i-1}^I, x_{i}^I)}
\sqrt{
\frac{P^{n\downarrow}(S_n^P(y_{i+1}^I, y_{i+2}^I))}{Q^{ \frac{H(P)}{H(Q)}n+s_2\sqrt{n}\downarrow}( S_n^P(x_{i-1}^I, x_{i}^I))}}
Q^{ \frac{H(P)}{H(Q)}n+s_2\sqrt{n}\downarrow}(j) \nonumber \\
&=&\liminf_{n\to\infty}
\sum_{i=1}^{I-2}
\sqrt{
\frac{P^{n\downarrow}(S_n^P(y_{i+1}^I, y_{i+2}^I))}
{Q^{ \frac{H(P)}{H(Q)}n+s_2\sqrt{n}\downarrow}( S_n^P(x_{i-1}^I, x_{i}^I))}}
Q^{ \frac{H(P)}{H(Q)}n+s_2\sqrt{n}\downarrow}( S_n^P(x_{i-1}^I, x_{i}^I))\\
&=&\liminf_{n\to\infty}
\sum_{i=1}^{I-2}\sqrt{P^{n\downarrow}(S_n^P(y_{i+1}^I, y_{i+2}^I))} 
\sqrt{Q^{ \frac{H(P)}{H(Q)}n+s_2\sqrt{n}\downarrow}( S_n^P(x_{i-1}^I, x_{i}^I))}\Label{H6-16-8}
\\
&=&\sum_{i=1}^{I-2}\sqrt{\Phi\left(\frac{y_{i+2}^I}{\sqrt{V(P)}}\right)-\Phi\left(\frac{y_{i+1}^I}{\sqrt{V(P)}}\right)} \sqrt{\Phi_{P,Q,b}\left(\frac{x_{i}^I}{\sqrt{V(P)}}\right)-\Phi_{P,Q,b}\left(\frac{x_{i-1}^I}{\sqrt{V(P)}}\right)},
\Label{lim.ineq}
\end{eqnarray}
where 
(\ref{lim.ineq}) follows from Lemma \ref{lem.central}.
Here, when we set as $\delta_I := \frac{\frac{s_2}{\sqrt{V(P)}}+\lambda}{I}$,
the right hand side of (\ref{lim.ineq}) is evaluated as follows:
\begin{eqnarray}
&=&\sum_{i=1}^{I-2}\sqrt{A\left(-\lambda+\delta_I(i+2)\right) - A\left(-\lambda+\delta_I(i+1)\right)} \nonumber\\
&&~~~~\times
\sqrt{\Phi_{P,Q,b}\left(-\lambda+\delta_Ii\right)-\Phi_{P,Q,b}\left(-\lambda+\delta_I(i-1)\right)}
\nonumber\\
&=&
\sum_{i=1}^{I-2}\sqrt{\int_{-\lambda+\delta_Ii}^{-\lambda+\delta_I(i+1)}
\frac{dA}{dx}\left(x+\delta_I\right)dx} 
\sqrt{\int_{-\lambda+\delta_Ii}^{-\lambda+\delta_I(i+1)}
\phi_{P,Q,b}\left(x-\delta_I\right)dx}\nonumber \\
&\ge&\sum_{i=1}^{I-2}
\int_{-\lambda+\delta_Ii}^{-\lambda+\delta_I(i+1)}
\sqrt{\frac{dA}{dx}\left(x+\delta_I\right)}
\sqrt{\phi_{P,Q,b}\left(x-\delta_I\right)}
dx\Label{Schwarz}
\\
%
&=&
\int_{-\lambda+\delta_I}^{-\lambda+\delta_I(I-1)}
\sqrt{\frac{dA}{dx}\left(x+\delta_I\right)}
\sqrt{\phi_{P,Q,b}\left(x-\delta_I\right)}
dx\Label{lowerdelta} 
\end{eqnarray}
From (\ref{lim.ineq}) and (\ref{lowerdelta}), 
we have 
\begin{eqnarray}
\lim_{I\to\infty}\liminf_{n\to\infty} F(Q'_{n,I}, Q^{ \frac{H(P)}{H(Q)}n+s_2\sqrt{n}\downarrow})
&\ge&\int_{-\lambda}^{\frac{s_2-\gamma}{\sqrt{V(P)}}}\sqrt{\frac{dA}{dx}(x)}\sqrt{\phi_{P,Q,b}(x)}dx.
\end{eqnarray}
Thus, when $I\in\N$ is large enough,
we have 
\begin{eqnarray}
\liminf_{n\to\infty} F(Q'_{n,I}, Q^{ \frac{H(P)}{H(Q)}n+s_2\sqrt{n}\downarrow})
&\ge&\int_{-\lambda}^{\frac{s_2-\gamma}{\sqrt{V(P)}}}\sqrt{\frac{dA}{dx}(x)}\sqrt{\phi_{P,Q,b}(x)}dx - \frac{\epsilon}{6}.
\Label{limitI}
\end{eqnarray}
Moreover,
\begin{eqnarray}
&&
\int_{-\lambda}^{\frac{s_2-\gamma}{\sqrt{V(P)}}}\sqrt{\frac{dA}{dx}(x)}\sqrt{\phi_{P,Q,b}(x)}dx\nonumber\\
&=&
\int_{-\infty}^{\infty}\sqrt{\frac{dA}{dx}(x)}\sqrt{\phi_{P,Q,b}(x)}dx
- \int^{\frac{s_2}{\sqrt{V(P)}}}_{\frac{s_2-\gamma}{\sqrt{V(P)}}}\sqrt{\frac{dA}{dx}(x)}\sqrt{\phi_{P,Q,b}(x)}dx
- \int^{-\lambda}_{-\infty}\sqrt{\frac{dA}{dx}(x)}\sqrt{\phi_{P,Q,b}(x)}dx
\nonumber\\
&\ge& {\cal F}\left(\frac{dA}{dx}, \phi_{P, Q, b}\right) - \frac{2\epsilon}{6},
\Label{K6-16-9}
\end{eqnarray}
where 
(\ref{K6-16-9})  follows from (\ref{bishou0'}) and (\ref{bishou0}).
Thus, we obtain (\ref{FFineq}) from (\ref{limitI}) and (\ref{K6-16-9}).
\endproof

\vspace{0.5em}
\subsubsection{Proof of (\ref{sep4})}
Next, we show (\ref{sep4}).
Let $I\in\N$, $\lambda>0$ and $\gamma>0$.
We set $I_{r}\in\N$  for $r\in\R$ as 
\begin{eqnarray}
I_{r}
:= \left\lceil \frac{A^{-1}(\Phi(\frac{r}{\sqrt{V(P)}}))+\lambda}{\frac{r}{\sqrt{V(P)}}+\lambda} I \right\rceil.
\end{eqnarray}
For $I_{s_2+\gamma} + 1 \le i\le I$,
we set sequences of real numbers as 
\begin{eqnarray}
\tilde{y}^I_i
&:=&\left\lceil \left(1 - \frac{P^{n\downarrow}(S^P_n( y^I_i, \infty))}{P^{n\downarrow}(S^P_n(s_2+\gamma, \infty))}\right)2^{H(P)n + s_2\sqrt{n}}\right\rceil.
\end{eqnarray}
We note that the following holds by the definition of $Q'_{n,I}$:
\begin{eqnarray}
Q'_{n,I}(S_n^P(x_{i-1}^I,x_{i}^I))
= P^{n\downarrow}(S_n^P(y_{i+1}^I, y_{i+2}^I))
\Label{assum4}
\end{eqnarray}
for $1\le i\le I-2$.
Then, it holds that
\begin{eqnarray}
P'_n(S_n^P(y_{i+1}^I, y_{i+2}^I))
\ge P^{n\downarrow}(S_n^P(y_{i+1}^I, y_{i+2}^I))
=Q'_{n,I}(S_n^P(x_{i-1}^I, x_{i}^I))
\Label{ass1}
\end{eqnarray}
for $1\le i\le I_{s_2}-2$ 
and 
\begin{eqnarray}
P'_n(\{ \tilde{y}^I_{i+1}+1, \ldots, \tilde{y}^I_{i+2} \})
\ge P^{n\downarrow}(S_n^P(y_{i+1}^I, y_{i+2}^I))
= Q'_{n,I}(S_n^P(x_{i-1}^I, x_{i}^I))
\Label{ass2}
\end{eqnarray}
for $I_{s_2+\gamma}+1\le i\le I-2$.
%
Thus, from Lemma \ref{Wlem},
we can choose a map $f'_{n,I}:\N\to\N$ such that 
\begin{eqnarray}
f'_{n,I}(S_n^P(y_{i+1}^I, y_{i+2}^I))
&\subset& S_n^P(x_{i-1}^I, x_{i}^I),
\Label{W2}\\
Q'_{n,I}(j)
&\le& W_{f'_{n,I}}(P'_n)(j) 
+\max_{k\in S_n^P(y_{i+1}^I, y_{i+2}^I)}P'_n(k)
\Label{W1}
\end{eqnarray}
for any $1\le i\le I_{s_2}-2$ and $j\in S_n^P(x_{i-1}^I, x_{i}^I)$,
and 
\begin{eqnarray}
f'_{n,I}(\{ \tilde{y}^I_{i+1}+1, \ldots, \tilde{y}^I_{i+2} \})
&\subset& S_n^P(x_{i-1}^I, x_{i}^I),
\Label{W2'}\\
Q'_{n,I}(j)
&\le& W_{f'_{n,I}}(P'_n)(j) 
+\max_{k\in \{ \tilde{y}^I_i+1, \ldots, \tilde{y}^I_{i+1} \}}P'_n(k)
\Label{W1'}
\end{eqnarray}
for any $I_{s_2+\gamma}+1\le i\le I-2$ and $j\in S_n^P(x_{i-1}^I, x_{i}^I)$.
For $j\notin S^P_n(x^I_0, x^I_{I_{s_2}-2}) \cup \{ \tilde{y}^I_{I_{s_2+\gamma}+1}, \ldots, \tilde{y}^I_{I}\}$,
there is no constraint for $f_{n,I}(j)$. 
%
%
Then,
we have
\begin{eqnarray}
F(W_{f'_{n,I}}(P'_n), Q'_{n,I})
&\ge&
\sum_{i=1}^{I_{s_2}-2}\sum_{j\in S_n^P(x_{i-1}^I, x_{i}^I)}
\sqrt{W_{f'_{n,I}}(P'_n)(j)} \sqrt{Q'_{n,I}(j)} \nonumber\\
&&+\sum_{i=I_{s_2+\gamma}+1}^{I-2}\sum_{j\in \{ \tilde{y}^I_{i+1}+1, \ldots, \tilde{y}^I_{i+2} \}}
\sqrt{W_{f'_{n,I}}(P'_n)(j)} \sqrt{Q'_{n,I}(j)}
\Label{inequ0}
\end{eqnarray}
In the following,
we show 
\begin{eqnarray}
\sum_{i=1}^{I_{s_2}-2}\sum_{j\in S_n^P(x_{i-1}^I, x_{i}^I)}
\sqrt{W_{f'_{n,I}}(P'_n)(j)} \sqrt{Q'_{n,I}(j)}
\ge \Phi\left(\frac{s_2}{\sqrt{V(P)}} \right) -\frac{\epsilon'}{2}
\Label{inequ1}
\end{eqnarray}
and 
\begin{eqnarray}
\sum_{i=I_{s_2+\gamma}+1}^{I-2}\sum_{j\in S_n^P(x_{i-1}^I, x_{i}^I)}
\sqrt{W_{f'_{n,I}}(P'_n)(j)} \sqrt{Q'_{n,I}(j)}
\ge 1 -\Phi\left(\frac{s_2}{\sqrt{V(P)}} \right) -\frac{\epsilon'}{2}.
\Label{inequ2}
\end{eqnarray}
Then,
we obtain (\ref{sep4}) from (\ref{inequ0}), (\ref{inequ1}) and (\ref{inequ2}).

First, we show (\ref{inequ1}).
Here,
note that
\begin{eqnarray}
\max_{k\in S_n^P(y_{i+1}^I, y_{i+2}^I)}P'_n(k)
&\le& P^{n\downarrow}(\lceil 2^{H(P)n + y_{i+1}^I\sqrt{n}} \rceil) + P^{n\downarrow}(S_n^P(s_2+\gamma, \infty)) 2^{-(H(P)n+s_2\sqrt{n})}\\
&\le& 2^{-(H(P)n + y_{i+1}^I\sqrt{n})}  + P^{n\downarrow}(S_n^P(s_2+\gamma, \infty)) 2^{-(H(P)n+s_2\sqrt{n})}\\
&\le& 2^{-(H(P)n + x_{i+1}^I\sqrt{n})}  + P^{n\downarrow}(S_n^P(s_2+\gamma, \infty)) 2^{-(H(P)n+s_2\sqrt{n})}
\Label{alphaineq}
\end{eqnarray}
where we used $x^I_i \le y^I_i$ since $A\ge\Phi$.
Combining (\ref{W1}) with (\ref{alphaineq}),
we have
\begin{eqnarray}
&& \sum_{i=1}^{I_{s_2}-2}\sum_{j\in S_n^P(x_{i-1}^I, x_{i}^I)}
\sqrt{W_{f'_{n,I}}(P'_n)(j)} \sqrt{Q'_{n,I}(j)}\nonumber\\
&\ge&\sum_{i=1}^{I_{s_2}-2}\sum_{j\in S_n^P(x_{i-1}^I, x_{i}^I)}
\sqrt{\max\{Q'_{n,I}(j) - (2^{-(H(P)n + x_{i+1}^I\sqrt{n})} + P^{n\downarrow}(S_n^P(s_2+\gamma, \infty)) 2^{-(H(P)n+s_2\sqrt{n})}),0\}} \sqrt{Q'_{n,I}(j)}
\nonumber\\
&\ge&
\sum_{i=1}^{I_{s_2}-2}\sum_{j\in S_n^P(x_{i-1}^I, x_{i}^I)}
Q'_{n,I}(j)
-
\sum_{i=1}^{I_{s_2}-2}\sum_{j\in S_n^P(x_{i-1}^I, x_{i}^I)}
\sqrt{2^{-(H(P)n + x_{i+1}^I\sqrt{n})}} 
\sqrt{Q'_{n,I}(j)}
\Label{H6-16-4}\\
&&-
\sum_{i=1}^{I_{s_2}-2}\sum_{j\in S_n^P(x_{i-1}^I, x_{i}^I)}
\sqrt{P^{n\downarrow}(S_n^P(s_2+\gamma, \infty)) 2^{-(H(P)n+s_2\sqrt{n})}} 
\sqrt{Q'_{n,I}(j)},
\nonumber
\end{eqnarray}
where (\ref{H6-16-4}) follows from (\ref{W1}) and  the last inequality follows from $\sqrt{x-y} \ge \sqrt{x}- \sqrt{y}$ for any $x\ge y \ge 0$.
Then, the first term of (\ref{H6-16-4}) satisfies the following:
\begin{eqnarray}
\sum_{i=1}^{I_{s_2}-2}\sum_{j\in S_n^P(x_{i-1}^I, x_{i}^I)}
Q'_{n,I}(j)
&=&
\sum_{i=1}^{I_{s_2}-2}\sum_{j\in S_n^P(x_{i-1}^I, x_{i}^I)}
Q'_{n,I}(j)
\nonumber\\
&=&
\sum_{i=1}^{I_{s_2}-2}
\sum_{k \in S_n^P(y_{i+1}^I, y_{i+2}^I)} P^{n\downarrow}(k)\\
&=&
\sum_{k \in S_n^P(y_{2}^I, y_{I_{s_2}}^I)} P^{n\downarrow}(k).
\end{eqnarray}
Then,
\begin{eqnarray}
\lim_{n\to\infty}
\sum_{i=1}^{I_{s_2}-2}\sum_{j\in S_n^P(x_{i-1}^I, x_{i}^I)}
Q'_{n,I}(j)
&=&
\lim_{n\to\infty}\sum_{k \in S_n^P(y_{2}^I, y_{I_{s_2}}^I)} P^{n\downarrow}(k)\\
&=&
\Phi\left(\frac{y_{I_{s_2}}^I}{\sqrt{V(P)}}\right) - \Phi\left(\frac{y_{2}^I}{\sqrt{V(P)}}\right)\\
&=&
\Phi\left(\frac{s_2-\gamma}{\sqrt{V(P)}}\right) -\Phi\left(-\lambda + 2\frac{\frac{s_2-\gamma}{\sqrt{V(P)}}+\lambda}{I} \right).
\end{eqnarray}
Here, for small $\gamma>0$,
we have
\begin{eqnarray}
\Phi\left(\frac{s_2-\gamma}{\sqrt{V(P)}}\right)
\ge \Phi\left(\frac{s_2}{\sqrt{V(P)}}\right) - \frac{\epsilon'}{4}.
\end{eqnarray}
In addition,
for large $I\in\N$ and large $\lambda>0$,
it holds that 
\begin{eqnarray}
\Phi\left(-\lambda + 2\frac{\frac{s_2-\gamma}{\sqrt{V(P)}}+\lambda}{I} \right)
\le \frac{\epsilon'}{4}.
\end{eqnarray}
Then,
we have 
\begin{eqnarray}
\lim_{n\to\infty}
\sum_{i=1}^{I_{s_2}-2}\sum_{j\in S_n^P(x_{i-1}^I, x_{i}^I)}
Q'_{n,I}(j)
&\ge&\Phi\left(\frac{s_2}{\sqrt{V(P)}}\right) - \frac{\epsilon'}{2}.
\Label{1-1}
\end{eqnarray}
The second term of (\ref{H6-16-4}) can be evaluated as follows using the Schwarz inequality:
\begin{eqnarray}
&&\sum_{i=1}^{I_{s_2}-2}\sum_{j\in S_n^P(x_{i-1}^I, x_{i}^I)}
\sqrt{2^{-(H(P)n + x_{i+1}^I\sqrt{n})}} 
\sqrt{Q'_{n,I}(j)}
\nonumber\\
&\le&
\sum_{i=1}^{I_{s_2}-2}
\sqrt{2^{-(H(P)n + x_{i+1}^I\sqrt{n})}} 
\sqrt{|S_n^P(x_{i-1}^I, x_{i}^I)|}
\sqrt{
\sum_{j\in S_n^P(x_{i-1}^I, x_{i}^I)}
Q'_{n,I}(j)}
\nonumber\\
&\le&
\sum_{i=1}^{I_{s_2}-2}
\sqrt{2^{-(H(P)n + x_{i+1}^I\sqrt{n})}} 
\sqrt{|S_n^P(x_{i}^I)|}
\nonumber\\
&\le & 
\sum_{i=1}^{I_{s_2}-2}
\sqrt{ 2^{-(x_{i+1}^I - x_{i}^I)\sqrt{n}} }\\
&=& 
\sum_{i=1}^{I_{s_2}-2}
\sqrt{ 2^{-\frac{s_2-\gamma+\sqrt{V(P)}\lambda}{I}\sqrt{n}} }
\Label{H6-16-7}\\
&\overset{n\to\infty}{\to}&0.
\Label{1-2}
\end{eqnarray}
The thrid term of (\ref{H6-16-4}) can be evaluated as follows:
\begin{eqnarray}
&&\sum_{i=1}^{I_{s_2}-2}
\sum_{j\in S_n^P(x_{i-1}^I, x_{i}^I)}
\sqrt{P^{n\downarrow}(S_n^P(s_2+\gamma, \infty)) 2^{-(H(P)n+s_2\sqrt{n})}} 
\sqrt{Q'_{n,I}(j)}\\
&\le& 
\sum_{i=1}^{I_{s_2}-2}
\sqrt{P^{n\downarrow}(S_n^P(s_2+\gamma, \infty)) 2^{-(H(P)n+s_2\sqrt{n})}} 
\sqrt{|S_n^P(x_{i-1}^I, x_{i}^I)|}
\sqrt{\sum_{j\in S_n^P(x_{i-1}^I, x_{i}^I)} Q'_{n,I}(j)}\\
&\le&
\sum_{i=1}^{I_{s_2}-2}
\sqrt{ 2^{-(H(P)n+s_2\sqrt{n})}} 
\sqrt{|S_n^P(x_{I}^I)|}
\\
&=&
(I_{s_2}-2)
\sqrt{ 2^{-(H(P)n+s_2\sqrt{n})}} 
\sqrt{|S_n^P(s_2 - \gamma)|}\\
&=&(I_{s_2}-2)
\sqrt{ 2^{-\gamma\sqrt{n}}} 
\nonumber\\
&\overset{n\to\infty}{\to}&0. 
\Label{1-3}
\end{eqnarray}
Thus, we obtain (\ref{inequ1}) from (\ref{1-1}), (\ref{1-2}) and (\ref{1-3}).

Next, we show (\ref{inequ2}).
Here,
note that
\begin{eqnarray}
\max_{k\in \{ \tilde{y}^I_i+1, \ldots, \tilde{y}^I_{i+1} \}}P'_n(k)
&\le& P^{n\downarrow}(\tilde{y}^I_i+1) + P^{n\downarrow}(S_n^P(s_2+\gamma, \infty)) 2^{-(H(P)n+s_2\sqrt{n})}\\
&\le& {1}/{\tilde{y}^I_i}+ P^{n\downarrow}(S_n^P(s_2+\gamma, \infty)) 2^{-(H(P)n+s_2\sqrt{n})}.
\Label{alphaineq'}
\end{eqnarray}
Combining (\ref{W1'}) with (\ref{alphaineq'}),
we have
\begin{eqnarray}
&&\sum_{i=I_{s_2+\gamma}+1}^{I-2}\sum_{j\in S_n^P(x_{i-1}^I, x_{i}^I)}
\sqrt{W_{f'_{n,I}}(P'_n)(j)} \sqrt{Q'_{n,I}(j)}\nonumber\\
&\ge&\sum_{i=I_{s_2+\gamma}+1}^{I-2}\sum_{j\in S_n^P(x_{i-1}^I, x_{i}^I)}
\sqrt{\max\{Q'_{n,I}(j) - ( {1}/{\tilde{y}^I_i} + P^{n\downarrow}(S_n^P(s_2+\gamma, \infty)) 2^{-(H(P)n+s_2\sqrt{n})}),0\}} \sqrt{Q'_{n,I}(j)}
\nonumber\\
&\ge&
\sum_{i=I_{s_2+\gamma}+1}^{I-2}
\sum_{j\in S_n^P(x_{i-1}^I, x_{i}^I)}
Q'_{n,I}(j)
-
\sum_{i=I_{s_2+\gamma}+1}^{I-2}\sum_{j\in S_n^P(x_{i-1}^I, x_{i}^I)}
\sqrt{1/\tilde{y}^I_i} 
\sqrt{Q'_{n,I}(j)}
\Label{H6-16-4''}\\
&&-
\sum_{i=I_{s_2+\gamma}+1}^{I-2}\sum_{j\in S_n^P(x_{i-1}^I, x_{i}^I)}
\sqrt{P^{n\downarrow}(S_n^P(s_2+\gamma, \infty)) 2^{-(H(P)n+s_2\sqrt{n})}} 
\sqrt{Q'_{n,I}(j)},
\nonumber
\end{eqnarray}
where (\ref{H6-16-4''}) follows from (\ref{W1'}) and  the last inequality follows from $\sqrt{x-y} \ge \sqrt{x}- \sqrt{y}$ for any $x\ge y \ge 0$.

Then, the first term of (\ref{H6-16-4''}) satisfies the following:
\begin{eqnarray}
\sum_{i=I_{s_2+\gamma}+1}^{I-2}
\sum_{j\in S_n^P(x_{i-1}^I, x_{i}^I)}
Q'_{n,I}(j)
&=&
\sum_{i=I_{s_2+\gamma}+1}^{I-2}\sum_{j\in S_n^P(x_{i-1}^I, x_{i}^I)}
Q'_{n,I}(j)
\nonumber\\
&=&
\sum_{i=I_{s_2+\gamma}+1}^{I-2}
\sum_{k \in S_n^P(y_{i+1}^I, y_{i+2}^I)} P^{n\downarrow}(k)\\
&=&
\sum_{k \in S_n^P(y_{I_{s_2+\gamma}}^I, y_{I}^I)} P^{n\downarrow}(k).
\end{eqnarray}
Then,
for small $\gamma>0$,
\begin{eqnarray}
\lim_{n\to\infty}
\sum_{i=I_{s_2+\gamma}+1}^{I-2}
\sum_{j\in S_n^P(x_{i-1}^I, x_{i}^I)}
Q'_{n,I}(j)
&=&
\lim_{n\to\infty}
\sum_{k \in S_n^P(y_{I_{s_2+\gamma}}^I, y_{I}^I)} P^{n\downarrow}(k)\\
&=&
\Phi\left(\frac{y_{I}^I}{\sqrt{V(P)}}\right) - \Phi\left(\frac{y_{I_{s_2+\gamma}}^I}{\sqrt{V(P)}}\right)\\
&=&
1 -\Phi\left(\frac{s_2+\gamma}{\sqrt{V(P)}} \right)\\
&\ge&
1 -\Phi\left(\frac{s_2}{\sqrt{V(P)}} \right) -\frac{\epsilon'}{2}\Label{2-1}
\end{eqnarray}

The second term of (\ref{H6-16-4''}) can be evaluated as follows using the Schwarz inequality:
\begin{eqnarray}
&&\sum_{i=I_{s_2+\gamma}+1}^{I-2}\sum_{j\in S_n^P(x_{i-1}^I, x_{i}^I)}
\sqrt{1/\tilde{y}^I_i} 
\sqrt{Q'_{n,I}(j)}
\nonumber\\
&\le&
\sum_{i=I_{s_2+\gamma}+1}^{I-2}
\sqrt{1/\tilde{y}^I_i} 
\sqrt{|S_n^P(x_{i-1}^I, x_{i}^I)|}
\sqrt{
\sum_{j\in S_n^P(x_{i-1}^I, x_{i}^I)}
Q'_{n,I}(j)}
\nonumber\\
&\le&
\sum_{i=I_{s_2+\gamma}+1}^{I-2}
\sqrt{1/\tilde{y}^I_i} 
\sqrt{|S_n^P(x_{I}^I)|}
\nonumber\\
&\le & 
\sum_{i=I_{s_2+\gamma}+1}^{I-2}
\sqrt{ 2^{-\gamma\sqrt{n}}} \sqrt{ 1 - \frac{P^{n\downarrow}(S^P_n( y^I_{I_{s_2+\gamma}+1}, \infty))}{P^{n\downarrow}(S^P_n(s_2+\gamma, \infty))}}^{-1} 
\\
&=&
(I-I_{s_2+\gamma}-3)
\sqrt{ 2^{-\gamma\sqrt{n}}} 
\sqrt{ \frac{P^{n\downarrow}(S^P_n(s_2+\gamma, \infty))}{P^{n\downarrow}(S^P_n( s_2+\gamma, y^I_{I_{s_2+\gamma}+1}))}} 
\\
&\overset{n\to\infty}{\to}&0, 
\Label{2-2}
\end{eqnarray}
where we used the fact that $\lim_{n\to\infty}P^{n\downarrow}(S^P_n( s_2+\gamma, y^I_{I_{s_2+\gamma}+1}))>0$ from $s_2+\gamma< y^I_{I_{s_2+\gamma}+1}$ and Lemma \ref{lem.central}.
The thrid term of (\ref{H6-16-4''}) can be evaluated as follows:
\begin{eqnarray}
&&\sum_{i=I_{s_2+\gamma}+1}^{I-2}\sum_{j\in S_n^P(x_{i-1}^I, x_{i}^I)}
\sqrt{P^{n\downarrow}(S_n^P(s_2+\gamma, \infty)) 2^{-(H(P)n+s_2\sqrt{n})}} 
\sqrt{Q'_{n,I}(j)}\\
&\le& 
\sum_{i=I_{s_2+\gamma}+1}^{I-2}
\sqrt{P^{n\downarrow}(S_n^P(s_2+\gamma, \infty)) 2^{-(H(P)n+s_2\sqrt{n})}} 
\sqrt{|S_n^P(x_{i-1}^I, x_{i}^I)|}
\sqrt{\sum_{j\in S_n^P(x_{i-1}^I, x_{i}^I)} Q'_{n,I}(j)}\\
&\le&
\sum_{i=I_{s_2+\gamma}+1}^{I-2}
\sqrt{ 2^{-(H(P)n+s_2\sqrt{n})}} 
\sqrt{|S_n^P(x_{I}^I)|}
\\
&=&
(I-I_{s_2+\gamma}-3)
\sqrt{ 2^{-(H(P)n+s_2\sqrt{n})}} 
\sqrt{|S_n^P(s_2 - \gamma)|}\\
&=&(I-I_{s_2+\gamma}-3)
\sqrt{ 2^{-\gamma\sqrt{n}}} 
\nonumber\\
&\overset{n\to\infty}{\to}&0. 
\Label{2-3}
\end{eqnarray}
Thus, we obtain (\ref{inequ2}) from (\ref{2-1}), (\ref{2-2}) and (\ref{2-3}).
\endproof

\subsection{Proof of Converse Part of Lemma \ref{main}}\Label{main2.app}

To prove the converse part, we prepare some lemmas.
We abbreviate the normal distribution with specific parameters as
\begin{eqnarray*}
\Phi_{P,Q,b}&:=&\Phi_{bD_{P,Q}, C_{P,Q}},\nonumber\\
\phi_{P,Q,b}&:=&\frac{d\Phi_{P,Q,b}}{dx}.
\end{eqnarray*}
We set the subsets of $\N$ which depends on $x$ and $x'\in\R$ as 
\begin{eqnarray*}
S_n^P(x)&:=&\{1, 2, ..., \lceil 2^{H(P)n+x\sqrt{n}} \rceil\} \Label{S1}\\
S_n^P(x, x')&:=& S_n^P(x')\setminus S_n^P(x). \Label{S2}
\end{eqnarray*}
The following lemma is obtained in \cite{KH13-2}.
\begin{lem}\Label{lem.central}
When both $P$ and $Q$ are non-uniform distributions,
\begin{eqnarray*}
\displaystyle\lim_{n\to\infty} Q^{\frac{H(P)}{H(Q)}n+b\sqrt{n}\downarrow}(S_n^P(x))
&=&\Phi_{P,Q,b}\left(\frac{x}{\sqrt{V(P)}}\right).
\end{eqnarray*}
\end{lem}

In addition, we prepare the following lemma.
%
\begin{lem}\Label{lem.converse}
Suppose that real numbers $v\le v'$ satisfy the following condition {\rm ($\star$)}.

\hspace{-1em}{\rm ($\star$)} There exist $u$ and $u'$ which satisfy the following three conditions:
\begin{eqnarray}
&&\hspace{-1.5em}~{\rm (I)}u\le v\le v'\le u' ~and~ v'\le s_2 ,\label{I}\nonumber\\
&&\hspace{-1.5em}~{\rm (II)}\frac{\Phi(v)}{\Phi_{P, Q, t_2}(v)}=\frac{\phi(u)}{\phi_{P,Q,t_2}(u)} ~and\nonumber\\
&&\hspace{-1.5em}\hspace{2em}\frac{1-\Phi(v')}{\Phi_{P, Q, t_2}(s_2)-\Phi_{P, Q, t_2}(v')}=\frac{\phi(u')}{\phi_{P,Q,t_2}(u')},\label{II}\\
&&\hspace{-1.5em}~{\rm (III)}\frac{\phi(x)}{\phi_{P,Q,t_2}(x)}~is~monotonically~decreasing~on~(u,u').\label{III}\nonumber
\end{eqnarray}
Then the following inequality holds
\begin{eqnarray}
&&F^{\cal M}_{P,Q,s_2}(t_2)\nonumber\\
&&\le
\sqrt{\Phi(v)} \sqrt{\Phi_{P, Q, t_2}(v)}
+
\int_{v}^{v'}\sqrt{\phi(x)} \sqrt{\phi_{P, Q, b}(x)}dx \nonumber\\
&&~~~+\sqrt{1-\Phi(v')} \sqrt{\Phi_{P, Q, t_2}(s_2)-\Phi_{P, Q, t_2}(v')}.\Label{lem.con.ineq}
\end{eqnarray}
\end{lem}

\begin{proof}
Let $P'_n$ be a probability distribution on $S_n^P(x)$ defined in (\ref{S1}) such that $P'_n\succ P_n$.
%
When we set a sequence $\{x_i^I\}_{i=0}^{I}$ for $I\in\N$ as $x_i^I:=v+\frac{v'-v}{I}i$,
we have the following by the monotonicity of the fidelity \cite{NC00}: 
\begin{eqnarray}
&&\hspace{-1em}F(P'^{\downarrow}_n, Q^{\frac{H(P)}{H(Q)}n+t_2\sqrt{n}\downarrow})\nonumber\\
&\hspace{-1em}\le&\hspace{-1em}\sqrt{P'^{\downarrow}_n( S_n^P(x_{0}^I))} \sqrt{Q^{\frac{H(P)}{H(Q)}n+t_2\sqrt{n}\downarrow}( S_n^P(x_{0}^I))}\nonumber\\
&&\hspace{-1.2em}+\sum_{i=1}^{I}\sqrt{P'^{\downarrow}_n( S_n^P(x_{i-1}^I, x_{i}^I))} \sqrt{Q^{\frac{H(P)}{H(Q)}n+t_2\sqrt{n}\downarrow}( S_n^P(x_{i-1}^I, x_{i}^I))}\nonumber\\
&&\hspace{-1.2em}+\sqrt{P'^{\downarrow}_n(S_n^P(s_2))-P'^{\downarrow}_n(S_n^P(x_{I}^I))}\nonumber\\
&&\hspace{-0.2em}\times \sqrt{Q^{\frac{H(P)}{H(Q)}n+t_2\sqrt{n}\downarrow}(S_n^P(s_2))-Q^{\frac{H(P)}{H(Q)}n+t_2\sqrt{n}\downarrow}(S_n^P(x_{I}^I))}\nonumber\\
&&\hspace{-1.2em}+\sqrt{1-P'^{\downarrow}_n( S_n^P(s_2))} \sqrt{1-Q^{\frac{H(P)}{H(Q)}n+t_2\sqrt{n}\downarrow}(S_n^P(s_2))}.
\Label{H6-18-1}
\end{eqnarray}

Here, we denote the right-hand side of (\ref{H6-18-1}) by $R_I(n)$.
Then, we can choose a subsequence $\{n_l\}_l\subset\{n\}$ 
such that 
\begin{eqnarray*}
\lim_{l \to \infty} R_I(n_l)
={\limsup_{n\to\infty} } R_I(n)
\end{eqnarray*}
and the limits 
\begin{eqnarray*}
c_0
&:=&\displaystyle\lim_{l\to\infty} P'^{\downarrow}_{n_l}( S_{n_l}(x_{0}^I)),\\
c_i
&:=&\displaystyle\lim_{l\to\infty} P'^{\downarrow}_{n_l}( S_{n_l}(x_{i-1}^I, x_{i}^I)),\\
c_{I+1}
&:=&\displaystyle\lim_{l\to\infty} \{P'^{\downarrow}_{n_l}(S_{n_l}(s_2)) - P'^{\downarrow}_{n_l}( S_{n_l}(x_{I}^I))\}\\
&=&1-\displaystyle\lim_{l\to\infty} P'^{\downarrow}_{n_l}( S_{n_l}(x_{I}^I))\\
c_{I+2}
&:=&0
\end{eqnarray*}
exist for $i=1, \ldots, I$.
Hence, we obtain
\begin{eqnarray}
&&\limsup_{n\to\infty} F(P'^{\downarrow}_n, Q^{\downarrow}_n)\nonumber\\
&\le&{\limsup_{n\to\infty} }R_I(n)
=\lim_{l\to\infty} R_I(n_l)\nonumber\\
&=&\sqrt{c_0} \sqrt{\Phi_{P, Q, b}(x_{0})}\\
&&+\sum_{i=1}^{I}\sqrt{c_i} \sqrt{\Phi_{P, Q, b}(x_{i}^I)-\Phi_{P, Q, b}(x_{i-1}^I)}\nonumber\\
&&+\sqrt{c_{I+1}} \sqrt{\Phi_{P, Q, b}(s_2)-\Phi_{P, Q, b}(x_{I}^I)},\nonumber
\end{eqnarray} 
where we used Lamma \ref{lem.central} in the last equality.

When we set as 
\begin{eqnarray*}
a_0
&:=&\Phi(x_{0}^I),\\
a_i
&:=&\Phi(x_i^I)-\Phi(x_{i-1}^I),\\
a_{I+1}
&:=&1-\Phi(x_{I}^I),\\
a_{I+2}
&:=&0,\\
b_0
&:=&\Phi_{P, Q, b}(x_{0}),\\
b_i
&:=&\Phi_{P, Q, b}(x_{i}^I)-\Phi_{P, Q, b}(x_{i-1}^I),\\
b_{I+1}
&:=&\Phi_{P,Q,b}(s_2)-\Phi_{P,Q,b}(x_{I}^I),\\
b_{I+2}
&:=&1-\Phi_{P,Q,b}(s_2)
\end{eqnarray*}
for $1,...,I$,
those satisfy the assumptions of Lemma \ref{naiseki2} as follows.
First, $a_0/b_0=\phi(u)/\phi_{P,Q,t_2}(u)$ and $a_{I+1}/b_{I+1}=\phi(u')/\phi_{P,Q,t_2}(u')$ hold by the assumption (II).
Moreover, there exist $z_i\in[x_{i-1}^I,x_i^I]$ for $i=1,...,I$ such that 
 $a_i/b_i=\phi(z_i)/\phi_{P,Q,t_2}(z_i)$ for $i=1,...,I$ due to the mean value theorem. 
Then $z_i\in(u,u')$ holds because of the relation $v=x_{0}^I\le x_{i-1}^I\le z_i\le x_{i}^I \le x_I^I= v' $ and the assumption (I).
Since $\phi(x)/\phi_{P,Q,t_2}(x)$ is monotonically decreasing on $(u,u')$ by the assumption (III), we have $a_{i-1}/b_{i-1}\ge a_i/b_i$ for $i=1,...,I+1$.
Moreover, 
\begin{eqnarray}
\sum_{i=0}^k a_i
&=&\Phi(x_{k}^I)\nonumber\\
&=&\lim_{l\to\infty} P^{n_l\downarrow}(S_{n_l}^P(x_{k}^I))\nonumber\\
&\le& \lim_{l\to\infty} P'^{\downarrow}_{n_l}(S_{n_l}^P(x_{k}^I))\nonumber\\
&=&\sum_{i=0}^k c_i
\end{eqnarray} 
holds for $k=0,1,...,I$ since $P^n\prec P'_n$, and $\sum_{i=0}^{I+1} a_i
=1=\sum_{i=0}^{I+1} c_i$ holds.

From the above discussion, we can use Lemma \ref{naiseki2}.
Therefore, the following hold:
\begin{eqnarray}
&&{\limsup_{n\to\infty} }F(P'^{\downarrow}_n, Q^{\downarrow}_n)\nonumber\\
&\le&\sqrt{c_0} \sqrt{\Phi_{P, Q, b}(x_{0}^I)}\nonumber\\
&&+\sum_{i=1}^{I}\sqrt{c_i} \sqrt{\Phi_{P, Q, b}(x_{i}^I)-\Phi_{P, Q, b}(x_{i-1}^I)}\nonumber\\
&&+\sqrt{c_0} \sqrt{\Phi_{P, Q, b}(s_2)-\Phi_{P, Q, b}(x_{I}^I)}\nonumber\\
&\le&\sqrt{\Phi(v)}\sqrt{\Phi_{P, Q, b}(v)}\Label{nai3}\\
&&+\sum_{i=1}^{I}\sqrt{\Phi(x_i^I)-\Phi(x_{i-1}^I)} \nonumber\\
&&\hspace{3em}\times\sqrt{\Phi_{P, Q, b}(x_{i}^I)-\Phi_{P, Q, b}(x_{i-1}^I)}\nonumber\\
&&+\sqrt{1-\Phi(v')} \sqrt{\Phi_{P, Q, b}(s_2)-\Phi_{P, Q, b}(v')}\nonumber
\end{eqnarray}
where we used $x_{0}^I=v$ and $x_{I}^I=v'$. 
Since
\begin{eqnarray*}
&&\lim_{I\to \infty} \sum_{i=1}^{I}\sqrt{\Phi(x_i^I)-\Phi(x_{i-1}^I)}\\
&&\hspace{3em} \times \sqrt{\Phi_{P, Q, b}(x_{i}^I)-\Phi_{P, Q, b}(x_{i-1}^I)}\\
&=&\lim_{I\to \infty} \sum_{i=1}^{I}\sqrt{\frac{\Phi(x_i^I)-\Phi(x_{i-1}^I)}{x_{i}^I-x_{i-1}^I}}\\
&&\hspace{3em} \times \sqrt{\frac{\Phi_{P, Q, b}(x_{i}^I)-\Phi_{P, Q, b}(x_{i-1}^I)}{x_{i}^I-x_{i-1}^I}} 
(x_{i}^I-x_{i-1}^I)\\
&=&\int_{v}^{v'}\sqrt{\phi(x)}\sqrt{\phi_{P, Q, b}(x)}dx,
\end{eqnarray*}
we obtain
\begin{eqnarray}
&&\limsup_{n\to\infty} F(P'^{\downarrow}_n, Q^{\downarrow}_n)\nonumber\\
&\le&\sqrt{\Phi(v)}\sqrt{\Phi_{P, Q, b}(v)}
+\int_{v}^{v'}\sqrt{\phi(x)}\sqrt{\phi_{P, Q, b}(x)}dx\nonumber\\
&&+\sqrt{1-\Phi(v')} \sqrt{\Phi_{P, Q, b}(s_2)-\Phi_{P, Q, b}(v')}.\nonumber 
\end{eqnarray}
\end{proof}


We treat the case when $v<1$.
Here, we use Lemma \ref{lem.converse}.
For any $v\in\R$, the existence of $u$ such that $u\le v$ and
\begin{eqnarray}
\frac{\Phi(v)}{\Phi_{P, Q, t_2}(v)}=\frac{\phi(u)}{\phi_{P,Q,t_2}(u)}
\end{eqnarray}
can be easily verified by the mean value theorem.
Moreover, when we take as $u'=v'=\beta:=\beta_{t_2D_{P,Q},C_{P,Q},\frac{s_2}{\sqrt{V(P)}}}$,
then $\beta\le s_2$ and
\begin{eqnarray}
\frac{1-\Phi(\beta)}{\Phi_{P, Q, t_2}(s_2)-\Phi_{P, Q, t_2}(\beta)}
=\frac{\phi(\beta)}{\phi_{P,Q,t_2}(\beta)}
\end{eqnarray}
hold by Lemma \ref{sol2}.
From Lemma \ref{monotone},
$\frac{\phi(u)}{\phi_{P,Q,t_2}(u)}$ is monotonically decreasing on $(-\infty,\frac{bH(Q)}{1-C_{P,Q}})$.
Since $\beta\le \frac{bH(Q)}{1-C_{P,Q}}$,
thus (III) holds.
Taking the limit $v\to-\infty$ in (\ref{lem.con.ineq}),
we have the following inequality
\begin{eqnarray*}
&&\hspace{-2em}F^{\cal M}_{P,Q,s_2}(t_2)\nonumber\\
&\hspace{-4em}\le&\hspace{-2em}
\int_{-\infty}^{\beta}\sqrt{\phi(x)} \sqrt{\phi_{P, Q, b}(x)}dx \nonumber\\
&&\hspace{-2em}+\sqrt{1-\Phi(\beta)} \sqrt{\Phi_{P, Q, t_2}(s_2)-\Phi_{P, Q, t_2}(\beta)}\\
&\hspace{-4em}=&\hspace{-2em}I_{P,Q,t_2}(\beta)\nonumber\\
&&\hspace{-2em}+\sqrt{1-\Phi(\beta)} \sqrt{\Phi_{P, Q, t_2}(s_2)-\Phi_{P, Q, t_2}(\beta)}
\end{eqnarray*}
and thus, the proof is completed.


Then, we treat the case when $v=1$
First, we treat the case when $t_2\le0$.
Since it holds that 
\begin{eqnarray}
F^{\cal M}(P\to Q|N)
\le \sqrt{\sum_{i=1}^{2^N} Q^{\downarrow}(i)}
= \sqrt{ Q^{\downarrow}(\{1,...,2^N\})},
\Label{jimei}
\end{eqnarray}
for an arbitrary $N\in\N$,
we have
\begin{eqnarray}
F^{\cal M}_{P,Q,s_2}(t_2)
&\le& \liminf_{n\to\infty} \sqrt{Q^{\frac{H(P)}{H(Q)}n+t_2\sqrt{n}\downarrow}(S_n^P(s_2))}\nonumber\\
&=&\sqrt{\Phi_{P,Q,t_2}(s_2)},\nonumber
\end{eqnarray}
where we used Lemma \ref{lem.central} in the last equality.
Next, we treat the case when $t_2>0$.
Here, we use Lemma \ref{lem.converse}.
For any $v\in\R$, the existence of $u$ such that $u\le v$ and
\begin{eqnarray}
\frac{\Phi(v)}{\Phi_{P, Q, t_2}(v)}=\frac{\phi(u)}{\phi_{P,Q,t_2}(u)}
\end{eqnarray}
can be easily verified by the mean value theorem.
Moreover, when we take as $u'=v'=\beta$,
then $\beta\le s_2$ and
\begin{eqnarray}
\frac{1-\Phi(\beta)}{\Phi_{P, Q, t_2}(s_2)-\Phi_{P, Q, t_2}(\beta)}
=\frac{\phi(\beta)}{\phi_{P,Q,t_2}(\beta)}
\end{eqnarray}
hold by Lemma \ref{sol0}.
From Lemma \ref{monotone},
$\frac{\phi(u)}{\phi_{P,Q,t_2}(u)}$ is monotonically decreasing on $\R$,
 and thus (III) holds for any $u$ and $u'$.
Taking the limit $v\to-\infty$ in (\ref{lem.con.ineq}),
we have the following inequality
\begin{eqnarray}
&&\hspace{-2em}F^{\cal M}_{P,Q,s_2}(t_2)\nonumber\\
&&\hspace{-2em}\le
\int_{-\infty}^{\beta}\sqrt{\phi(x)} \sqrt{\phi_{P, Q, b}(x)}dx \nonumber\\
&&+\sqrt{1-\Phi(\beta)} \sqrt{\Phi_{P, Q, t_2}(s_2)-\Phi_{P, Q, t_2}(\beta)}.
\end{eqnarray}
Since
\begin{eqnarray}
&&\int_{-\infty}^{\beta}\sqrt{\phi(x)} \sqrt{\phi_{P, Q, b}(x)}dx \nonumber\\
&=&\Phi\left(\beta-\frac{D_{P,Q}t_2}{2}\right)
e^{-\frac{(D_{P,Q}t_2)^2}{8}},
\end{eqnarray}
the proof is completed.

Then, we treat the case when $v>1$.
At first, we treat the case when $s_2\le \Phi_{P, Q, t_2}^{-1}\left(\frac{\Phi_{P, Q, t_2}(\alpha)}{\Phi_{P}(\alpha)}\right)$, where $\alpha:=\alpha_{t_2D_{P,Q},C_{P,Q}}$.
For an arbitrary sequence $\{P'_n\}_{n=1}^{\infty}$ of probability distributions which satisfies $P'_n\succ P^{n}_{2^{H(P)n+s_2\sqrt{n}}}$, 
 the monotonicity of the fidelity follows
\begin{eqnarray}
F(P'_n, Q_n)
&\le&\sqrt{P'_n( S_n^P(s_2))} \sqrt{Q_n( S_n^P(s_2))}
\\&&+\sqrt{P'_n( S_n^P(s_2, \infty))} \sqrt{Q_n( S_n^P(s_2, \infty))}.\nonumber
\end{eqnarray}
Since
\begin{eqnarray}
&{\displaystyle\limsup_{n\to\infty} }P'_n( S_n^P(s_2, \infty))=0, &
\end{eqnarray}
we obtain
\begin{eqnarray}
&\displaystyle\limsup_{n\to\infty}F(P'_n, Q_n)
\le\sqrt{\Phi_{P, Q, t_2}(s_2)}.&
\end{eqnarray}
Next, we treat the case when $s_2> \Phi_{P, Q, t_2}^{-1}\left(\frac{\Phi_{P, Q, t_2}(\alpha)}{\Phi_{P}(\alpha)}\right)$.
Here, we use Lemma \ref{lem.converse}.
By Lemma \ref{sol1}, $\alpha$ satisfies
\begin{eqnarray}
\frac{\Phi(\alpha)}{\Phi_{P, Q, t_2}(\alpha)}=\frac{\phi(\alpha)}{\phi_{P,Q,t_2}(\alpha)},
\end{eqnarray}
and $\beta$ satisfies
\begin{eqnarray}
\frac{1-\Phi(\beta)}{\Phi_{P, Q, t_2}(s_2)-\Phi_{P, Q, t_2}(\beta)}
=\frac{\phi(\beta)}{\phi_{P,Q,t_2}(\beta)}.
\end{eqnarray}
When we take as $u=u'=\alpha$ and $v=v'=\beta$ in Lemma \ref{lem.converse},
those satisfy (I) and (II).
Moreover, from Lemma \ref{monotone},
$\frac{\phi(u)}{\phi_{P,Q,t_2}(u)}$ is monotonically decreasing on $(\frac{bH(Q)}{1-C_{P,Q}},\infty)$.
Since $\frac{bH(Q)}{1-C_{P,Q}} \le \alpha \le \beta$,
(III) holds.
Thus, we have the following inequality
\begin{eqnarray*}
&&F^{\cal M}_{P,Q,s_2}(t_2)\nonumber\\
&\le&\sqrt{\Phi_{P}(\alpha) \Phi_{P, Q, t_2}(\alpha)}
+(I_{P, Q, t_2}(\beta)-I_{P, Q, t_2}(\alpha))\\
&&+\sqrt{1-\Phi_{P}(\beta)}\sqrt{\Phi_{P, Q, t_2}(s_2)-\Phi_{P, Q, t_2}(\beta)},
\end{eqnarray*}
and thus, the proof is completed.
\endproof

\subsection{Proof of Theorem \ref{region2-uni1}}\Label{region2-uni1.app}

The function $F_{U_l,Q,s_2}$ in (\ref{2nd-dil}) is obviously continuous and strictly monotonically decreasing on $F_{U_l,Q,s_2}^{-1}((0,1))$.

We first prove the direct part  of  (\ref{feq}).
Let $s_2\ge0$.
Since the size of storage is greater than the size of support of $U_l^n$,
$U_l^n$ can be converted to $U_l^n$ itself in storage.
Thus, we have  
\begin{eqnarray}
F_{U_l,Q,s_2}^{\cal D}(t_2)
&\ge& \lim_{n\to\infty}F^{\cal D}(U_l^{n} \to Q^{\frac{\log l}{H(Q)}n+t_2 \sqrt{n}})\nonumber\\
&=& F_{U_l,Q,s_2}(t_2), \Label{14-1}
\end{eqnarray}
where the equality follows from Lemma \ref{lem.uni}.
Next, let $s_2<0$.
We have
\begin{eqnarray}
\liminf_{n\to \infty} F^{\cal D}(U_l^n\to U_2^{(\log l)n + s_2\sqrt{n}}) 
= 1.
\end{eqnarray}
Thus, using Lemma \ref{transition},
\begin{eqnarray}
F_{U_l,Q,s_2}^{\cal D}(t_2)
&\ge& \lim_{n\to\infty}F^{\cal D}(U_2^{(\log l)n + s_2\sqrt{n}} \to Q^{\frac{\log l}{H(Q)}n+t_2 \sqrt{n}})\nonumber\\
&=& F_{U_l,Q,s_2}(t_2). \Label{14-1'}
\end{eqnarray}

Then, we prove the converse part  of  (\ref{feq}).
Let $s_2\ge0$.
Then, the following inequality obviously holds:
\begin{eqnarray}
F_{U_l,Q,s_2}^{\cal M}(t_2)
&\le& \lim_{n\to\infty}F^{\cal M}(U_l^{n} \to Q^{\frac{\log l}{H(Q)}n+t_2 \sqrt{n}})\nonumber\\
&=& F_{U_l,Q,s_2}(t_2). \Label{14-2}
\end{eqnarray}
Next, let $s_2<0$.
Since an arbitrary probability distribution on $S_n^P(s_2)$ defined in (\ref{S1}) can be converted from the uniform distribution with size of ${(\log l)n + s_2\sqrt{n}}$ bits by majorization conversion.
Thus, we have
\begin{eqnarray}
F_{U_l,Q,s_2}^{\cal M}(t_2)
&\le& \lim_{n\to\infty}F^{\cal M}(U_2^{(\log l)n + s_2\sqrt{n}} \to Q^{\frac{\log l}{H(Q)}n+t_2 \sqrt{n}})\nonumber\\
&=& F_{U_l,Q,s_2}(t_2). \Label{14-2'}
\end{eqnarray}
From (\ref{fidelity-ineq}), (\ref{14-1}), (\ref{14-1'}), (\ref{14-2}) and (\ref{14-2'}),
we obtain (\ref{feq}).
\endproof

\subsection{Proof of Theorem \ref{region2-uni2}}\Label{region2-uni2.app}

The function $F_{P,U_l,s_2}$ in (\ref{2nd-con}) is obviously continuous and strictly monotonically decreasing on $F_{P,U_l,s_2}^{-1}((0,1))$.

We first prove the direct part of  (\ref{feq}).
Let $(\log l)t_2 \le s_2$.
Since the size of storage is greater than the size of support of $U_l^{\frac{H(P)}{\log l}n + t_2\sqrt{n}}$,
we 
have
\begin{eqnarray}
F_{P,U_l,s_2}^{\cal D}(t_2)
&=& \lim_{n\to\infty}F^{\cal D}(P^n \to U_l^{\frac{H(P)}{\log l}n + t_2\sqrt{n}})\nonumber\\
&=& \lim_{n\to\infty}F^{\cal D}(P^n \to U_2^{H(P)n + (\log l)t_2\sqrt{n}})\nonumber\\
&=& F_{P,U_l,s_2}(t_2).
\Label{15-1}
\end{eqnarray}
When $(\log l)t_2 > s_2$,
the direct part is obvious. 

Next, we prove the converse part of  (\ref{feq}).
Let $(\log l)t_2 \le s_2$.
Then, the following inequality holds:
\begin{eqnarray}
F_{P,U_l,s_2}^{\cal M}(t_2)
&\le& \lim_{n\to\infty}F^{\cal M}(P^n \to U_l^{\frac{H(P)}{\log l}n + t_2\sqrt{n}})\nonumber\\
&=& \lim_{n\to\infty}F^{\cal D}(P^n \to U_2^{H(P)n + (\log l)t_2\sqrt{n}})\nonumber\\
&=& F_{P,U_l,s_2}(t_2).
\Label{15-2}
\end{eqnarray}
Let $(\log l)t_2 > s_2$.
Since an arbitrary probability distribution on $S_n^P(s_2)$ can be converted from the uniform distribution with size of ${H(P)n + s_2\sqrt{n}}$ bits by majorization conversion.
Thus, we have
\begin{eqnarray}
&&F_{P,U_l,s_2}^{\cal M}(t_2)\nonumber\\
&\le& \lim_{n\to\infty}F^{\cal M}(U_2^{H(P)n + s_2\sqrt{n}} \to U_l^{\frac{H(P)}{\log l}n + t_2\sqrt{n}})\nonumber\\
&=& \lim_{n\to\infty}F^{\cal M}(U_2^{H(P)n + s_2\sqrt{n}} \to U_2^{H(P)n + (\log l)t_2\sqrt{n}})\nonumber\\
&=& 0.
\Label{15-2'}
\end{eqnarray}
From (\ref{fidelity-ineq}), (\ref{15-1}), (\ref{15-2}) and (\ref{15-2'}),
we obtain (\ref{feq}).
\endproof

\vspace{0.5em}

\subsection{Proof of Lemma \ref{q-to-c}}\Label{q-to-c.app}

Let $\psi_{M}$ be a pure state on $\C^M\otimes\C^M$ with the suquared Schmidt coefficient ${\cal C}_M(P_{\psi})$ defined in (\ref{PL}).
Then, according to Lemma \ref{opt trans}, 
an arbitrary pure state on $\C^M\otimes\C^M$ which can be converted from $\psi$ by LOCC can also be converted from $\psi$ via $\psi_{M}$ by LOCC.
Thus, if we convert $\psi$ to $\psi_{M}$ in the first step, 
the minimal error is attainable in the second step. 
Here, $\psi_{M}$ was given when the optimal entanglement concentration was performed for $\psi$ and does not depend on $\phi$.
Therefore, it is optimal to perform the entanglement concentration as LOCC in the first step and especially the optimal operation does not depend on $\phi$.

\begin{lem}\Label{universal}
Let $\psi$ be a pure state on a bipartite system $\mathcal{H}_{AB}$.
Then, there exists a LOCC map $\Gamma:\mathcal{S}(\mathcal{H}_{AB})\to\mathcal{S}(\C^M\otimes\C^M)$ which satisfies the following conditions:
\begin{description}
 \item[(I)]$\Gamma(\psi)=\psi_{M}$, 
 \item[(II)]For any LOCC map $\Gamma':\mathcal{S}(\mathcal{H}_{AB})\to\mathcal{S}(\C^M\otimes\C^M)$, 
there exists a LOCC map $\tilde{\Gamma}:\mathcal{S}(\C^M\otimes\C^M)\to\mathcal{S}(\C^M\otimes\C^M)$ such that $\Gamma'(\psi)=\tilde{\Gamma}(\psi_{M})$.
\end{description}
\end{lem}

\begin{proof}
Because of Nielsen's theorem \cite{Nie99}, 
there exists a LOCC map $\Gamma$ which satisfies (I).
Next, we prove that such $\Gamma$ satisfies (II).
Let a LOCC map $\Gamma':\mathcal{S}(\mathcal{H}_{AB})\to\mathcal{S}(\C^M\otimes\C^M)$ 
output a state $\eta_j$ with probability $q_j$.
Then, because of Jonathan-Plenio's theorem \cite{JP99},
\begin{eqnarray}
\sum_{i=1}^l P_{\psi}^{\downarrow}(i)
\le \sum_{i=1}^l  \sum_{j}q_j P_{\eta_j}^{\downarrow}(i)
\end{eqnarray}
holds for any $l=1,...,M$.
Since ${\cal C}_M(P_{\psi})(i)=P_{\psi}^{\downarrow}(i)$ for $l=1,...,J_{P_{\psi},M}-1$ where $J_{P_{\psi},M}$ was defined in (\ref{J}),
we have
\begin{eqnarray}\Label{ineq}
\sum_{i=1}^l {\cal C}_M(P_{\psi})(i)
\le \sum_{i=1}^l  \sum_{j}q_j P_{\eta_j}^{\downarrow}(i)
\end{eqnarray}
for any $l=1,...,J_{P_{\psi},M}-1$.
Moreover, 
(\ref{ineq}) holds for any $l=J_{P_{\psi},M},...,M$.
If it does not holds,
it is a contradiction as follows.
Then, there are the minimum numbers $k_0, l_0\in\{J_{P_{\psi},M},...,M\}$ such that 
\begin{eqnarray}
\sum_{i=1}^{k_0} {\cal C}_M(P_{\psi})(i)
&>& \sum_{i=1}^{k_0}  \sum_{j}q_j P_{\eta_j}^{\downarrow}(i),\Label{ineq2}
\\
\frac{\sum_{i=J_{P_{\psi}, M}}^{|\mathcal{X}|} P_{\psi}^{\downarrow}(i)}{M+1-J_{P_{\psi}, M}}
&>&\sum_{j} q_j P_{\eta_j}^{\downarrow}(l_0).\Label{ineq3}
\end{eqnarray}
and $k_0\ge l_0$.
Moreover, the inequality (\ref{ineq3}) holds for any $l\ge l_0$ because $\sum_{j} q_j P_{\eta_j}^{\downarrow}(l)$ is monotonically decreasing with respect to $l$.
Thus, we have the following contradiction.
\begin{eqnarray}
1
&=&
\sum_{i=1}^{k_0} {\cal C}_M(P_{\psi})(i)
+\sum_{i=k_0+1}^{M} {\cal C}_M(P_{\psi})(i)
\\
&>&
\sum_{i=1}^{k_0}  \sum_{j}q_j P_{\eta_j}^{\downarrow}(i)
+\sum_{i=k_0+1}^{M}  \sum_{j}q_j P_{\eta_j}^{\downarrow}(i)\\
&=&1.\Label{contradiction}
\end{eqnarray}
As proved above, (\ref{ineq}) holds for any $l=1,...,M$, and thus, we obtain (II) because of Jonathan-Plenio's theorem \cite{JP99}.
\end{proof}

From Lemma \ref{universal} with $M=2^N$, we have
\begin{eqnarray*}
F^{\cal Q}(\psi\to \ph|N)
&=&F^{\cal Q}(\psi_{2^N}\to\ph)\\
&=&F^{\cal M}({\cal C}_{2^N}(P_{\psi}) \to P_{\ph})\\
&=&F^{\cal M}(P_{\psi} \to P_{\ph}|N).
\end{eqnarray*}
Thus, the proof is completed.
\endproof

\section{Conclusion}
\Label{sec:Conclusion}

We have considered random number conversion (RNC) via random number storage with restricted size.
In particular, we derived the rate regions between the storage size and the conversion rate of RNC from the viewpoint of the first- and second-order asymptotics.
In the first-order rate region,
it was shown that there exists the trade-off when the rate of storage size is smaller than or equal to the entropy of the initial distribution as in Fig. \ref{1st}
and semi-admissible rate pairs characterize the trade-off.
When the conversion rate of RNC achieves a semi-admissible first-order rate pair,
the non-trivial second-order rate regions were obtained as in Figs. \ref{2nd}, \ref{fig-case2}, \ref{fig-case3}, \ref{2nd-reg2-uni1} and \ref{2nd-reg2-uni2}.
Especially, to derive the second-order rate region at the admissible first-order rate pair,
we introduced the generalized Rayleigh-normal distribution and investigate its basic properties.
From the second-order asymptotics,
we also obtained asymptotic expansion of maximum conversion number with high approximation accuracy.
%
Then, we applied the results for RNC via restricted storage to LOCC conversion via entanglement storage in quantum information theory. 
In the problem, we did not assume that an initial state and a target state are the same states,
However,
the LOCC conversion via storage can be regarded as compression process if the target state equals the initial state, and thus,
our problem setting is a kind of generalization of LOCC compression for pure states.

We gave some remarks on the admissibility of rate pairs. 
In the argument to characterization of the rate regions,
we defined the simple relations called ``dominate" and ``simulate" between two rate pairs,
and introduced the admissibility of rate pairs based on the relations in order to clarify essentially important rate pairs in the rate region.
We note that,
besides RNC via restricted storage, 
the notion of ``simulate'' was implicitly appeared in asymmetric information theoretic operations. 
For instance, Fig. 1 in \cite{CK78} represents the typical first-order rate region in the wiretap channel.
Then the left side boundary of the region is characterized as an interval between the origin and the other edge point,
and hence, the left side boundary is simulated by the edge point of the interval.
Besides of such an applicability of ``simulate'',
the notion of ``simulate'' has not been focused on, and thus, the admissibility
in the sense of this paper has not been recognized.
In particular, to our knowledge, 
it has not been appeared in the context of the second-order rate region in existing studies. 
Since the notion of ``simulate'' plays an important role in the characterization of the rate region, 
it will be widely used also in the rate region in the sense of the first- and second-order asymptotics.


We refer some future studies.
First, probability distributions or quantum states were assumed to be i.i.d. in this paper. 
To treat information sources with classical or quantum correlation, 
the extension from an i.i.d. sequence to general one is thought as a problem to be solved \cite{MPSVW10}. 
Second, we analyzed only the asymptotic performance of random number conversion and LOCC conversion.
On the other hand, what we can operate has only finite size.
Therefore, it is expected that conversion via restricted storage are analyzed in finite setting.
Third, since only pure states were treated in quantum information setting although mixed entangled states can be appear in practice, 
the extension from pure states to mixed states is thought to be important.
%
%
Finally, we have shown that the problem of
RNC via restricted storage has a non-trivial trade-off relation described by the second-order rate region
although trade-off relation in the first-order rate region is quite simple.
As is suggested by the results,
even when two kinds of first-order rates in an information theoretical problem simply and straightforward relate with each other, 
there is a possibility that 
the rate region has a non-trivial trade-off relation in the second order asymptotics.
We can conclude that consideration of the second order asymptotics
might bring a new trade-off relation in various information theoretical problems.

\section*{Acknowledgment}

WK was partially supported from Grant-in-Aid for JSPS Fellows No. 233283. 
MH is partially supported by a MEXT Grant-in-Aid for Scientific Research (A) No. 23246071 and the National Institute of Information and Communication
Technology (NICT), Japan.
The Centre for Quantum Technologies is funded by the Singapore Ministry of Education and the National Research Foundation
as part of the Research Centres of Excellence programme.


\end{document}